\pdfoutput=1 

\documentclass{elsarticle}

\usepackage[titletoc]{appendix}
\usepackage{etex,etoolbox}
\usepackage{lscape}
\usepackage{amsmath, amsthm, amssymb}
\usepackage{gamesem}
\usepackage{pst-tree}
\usepackage[inline]{enumitem}
\usepackage{pstring}
\usepackage{xspace}
\usepackage{todonotes}
\usepackage{amsmath,amssymb,amsthm}
\usepackage{bcprules}
\usepackage{algorithm}
\usepackage{algorithmic}
\usepackage{geometry}
\usepackage{pdflscape}
\usepackage{multirow}
\usepackage{subcaption}
\usepackage{tcolorbox}
\usepackage{wrapfig}
\usepackage{url}
\usetikzlibrary{arrows,automata}

\def\placeproofsatend{}

\def\includetodos{}

\newtcolorbox{ruletablebox}[1]{
    colframe=blue!75!black,
    title=#1
}

\newtcolorbox{remarkbox}{
    colframe=blue!75!black
}

\ifdefined\includetodos
\newtcolorbox{todobox}{
    colback=orange!5!white,
    colframe=orange!75!black,
    title=TODO
}
\else
\newcommand\todobox[1]{}
\fi

\makeatletter
\providecommand{\@fourthoffour}[4]{#4}
\def\fixstatement#1{%
    \AtEndEnvironment{#1}{%
    \xdef\pat@label{\expandafter\expandafter\expandafter
        \@fourthoffour\csname#1\endcsname\space\@currentlabel}}}

\globtoksblk\prooftoks{1000}
\newcounter{proofcount}

\ifdefined\placeproofsatend
\long\def\proofatend#1\endproofatend{%
    \edef\next{\noexpand\begin{proof}[Proof of \pat@label]}%
    \toks\numexpr\prooftoks+\value{proofcount}\relax=\expandafter{\next#1\end{proof}}
    \stepcounter{proofcount}}
\else
\def\proofatend{\begin{proof}}
\def\endproofatend{\end{proof}}
\fi

\def\printproofs{%
    \count@=\z@
    \loop
    \the\toks\numexpr\prooftoks+\count@\relax
    \ifnum\count@<\value{proofcount}%
    \advance\count@\@ne
    \repeat}
\makeatother

\theoremstyle{plain}
\newtheorem{theorem}{Theorem}[section]
\newtheorem{proposition}[theorem]{Proposition}
\newtheorem{lemma}[theorem]{Lemma}
\newtheorem{corollary}[theorem]{Corollary}
\newtheorem{property}[theorem]{Property}

\theoremstyle{definition}
\newtheorem{definition}{Definition}[section]
\newtheorem{conjecture}{Conjecture}[section]
\newtheorem{example}{Example}[section]

\theoremstyle{remark}
\newtheorem{remark}{Remark}[section]

\fixstatement{theorem}
\fixstatement{proposition}
\fixstatement{property}
\fixstatement{lemma}

\setlist[itemize]{itemindent=0.3em}
\setlist[description]{itemindent=0.2em}

\newcommand\VarOcc{\mathcal{O}{cc}}
\newcommand\VarSet{\mathcal{V}}
\newcommand\Nodes{\mathcal{N}}
\newcommand\NodesVar{\Nodes_{\sf var}}
\newcommand\NodesLmd{\Nodes_\lambda}
\newcommand\NodesApp{\Nodes_@}

\newcommand\ExtendedNodes{\tilde{\Nodes}}

\newcommand{\ghostlmd}{{\lambda\!\!\lambda}}
\newcommand{\ghostvar}{\theta}

\newcommand\ImNodesVar{\NodesVar^\ghostvar}
\newcommand\ImNodesLmd{\NodesLmd^\ghostlmd}

\newcommand{\normalizing}{{\sf norm}}
\newcommand{\branching}{{\sf branch}}

\newcommand{\travsetbr}{{\travset^\branching}}
\newcommand{\travsetnorm}{\travset^\normalizing}

\newcommand{\travulc}{\travset}
\newcommand{\travstlc}{{\travset_{\sf STLC}}}

\newcommand{\child}{{\small \sf Ch}} 

\newcommand{\rulefont}[1]{\mathbf{\sf #1}}

\def\coresymbol{\pi} 
\newcommand{\core}[1]{\coresymbol(#1)} 

\newcommand{\enables}{\vdash} 
\newcommand{\ctree}{CT} 

\newcommand{\ExtNodes}{\Nodes^{\sf ext}}

\newcommand{\IntNodes}{\Nodes^{\sf int}}

\newcommand\pathset{{\mathcal{P}aths}} 
\newcommand\arth{\textsf{arth}} 

\renewcommand\ie{{\it i.e.\@\xspace}}
\renewcommand\eg{{\it e.g.\@\xspace}}

\newcommand\bulklambda[1]{{\lambda\overline{#1}}}

\newcommand{\etal}{\textit{et al}. }

\newcommand{\hlred}{\rightarrow_{hl}}

\newcommand{\llred}{\rightarrow_{ll}}

\newcommand{\hoc}{\emph{hoc}}

\newcommand{\alphaequiv}{\equiv}

\def\justseqset{\mathcal{J}}
\def\istraversal{\models}

\author{William Blum}
\ead{william.blum@microsoft.com}
\address{Microsoft Research, One Microsoft Way, Redmond, WA, 98052, USA}

\title{Reducing Lambda Terms with Traversals (preprint)
\\ \vspace*{1em}
\today}

\begin{document}

\begin{abstract}
    We introduce a method to evaluate untyped lambda terms by combining the theory of traversals, a term-tree traversing technique inspired from Game Semantics~\cite{Ong2006,BlumPhd}, with judicious use of the eta-conversion rule of the lambda calculus.

    The traversal theory of the simply-typed lambda calculus relies on the eta-long transform to ensure that when traversing an application, there is a subterm representing every possible operator's argument\cite{BlumPhd, Ong2006}. In the untyped setting, we instead exhibit the missing operand via ad-hoc instantiation of the eta-expansion rule, which allows the traversal to proceed as if the operand existed in the original term. This gives rise to a more generic concept of traversals for lambda terms. A notable improvement, in addition to handling untyped terms, is that no preliminary transformation is required: the original unaltered lambda term is traversed.


    We show that by bounding the non-determinism of the traversal rule for free variables, one can effectively compute a set of traversals characterizing the paths in the tree representation of the beta-normal form, when it exists. This yields an evaluation algorithm for untyped lambda-terms. We prove correctness by showing that traversals implement \emph{leftmost linear reduction}, a generalization of the \emph{head linear reduction} of Danos~\etal \cite{danos-head,danosherbelinregnier1996}.
\end{abstract}

\maketitle

\tableofcontents

\section{Background and overview}

Traversals were originally introduced in the context of higher-order recursion schemes~\cite{Ong2006} as tree generators.
They were then extended to the more general setting of simply-typed languages with higher-order free variables, such as the simply-typed lambda calculus (STLC)~\cite{BlumPhd}. In the typed setting, the traversal theory has a well understood relationship with Game Semantics: traversals are in bijection with the plays of the game denotation of a term in which all the internal moves are revealed. Further, the \emph{core projections} of traversals are in bijection with the \emph{standard} game denotation, in which internal moves are hidden (Theorem~\ref{thm:gamesem_correspondence_stlc}). This correspondence gives rise to a method for reducing beta-redexes of a simply-typed lambda term that does not involve the traditional capture-avoiding substitution~\cite{danos-head,BlumPhd,BlumGalop2008,Blum-LocalBeta2008}.

In this paper we (i) generalize game-semantic traversals of the simply-typed lambda calculus\cite{BlumPhd} to the untyped lambda calculus (ULC), which we call \emph{imaginery travserals}; (ii) infer a method to evaluate untyped lambda terms based on enumeration of their traversals; (iii) formalize the connection between traversals and the notion of \emph{head linear reduction}.

\emph{Key idea behind imaginary traversals}: In the absence of types, eta-long expansion is not feasible as it would yield an infinite term. Imaginary traversals overcome this problem by performing eta-expansion \emph{on-demand}.
Traversing an untyped lambda term using the original traversal rules can lead
the traversal to get stuck. This happens when the operator sub-term of an application has an insufficient number of operands. In such situation, there is no operand node in the term tree that the traversal can jump to.
In the term $(\lambda f.f)(\lambda y.y)$, for example, an STLC traversal would get stuck when trying to resolve the ``$y$ argument'' of the variable occurrence $f$.
In contrasts, an imaginary traversal just jumps to a fictitious fresh node $z$ that would exist if the term were eta-expanded as $\lambda z. (\lambda f.f)(\lambda y.y) z$. This trick allows the traversal to proceed.

For terms having a beta-normal form, normalization can then be implemented by traversal enumeration. Although the set of imaginary traversals can be infinite, it is sufficient to consider sufficiently enough traversals to cover all paths of the tree representation of the beta-normal form. We demonstrated how such set is effectively computable by restricting the non-determinism of the free-variable rule. In other words, only a finite number of on-demand eta-expansions is required to evaluate the term. We implemented the normalization procedure in the HOG software~\cite{Blum-HogTool}. The last section provides some examples of terms normalized with it.

Finally, to prove correctness of the evaluation procedure, we show that traversals essentially implement a recursive version of \emph{head linear reduction}, a non-standard reduction strategy based on an alternative notion of redex where individual variable occurrences are matched to their corresponding arguments and where substitution occurs one variable occurrence at a time.\footnote{This is reminiscent of the standard call-by-name strategy except that linear reduction can occur under lambda, and performs substitution one variable occurrence at a time.}

\subsection{Related work}

\subsubsection{Game-semantics Traversal of STLC}

In~\cite{BlumPhd} we extended the theory of traversals to the simply-typed lambda calculus by introducing the new traversal rule \rulenamet{IVar} to traverse higher-order free variables. We formalized the correspondence between the theory of traversals and Game semantics by establishing a bijection between traversals and the game denotation of a term. This technical proof relies on several concepts~\cite{BlumPhd}:
\begin{description}
  \item[Free variables] The introduction of the traversal rule \rulenamet{IVar} modeling  free variables of the lambda calculus.
  Such rule is not needed in the original presentation of traversals because higher-order recursion schemes are necessarily closed terms of ground type where terminals are all of order $1$ at most\cite{Ong2006}.
  \item[Interaction Game Semantics] Traversals do not immediately correspond to the standard game denotation of a term. Instead they correspond to a more verbose variant of the innocent game denotation that preserves all the internal moves played while composing the sub-strategy denoting its subterms; as opposed to the standard innocent game denotation which hides all internal moves;
  \item[Traversal core] The correspondence makes use of various operations on justified sequences and traversals. The \emph{projection with respect to the tree root}, also called  ``core of a traversal'', plays a key role in the proof.
\end{description}

 This correspondence with Game Semantics yields a method to reduce beta-redexes in eta-long simply-typed lambda-terms. This procedure was studied in~\cite{BlumPhd,BlumGalop2008,Blum-HogTool,Ong-NormByTrav2015} and implemented in the HOG tool\cite{BlumGalop2008, BlumPhd}. We recall these results in Section~\ref{sec:traversal_correspondence_stlc}.

 The results presented in the present paper can be viewed as a generalization of this work to the untyped case. Adapting the traversal theory to the untyped setting requires new ingredients detailed in the rest of the paper. In particular:
\begin{itemize}[nosep]
 \item Eta-expansion is performed `on-the-fly' at each point where the traditional STLC traversal would get stuck: when there is insufficiently many operands in an application to continue traversing the term.
 \item  We augment the rules traversing lambda and variable nodes to support ad-hoc eta-expansion.
 \item We introduce the concept of ``ghost'' nodes: those are fictitious tree nodes that progressively appear as the traversal eta-expands the sub-terms.
\item The term tree itself is not altered during eta-expansion. Ghost nodes only fictitiously appear in the traversal and are only visible in the context where they get introduced.
\item We constrain the non-determinism of the rule traversing free variables by a critical quantity called the \emph{arity threshold}, calculated in linear time from the traversal itself. This limit guarantees that the normalization procedure terminates when the beta-normal form exists.
\end{itemize}

Conveniently, when restricted to simply-typed terms, the two notion of traversal coincide (Proposition~\ref{prop:ulc_and_stlc_trav_coincide}).

The absence of $\eta$-long transform of the term is a notable difference with the traversals for recursion scheme~\cite{Ong2006} or the simply-typed lambda calculus~\cite{BlumPhd}. It simplifies the presentation as the traversed tree is a straight abstract syntax representation of the original, unmodified, lambda term. Similarly, the evaluation procedure from Algorithm~\ref{algo:ulc_normalization_by_traversals} produces the traditional beta-normal form rather than its beta-\emph{eta-long} normal form. Furthermore, we prove soundness by establishing a connection with head-linear reduction rather than appealing to Game Semantics.

\subsubsection{Head linear reduction}
Danos \etal established the first connection between Game Semantics and the concept of \emph{head linear reduction}\cite{danos-head,danosherbelinregnier1996}. Despite this connection, their work does not explicitly describe how to use head linear reduction to fully evaluate a lambda term. The \emph{pointer abstract machine} introduced in the same paper, for instance, yields the ``\emph{quasi-head normal form (qhn)}'' of the term, not the normal form itself. The resulting \emph{qhn} still needs to be first reduced to head normal form (using head reduction) and then normalized using standard normal reduction, which must terminate by the $\lambda$-calculus theory.

The traversals discussed in this paper, just like the traversals of Berezun-Jones \cite{JonesBerezunLLL}, essentially implement
a recursive variation of the head linear reduction strategy on the lambda term to find its normal form. A traversal starts with a depth-first search for what Danos-Reigner call the \emph{head occurrence} of the \emph{hoc redex}. Traversals implement linear substitution by ``jumping'' to the node in the tree that represents the term to be substituted for the head occurrence. After a jump, a traversal just continues its descent for the next head occurrence, as if the first head occurrence in the tree had been replaced by the redex argument. A property of head-linear reduction is that all the terms appearing in the reduction sequence consist of sub-terms of the original term. This property means that there is no need for the traditional notions of environments or closures to reduce beta-redexes. Such mechanisms are instead replaced by a system of pointers referencing subterms of the original term.

In this paper, we generalize head linear reduction to \emph{leftmost linear reduction}, a strategy that recursively applies head linear reduction; just like normal reduction from the lambda calculus literature which proceeds by recursive application of head reduction.
We show how the \emph{imaginary traversals} effectively implement leftmost linear reduction, which proves soundness of the traversal normalization procedure.

\subsubsection{Berezun-Jones traversals}
Berezun and Jones introduced the first notion of traversals for the \emph{untyped} lambda calculus (ULC)~\cite{JonesBerezunLLL}. Although they took inspiration from the Game Semantic traversals of~\cite{Ong2006,BlumGalop2008}, their definition is more operational in nature and does not directly relate to the game denotation of a term. Starting from an operational semantics of the untyped lambda calculus they derive a normalization method that, very much like the traversal of~\cite{Ong2006, BlumPhd}, proceeds by traversing some tree representation of the term.

The tree representation used in Berezun-Jones' traversals is a direct abstract syntax tree representation of the lambda-term, whereas in the present paper we use a more compressed form where consecutive lambdas and applications are merged into a single bulk node. Unlike the original definition of traversals for STLC, where eta-long transformation is performed prior to traversing a term, neither Berezun-Jones' traversals nor imaginary traversals require any prior syntactic transformation of the input term.

The Berezun-Jones traversals distinctively rely on two justification pointers: each `token' of the traversal can have one \emph{binding pointer} as well as one \emph{control pointer}. They also involve the use of a `flag` boolean parameter associated with every token of the traversal. In contrast to Berezun-Jones' traversals, imaginary traversals necessitate a single justification pointer per node occurrence.

Another notable difference is that the normalization algorithm of Berezun-Jones requires a single traversal of the term. Our normalization algorithm, on the other hand, produces one traversal for every branch in the term tree f  the beta-normal form. Each branching point corresponds to some occurrence of a variable in the normal form, and each branch corresponds to one operand of the variable. This non-deterministic branching may conceivably be eliminated by adding auxiliary pointers to allow imaginary traversals to backtrack: after the first operand of an application is evaluated, the traversal would return to the operator variable and start exploring the remaining operands.

\subsection{Rest of the paper}

In the remaining of this paper we introduce basic definition from the traversal theory and introduce the notion of \emph{imaginary traversals} for the untyped lambda calculus. We briefly discuss the correspondence with game models of the Untyped Lambda Calculus~\cite{KerThesis} and show how they yield an algorithm to evaluate untyped lambda terms. We illustrate the evaluation algorithm on some examples, and finally prove correctness by establishing a correspondence with head-linear reduction.

\section{Terms, trees and justified sequences}
\subsection{Notations}
\emph{Sequences:} For any alphabet $\Sigma$ we write $\Sigma^*$ to denote the set of all (possibly infinite) sequences over elements of $\Sigma$.
 We use the abbreviated syntax $\overline{x}$ for the sequence of elements $x_1 \ldots x_n$ for some $n\geq0$.

 \emph{Labelled trees:} Given a set of labels $\mathcal{L}$, we use expressions of the form `$l\langle t_1, \ldots, t_m \rangle$' for $m \geq 0$ to denote a $\mathcal{L}$-labelled tree with root labelled $l\in \mathcal{L}$. If $m=0$ then the tree is a leaf otherwise it's a node labelled $l$ with $m$ ordered children trees $t_1$, \ldots, $t_m$. We will also use the formal definition of a labelled-tree given by a partial function $\nat^* \rightarrow \mathcal{L}$ mapping directed tree-paths to labels in $\mathcal{L}$, where
 a directed path is given by the sequence of child indices to follow to reach a given node from the tree root (represented by the empty sequence).

\subsection{Untyped lambda calculus}
We consider the set of terms $\Lambda$ of the untyped lambda calculus constructed from the grammar $\Lambda := x\ |\ (\Lambda\ \Lambda)\ |\ \lambda x. \Lambda $
where $x$ ranges over a countable set of variable names $\VarSet$.
For conciseness we abbreviate consecutive lambda abstractions when writing lambda terms, so that $\lambda x_1 \ldots \lambda x_n . U$ for some $n\geq 0$ and term $U$ will be written $\lambda x_1 \ldots x_n . U$; or just
$\lambda \overline{x} . U$ if $\overline{x} = x_1 \ldots x_n$. We write $M \alphaequiv N$ to denote that $M$ and $N$ are equal up to variable renaming.

In the remaining of the document we use uppercase letters $M, N, T, U, V$ for lambda terms and lowercase $x,y,z$ to denote variables. Terms will all have variable names ranging in some fixed set $\VarSet$ of identifiers.

\subsection{Computation tree and enabling relation}

Given an untyped lambda term $M$ we define its \defname{computation tree} as the abstract syntax tree representation of the term where (i) consecutive lambda abstractions are merged into a single `bulk' node labelled by the list of bound variables and have a single child; (ii) consecutive applications are represented by a single node labelled $@$ whose first child represents the operator, and subsequent children represents the operands.
This definition is similar to that of STLC~\cite{Ong2006, BlumPhd} with the notable difference that we do not eta-long expand the term prior to constructing the computation tree.
\begin{definition}[Computation tree]
Let $M$ be an untyped lambda term with variable names in $\VarSet$.
\begin{itemize}
\item The set of labels
$\mathcal{L} = \{ @ \} \union \VarSet \union \{ \lambda x_1 \ldots x_n \ | \ x_1 ,\ldots, x_n \in
    \VarSet, n\in\nat \}$
\item
    The \defname{computation tree} $\ctree(M)$ is a $\mathcal{L}$-labelled tree defined inductively on the syntax of $M$:
    \begin{eqnarray*}
        \ctree(\lambda \overline{x} . z s_1 \ldots s_m) &=& \lambda \overline{x}\; \langle\; z\; \langle \ctree(s_1),\ldots,\ctree(s_m)\rangle\rangle\\
        && \hbox {where $m\geq 0$, $z \in \VarSet$}, \\
 \ctree(\lambda \overline{x} . (\lambda y.t) s_1 \ldots s_m) &=& \lambda \overline{x}\; \langle\; @ \; \langle \ctree(\lambda y.t),\ctree(s_1),\ldots,\ctree(s_m) \rangle \rangle \enspace\\
&&  \hbox{where $m\geq 1$, $y\in \VarSet$}.
    \end{eqnarray*}

\item We write $\Nodes(M)$ to denote the set of nodes of $\ctree(M)$, or just $\Nodes$ if the term is clear from context; $\NodesVar$ for the set of variable nodes; $\NodesLmd$ for lambda nodes, and $\NodesApp$ for  application nodes.

\item For any lambda node $\alpha\in\NodesLmd$ we write $\child(\alpha)$ to denote its unique child node (either an $@$ or variable node).
\end{itemize}
\end{definition}

Note that any lambda term can indeed be written in one of the two forms above. In particular, applicative terms are handled by the case $n=0$ of the form $\lambda . N$ for some term $N$ where `$\lambda$' is referred to as a ``dummy lambda''. This compact tree representation helps maintain alternation between lambda nodes (at odd level, counting from 1 onwards) and variable nodes (at even level).

We define the \defname{enabling relation} $\enables$ on tree nodes $\Nodes$ as the relation mapping (i) each lambda node to all the variable nodes that it binds (ii) the root of the tree to all the free variable nodes (iii) variable and application node to each one of their child node. We write $m \enables_k \alpha$ to indicate that a lambda node $\alpha$ is the $k$th child of a variable or application node $m$ for some $k\geq0$; and $\alpha \enables_k m$ to indicate that variable node $m$ is labelled by the $k$th variable bound by $\alpha$ for $k\geq1$.

We define \defname{hereditarily enabling} as the reflexive transitive closure of the enabling relation, and we say that node $m$ \defname{hereditarily enables} $n$ written $m \enables^* n$.

A tree node is \defname{external} if it is hereditarily enabled by the root, otherwise it is \defname{internal}. We write $\ExtNodes$ for the set of external nodes, that is the image of the tree root by the reflexive transitive closure of the enabling relation. Internal nodes are therefore precisely the nodes hereditarily justified by application nodes $@$. We write $\IntNodes$ to denote the set of internal nodes.

We define the \defname{arity} of a node $n$, written $|n|$ as follows: the arity of a lambda node $\lambda x_1 \cdots x_k$ for $k\geq 0$ is defined as $k$; the arity of a variable node $x$ is given by the number of its children in the computation tree; the arity of an $@$-node is the number of its children nodes minus 1 (\ie, the number of operands in the application).

\subsection{Ghost nodes}
The nodes in $\Nodes$ correspond to tokens in the syntax representation of the lambda term; we call them \defname{structural nodes}. We complement the computation tree with two additional infinite sets of nodes representing fictitious term tokens:
(i) the set $\ghostvar$ of \emph{ghost variable nodes}, (ii) the set $\ghostlmd$ of \emph{ghost lambda nodes}. Together, nodes in $\ghostlmd$ and $\ghostvar$ are called the \defname{ghost nodes}. We write $\ExtendedNodes$ for the extended set of nodes $\Nodes + \ghostvar + \ghostlmd$; $\ImNodesLmd$ as a shorthand for $\NodesLmd + \ghostlmd$ and $\ImNodesVar$ as a shorthand for $\NodesVar + \ghostvar$. Ghost lambda nodes and ghost variables all have the same unique label $\ghostlmd$ and $\ghostvar$ respectively. We will thus not distinguish between them and will just use the identifiers $\ghostvar$ and $\ghostlmd$ as placeholders for some ghost variable node or ghost lambda node respectively.

Formally, we simultaneously define ghost nodes and the extension of the enabling relation $\enables$ to ghost nodes as follows: (i) for every (ghost) variable or application node $m$ and for all $k>|m|$ there is a ghost lambda node $\ghostlmd$ such that $m \enables_k \ghostlmd$; (ii) for every (ghost) lambda node $\alpha$ and $k>|\alpha|$ there is a ghost variable node $\ghostvar$ such that $\alpha \enables_k \ghostvar$. Thus, variable and ghost lambda nodes are uniquely defined by their enabler node (possibly themselves ghost nodes) and associated label $k\geq 1$.
The sets $\ghostvar$ and $\ghostlmd$ are therefore uniquely determined by the pair $(\Nodes, \enables)$.

By convention ghost variables and lambda nodes are assigned arity $0$.

\subsection{Justified sequence of nodes}
\label{sec:justseq}

A \defname{justified sequence} is defined a sequence of (extended) nodes from the computation tree where every occurrence $n$ in the sequence, except the first one, has an associated link--the ``justification pointer''--pointing to some previous node occurrence $j$ in the sequence--its ``justifier''--with an associated ``link label'' $l\geq0$, such that the justifier node $\enables$-enables the source node with the corresponding label. That is: $j \enables_l n$. We represent justified sequences as a sequence of labelled node occurrences with back-pointing arrows representing justification pointers:
$\Pstr[10pt]{ s = \cdots (j){j} \ldots (n-j,25:l){n} }$.
For readability, we sometime just indicate the link label in exponent of the justified occurrence, or even omit the justification pointers altogether.
Note that the enabling relation is \emph{statically} induced by the structure of the tree whereas the justification relation is defined on node \emph{occurrences} in one specific justified sequence. The set $\justseqset(M)$ of justified sequences over $M$ consists of all justified sequences over the extended set of nodes $\ExtendedNodes$.

{\bf Name-free representation}: We define the \defname{structure} of a justified sequence as the sequence obtained by discarding node labels while preserving node types. It is encoded as a sequence of triples $\{\lambda, {\sf Var}, @ \}\times\nat\times\nat$ where the first component indicates the type of each occurrence (lambda bunch, variable, application); the second component indicates the link distance, and the last component gives the link label ($0$ for no pointer).
Two sequences $s$ and $u$ over possibly two distinct terms are \defname{equivalent}, written $s \equiv u$ just if they have the same structure: the types of the underlying nodes in the two sequences match pairwise (either two variable/@ nodes or two lambda nodes) and have the same justification pointers. Two equivalent sequences are not necessarily equal since they may have different variable names.
We say that two subsets $J_1\subseteq \justseqset(M_1)$ and $J_2\subseteq\justseqset(M_2)$ are isomorphic, written $J_1\cong J_2$, if  there exists a structure-preserving bijection between $J_1$ and $J_2$. That is there exists a bijection $\phi :J_1\longrightarrow J_2$ such that for any $j\in J$, $j\equiv\phi(j)$.

{\bf Quadruplet sequence}:
We define the set of \defname{quadruplet sequences} over variables in $\mathcal{V}$ as $O_4^*$ where $O_4 = \nat^*\times \mathcal{V}^* \times\nat\times\nat$. The first component represents a node path; the second is a sequence of variable names representing \emph{pending lambdas}; the third represents a link distance and the last is a link label.

The `pending lambdas' component represents a sequence of variable names meant to be appended to the bulk of variables already bound in a lambda node. Although not used when encoding raw justified sequences of $M$, it will be useful when defining transformation on justified sequences that involve modifications to the lambda abstractions.

\begin{definition}[Quadruplet encoding]
    \label{def:quadruplet_encoding}
    We encode a justified sequence of nodes over a term $M$
    as a sequence of quadruplets in $O_4^*$ obtained by applying the following mapping occurrence-wise:
    \begin{enumerate}[nosep]
        \item (Node) The first component is the path to an extended node of the computation tree that uniquely defines the type and label of the node occurrence.
        \item (Pending lambdas) The second component is set to the empty list $\epsilon$.
        \item (Link distance) The third component gives the distance (number of occurrences in the sequence) between the encoded occurrence and its justifier.
        \item (Link label) The fourth component gives the justification link label.
    \end{enumerate}
\end{definition}

We adopt the quadruplet encoding throughout the rest of the paper and thus assume $\justseqset(M)$ to be a (strict) subset of $O_4^*$.


 When typesetting a sequence from $\justseqset(M)$, we represent occurrences by their associated node labels, and we write the \emph{pending lambdas} in exponent. For example:
$t = \cdots @ \cdot x \cdot \lambda\overline{x}^{[y_1 \ldots y_n]}$
where $\lambda\overline{x}$ is the label of the tree node associated with that occurrence, and $y_1 \ldots y_n$ represent \emph{pending lambdas} for this occurrence. We omit the exponent if the list of pending lambdas is empty.

\begin{example}
\label{examp:ghost_materialization}
    Let $M = \lambda x. (\lambda y z.z) u$. The set of structural paths of the computation tree of $M$ is given by $\Sigma = \{\epsilon, 0, 00,000, 01,010\}$. The justified sequence
    $t = \Pstr[10pt]{
        (l){\lambda x} \cdot (a)@ \cdot (lyz-a,25:0){\lambda y z} \cdot (z-lyz,25:2)z \cdot (gl-a,30:2)\ghostlmd \cdot (gv-l,32:2)\ghostvar }
    $
    has quadruplet encoding
    $(\epsilon, \epsilon ,0,0), (0,@,0,0), (00, \epsilon, 1,0), (000,\epsilon, 1,2), (02,\epsilon,3,2), (02,\epsilon,5,2) $ whereas the quadruplets $ (\epsilon, g ,0,0), (0,@,0,0), (00, h~k, 1,0), (000,\epsilon, 1,2), (02,\epsilon,3,2), (02,\epsilon,5,2) $ encode the justified sequence
    $\Pstr[10pt]{
        (l){\lambda x^{[g]}} \cdot (a)@ \cdot (lyz-a,25:0){\lambda y z^{[h~k]}} \cdot (z-lyz,25:2)z \cdot (gl-a,30:2)\ghostlmd \cdot (gv-l,31:2)\ghostvar }
    $.
\end{example}

Two justified sequences are considered \defname{equal} if they have the same structure and same variable names after appending each pending lambda list to the bound variables in the corresponding lambdas. Formally, given two terms $M$ and $N$ we define the relation
$Eq_{M,N} \subseteq O_4(N) \times O_4(M)$ as:
\begin{align*}
(n,\overline{p},d,k)~Eq_{M,N}~(n',\overline{p'},d',k') &\iff
\begin{cases}
    n \mbox{ and $n'$ have the same type} \\
    d = d' \\
    k = k' \\
    \mbox{If $n=\lambda\overline{x}$ and $n'=\lambda\overline{y}$ then
        $\overline{x} \cdot \overline{p} = \overline{y} \cdot \overline{p'}$.
    }
\end{cases}
\end{align*}
The reflexive closure of $Eq_{M,N}^*$ defines an equivalence relation on
 the disjoint union $O_4(M) + O_4(N)$. By extension the relation $Eq_{M,N}^*$ defines an equivalence relation on the disjoint union $\justseqset(M) + \justseqset(N)$ where two sequences are in relation if their element-wise occurrences are.

In the remaining of the article we will consider occurrences and justified sequences equal up to this equivalence relation. For all variable names $\overline{x}$ and $\overline{y}$ the two notations $\lambda\overline{x}^{[\overline{y}]}$ and $\lambda\overline{x}\overline{y}$ will denote the same occurrence and will be used interchangeably.

We call a justified sequence \defname{canonical} just if all the \emph{pending lambda} components are empty. Observe that given a fixed term $M$, the encoding of a canonical justified sequence is uniquely determined by its structure and reciprocally: discarding the label information yields the structure; reciprocally the labels can be uniquely reconstructed from the structure of the sequence and the tree enabling relation (a chain of justification pointers in the sequence corresponds to a path in the tree).

\begin{example}
We have $\Pstr[10pt]{ (l){\lambda x y} \cdot (n-l,25){x} } \equiv \Pstr[10pt]{ (l){\lambda z y} \cdot (n-l,25){z}}$. But we also have $\Pstr[10pt]{ (l){\lambda x y z w t} \cdot (n-l,25){x} } \equiv \Pstr[10pt]{ (l){\lambda x} \cdot (n-l,25){x}}$ since they have same structure $(\lambda,0)\cdot({\sf Var},1)$.
\end{example}

We will make use of non-canonical justified sequences when studying
justified sequences in the context of $\beta$-reduction. In particular the \emph{pending lambda} component will allow us to define two useful operations:
\begin{itemize}
    \item Lambda merging: where abstractions are merged in a combined bunch lambda node;
    \item Node materialization: where a ghost node materializes into an occurrence of a bound or free variable. This occurs when a ghost variable points to a lambda node $\lambda\overline{x}^{[\overline{y}]}$ with label $i$ strictly greater than the arity of $|x|$ but smaller than $|x|+|y|$.
\end{itemize}


A justified sequence verifies the \defname{alternation condition} if the first node is a lambda node and subsequent nodes occurrences alternate between (i) a variable or application node (ii) a lambda node.

We say that an occurrence of a node in a justified sequence is \defname{hereditarily justified by some other occurrence} if recursively following justification pointers starting from the first occurrence in the sequence leads to the second one. Because justification pointers must honor the enabling relation $\enables$ induced by the term structure, if a node occurrence $n$ is hereditarily justified by some occurrence of a node $m\in\Nodes$ then $n$ is necessarily hereditarily enabled by $m$. Further if $m$ occurs only once in the justified sequence then the occurrences hereditarily justified by $m$ are precisely the occurrences of nodes that are hereditarily \emph{enabled} by $m$.

For any justified sequence $t$ we write $t^\omega$ to denote the last occurrence in $t$. The notion of sequence prefix naturally extends to justified sequences.
For any occurrence $n$ in $t$ we write $t_{\leq n}$ for the prefix subsequence of $t$ ending at $n$, and $t_{<n}$ for the prefix ending at the occurrence immediately preceding $n$ (or the empty sequence if $n$ is the first occurrence in $t$).
 We say that $t$ is an \defname{extension} of justified sequence $u$ if $u$ is a strict prefix of $t$ sharing the same justification pointers.

We use the standard operations borrowed from Game Semantics on justified sequences~\cite{Abr02}.

\begin{definition}[Projection]
Let $s$ be a justified sequence of nodes and $n$ be an occurrence in $s$.

\begin{itemize}
\item The \defname{projection of $s$} with respect to $n$, written $s\filter n$, is the subsequence of $s$ obtained by keeping only nodes that are hereditarily justified by $n$ in $s$;

 \item Let $A$ be a subset of nodes in $\Nodes$, we write $s \filter A$ to denote the subsequence of $s$ obtained by keeping only occurrences of nodes that belong to $A$;

 \item We will consider the subsequence $s\filter\ExtNodes$ consisting of external nodes of $s$. Observe that if $r_1, \ldots r_n$, $n\geq 1$ are the occurrences of the root in $s$ then it is also given by the projection of $s$ with respect to those occurrences: $s\filter\ExtNodes = s\filter r_1 \filter \ldots \filter r_n$.
\end{itemize}
\end{definition}

The P-view of a justified sequence is the sub-sequence obtained by reading the sequence backwards and following the justification pointer every other node: (i) if the last node read is a variable node then follow its justification pointer (\ie, skip all the occurrences lying underneath) (ii) if the node is a lambda node then move to the preceding node. Formally:
\begin{definition}[Views]
\label{def:views}
The P-view $\pview{s}$ of a justified sequence $s$ is defined recursively by:
$$\begin{array}{rcll}
 \pview{\epsilon} &=&  \epsilon \\
 \pview{s \cdot n }  &=&  \pview{s} \cdot n
    & \mbox{if $n$ is a variable or $@$ node;}
    \\
 \pview{\Pstr[10pt]{ s \cdot (m){m} \cdots (lmd-m,25){n}}} &=&
        \Pstr{ \pview{s} \cdot (m2){m} \cdot (lmd2-m2,30){n} }
    & \mbox{if $n$ is a lambda node;}
    \\
 \pview{s \cdot n }  &=&  n & \mbox{if $n$ is a lambda node with no pointer.}
\end{array}$$
The O-view, denoted $\oview{s}$, is the defined dually:
$$\begin{array}{rcll}
 \oview{\epsilon} &=&  \epsilon \\
 \oview{s \cdot n }  &=&  \oview{s} \cdot n
    & \mbox{if $n$ is an $\lambda$-node;}
    \\
 \oview{\Pstr[10pt]{s \cdot (m){m} \cdot \cdots \cdot (x-m,30){n}}} &=&
    \Pstr{ \oview{s} \cdot (m2){m} \cdot (n2-m2,60){n} }
    & \mbox{if $n$ is a variable node;}
    \\
 \oview{s \cdot n }  &=&  n
    & \mbox{if $n$ is an $@$ node.}
\end{array}$$
\end{definition}

Given two node occurrences occurrence $n$ and $m$ in a justified sequence $s$, we say that $n$ is \defname{visible at} $m$ just if $n$ occurs in the P-view $\pview{s_{\leq m}}$.

\subsection{Justified paths of the term tree}

We will consider paths of a term tree as a set of justified sequences. For any term $M$ we define the \defname{set of justified paths} $\pathset(M)$ as the set of justified sequences in $\justseqset(M)$ whose underlying sequences of nodes are the structural paths (\ie, path without ghosts) in the labelled-tree $\ctree(M)$ with associated justification pointers induced by the enabling relation as follows. Occurrence of bound variable nodes are justified by their binder node with link label determined by the variable index in the binding node; free variables are justified by the root (the first node in the sequence) with label index determined by the free variable index; and lambda nodes are justified by their parent (the immediate predecessor in the sequence) with label index given by the child index. This notion is well defined because a variable binder necessarily occurs in the path from it to the root.
\begin{property}[Path characterization]
\label{prop:tree_path_charact}
(i) An untyped term $M$ is uniquely determined by the subset of maximal justified paths of $\pathset$.
(ii) Further $M$ is uniquely determined, up to $\alpha$-conversion, by the \emph{structure} (\ie, node types and pointers) of the maximal justified paths in $\pathset$.
\end{property}
\begin{proof}
(i) Computation trees are in one-to-one correspondence with the standard tree representation of a lambda term. Further, because $\pathset(M)$ is prefix-closed it is uniquely determined by its maximal elements. (ii) Variables names are uniquely determined by the justification pointers and their associated label index.
\end{proof}

\begin{example}
  $\pathset((\lambda x.x y) (\lambda z.z))$ is the prefix closure of
  $\{
  \Pstr{ (r){\lambda} (at){@} (lx-at,30:0){\lambda x} (x-lx,30:1){x} (l2-x,30:1){\lambda} (y-r,40:1){y} },
  \Pstr{ (r){\lambda} (at){@} (lz-at,35:1){\lambda z} (z-lz,35:1){z}}
  \}$.
\end{example}

Algorithm~\ref{algo:termtree_readout_from_justitied_paths} shows how to reconstruct the tree from a set of justified sequences that is isomorphic to the justified paths of the term. Note how ghost variables get materialized by  assigning real variable names. Observe as well that variable nodes with same binder and same link label are assigned the same variable name. This assignment is sound since by Property~\ref{prop:tree_path_charact} justification pointers in justified paths uniquely determine variables names.

\begin{definition}[Induced labelled-tree]
Let $\mathcal{P}$ be a set of justified sequences of $M$ that is isomorphic to justified paths $\pathset(M_2)$ of some other term $M_2$.
Then $\mathcal{P}$ induces the labelled-tree $(P,L)$ defined as follows:
\begin{itemize}[nosep]
    \item Directed paths $P \subseteq \nat^*$ are \emph{generated} from justified sequences in $\mathcal{P}$: variable occurrences have child index $1$; the child index of lambda-node occurrences is given by their justification label.
    \item Labels are $\mathcal{L} = \ExtendedNodes \times P \times \nat$.
    \item Label function $L : P \rightarrow\mathcal{L}$.
    For $p\neq\epsilon$, let $s = s' \cdot n$ be any sequence in $\mathcal{P}$ generating path $p$.
    If $n$ has not justifier then $L(p) = (n, p, 0)$.
    Otherwise $L(p) = (n, b, i)$ where
    $b$ is the prefix of $p$ corresponding to $n$'s justifier and $i$ is the label of $n$'s justification pointer.
    \end{itemize}
\end{definition}
It is well-defined because sequences in $\mathcal{P}$ correspond to actual paths in some other computation tree, therefore for any two sequences $s_1, s_2$ generating a given path $p$, their last occurrence must have be the same node and have the same justifier.

\begin{algorithm}[!ht]
\begin{algorithmic}
\caption{Reconstructing a term tree from justified paths.}
\label{algo:termtree_readout_from_justitied_paths}
\REQUIRE{A set $\mathcal{P}$ of justified sequences of a term $M$ verifying $\mathcal{P} \equiv \pathset(N)$ for some other term $N$.}
\ENSURE{Computation tree of $N$ modulo $\alpha$-conversion.}
\STATE Let $(P, L)$ be the labelled-tree induced by $\mathcal{P}$.
\FORALL{node $n \in P$ (by depth-first enumeration)}
    \STATE Let $(l,b,i) = L(n)$
    \IF{$n \in \NodesVar \union \NodesLmd \union \NodesApp$}
        \STATE $L'(n) \leftarrow l$
    \ELSIF{$n \in\ghostlmd$}
        \STATE Let $\alpha$ be a fresh identifier
        \COMMENT{Will be used as prefix to name variable bound by $n$}
        \STATE $id(n) \leftarrow \alpha$
        \STATE $arity \leftarrow \max \{ i \ | \ L(v) = (\ghostvar, n, i) \mbox{ for some } v \in P \}$
        \STATE $L'(n) \leftarrow \lambda \alpha_1 \ldots \alpha_{arity} $
    \ELSIF{$n \in\ghostvar$}
            \STATE $\alpha \leftarrow id(b)$
            \COMMENT{Get naming prefix associated with this binder}
            \STATE $L'(n) \leftarrow \alpha_i $
            \COMMENT{Variable occurrences with same binder and justification label have same name.}
    \ENDIF
\ENDFOR
\RETURN The labelled-tree $(P, L')$.
\end{algorithmic}
\end{algorithm}

\section{Imaginary Traversals}

\subsection{Definition and properties}

\begin{definition}[ULC traversals]
The set of \defname{traversals} of an untyped lambda term $M$, denoted $\travulc(M)$, abbreviated $\travulc$ when the term is clear from context, is the set of justified sequences of nodes over $M$ recursively defined by the rules of Table~\ref{tab:trav_rules}. A traversal that does not have any extension is a \defname{maximal traversal}.
\end{definition}

\begin{table}
\begin{ruletablebox}{Imaginary traversals of $\lambda$-calculus}
\noindent {\bf PROGRAM -- Structural rules}
\begin{itemize}[leftmargin=3em]
    \item[\rulenamet{Root}] The singleton sequence $r$ is in $\travulc$ where $r$ is the root of the tree.

    \item[\rulenamet{App}] If $t \cdot @$ is a traversal then so is \Pstr[0.4cm]{t \cdot (at) @  \cdot (a-at,40:0) \alpha} where $\alpha\in\NodesLmd$ is $@$'s $0$th child $\lambda$-node.

    \item[\rulenamet{Lam}] If $t \cdot \alpha$ is a traversal where $\alpha\in\NodesLmd$ then so is $t \cdot \alpha \cdot n$ where $n$ denotes $\alpha$'s unique child. Furthermore:
        \begin{compactitem}
            \item[\rulenamet{Lam^@}] If $n$ is an $@$-node then it has no justifier,
            \item[\rulenamet{Lam^{\sf var}}] If $n$ is a free variable node then it points to the only occurrence of the root in
            $\pview{t \cdot \alpha}$. Otherwise it points to the only occurrence of its binder in $\pview{t\cdot \alpha}$.
        \end{compactitem}

    \item[\rulenamet{Lam^\ghostlmd}] If
  $\Pstr[0.5cm]{ t \cdot(alpha){\alpha} \cdot
(n){n} \cdot \ldots \cdot
(gl-n,40:i){\ghostlmd} \in\travulc}$ for some prefix $t$, $\alpha \in \ImNodesLmd$ and $n \in\ImNodesVar$ then
$$\Pstr[0.5cm]{ t \cdot(alpha){\alpha} \cdot
(n){n}
\cdot \ldots \cdot
(gl-n,40:i){\ghostlmd}\cdot (al-alpha,40:{|\alpha| + i - |n|}){\ghostvar}
      \in\travulc}$$
\end{itemize}
\emph{\bf PROGRAM -- Copy-cat rules}
\begin{itemize}[leftmargin=3em]
\item[\rulenamet{Var}] If \Pstr[0.5cm]{t \cdot (m){m} \cdot (alpha){\alpha}
    \ldots (n-alpha,50:i){n} \in \travulc} for $i>0$, $n \in \ImNodesVar$ hereditarily justified by an $@$-node;
     $m \in\ImNodesVar \union \NodesApp$; and $\alpha \in \ImNodesLmd$ then:
  $$\Pstr[0.5cm]{ t  \cdot
(m){m} \cdot (lx){\alpha}  \ldots (x-lx,30:i){n}
    \cdot (letai-m,40:i){\beta}
     \in\travulc}$$
\begin{description}[align=right,leftmargin=10em]
\item[Concrete] $i \leq |m|$: $\beta\in\NodesLmd$ is the $i$th child of $m$;
\item[Eta-expanded] $i > |m|$: $\beta$ is a ghost lambda node $\ghostlmd$.
\end{description}
\end{itemize}

\emph{\bf DATA -- Input-variable rules}
\begin{itemize}[leftmargin=3em]
\item[\rulenamet{IVar}] If $t \cdot n$ is a traversal where $n \in \ImNodesVar$ is hereditarily justified by the root. For every node $m \in \ImNodesVar$ occurring in $\oview{t\cdot n}$
and every $i\geq1$ we have $t \cdot n \cdot\alpha\in\travulc$ with $\alpha$ pointing to $m$ with label $i$, where:
\begin{description}[align=right,leftmargin=10em]
\item[Concrete] $i\leq|m|$: $\alpha\in\Nodes_\lambda$ is the $i$th child of $m$;
\item[Eta-expanded] $i>|m|$: $\alpha$ is a ghost lambda node $\ghostlmd$.
\end{description}
\end{itemize}

\caption{Imaginary traversals $\travulc$ of the untyped lambda calculus.}
 \label{tab:trav_rules}
\end{ruletablebox}
\end{table}

Some immediate properties that can be shown by induction on the rules:
\begin{property}
    \label{prop:trav_immediate}
   For any traversal $t$
   \begin{enumerate}[nosep]
   \item $t$ verifies the alternation condition;
   \item $t\filter\ExtNodes$ is a valid justified sequence (with respect to the enabling relation $\enables$) verifying the alternation condition;
   \item If the last node in $t$ is external then all the nodes in $\oview{t}$ are also external.
   \end{enumerate}
\end{property}

\paragraph{On-the-fly expansion}
The eta-expansion of $M$ with respect to a traversal is a term obtained by eta-expanding some of the subterms of $M$ according to the rules used to produce the traversal.
\begin{definition}[On-the-fly eta-expansion]
\label{def:onthefly_etaexpansion}
Let $M$ be an untyped term.
\begin{itemize}[nosep]
\item Let $n\in \Nodes$ be a node and $N$ denote the subterm of $M$ rooted at $n$. We write ${\sf ETA}(M, n)$ to denote the term obtained by substituting
the subterm $N$ in $M$ with the term $\lambda\theta. N \theta$ for some variable $\theta$ fresh in $N$.

\item Let $t$ be a \emph{finite} traversal of $M$. The \defname{eta-expansion of $M$ with respect to $t$}, written $M^t$, is defined as just $M$ if $t$ is the empty traversal; and we define $M^{u \cdot n}$ inductively by case analysis on the rule used to traverse node $n$, where $m$ denotes $n$'s justifier in $u$, if it exists, with link label $i$:
$$
\begin{array}{rlcrll}
    \rulename{Root}          & M^{u}  &\qquad & \rulename{Var}  & M^{u} \hfill &\mbox{if $i \leq |m|$; } \\
    \rulename{App}           & M^{u}  &       &                 & {\sf ETA}(M^{u}, m) \hfill  &\mbox{if $i > |m|$. } \\
    \rulename{Lam}           & M^{u}  &       & \rulename{IVar} & M^{u} \hfill &\mbox{if $i \leq |m|$; } \\
    \rulename{Lam^\ghostlmd} & M^{u}  &       &                 & {\sf ETA}(M^{u}, m)  \hfill &\mbox{if $i > |m|$.}
\end{array}
$$
\end{itemize}
\end{definition}
Observe that the tree obtained after eta-expansion with respect to $t$ contains the tree of $M$ itself, therefore paths in the tree of $M$ are also paths in the tree of $M^t$.

\begin{example}
Take $M = (\lambda u . u\ (y_1\ u)) (\lambda v . v\ y_2)$. Its computation tree is shown in Ex.~\ref{ex:missingoperand}. A valid traversal is $t = \Pstr[0.7cm]{(n0){\lambda }\ (n1){@}\ (n2-n1){\lambda u}\ (n3-n2){u}\ (n4-n1){\lambda v}\ (n5-n4){v}\ (n6-n3){\lambda }\ (n7-n0){y_1}\ (n8-n7){\lambda }\ (n9-n2){u}\ (n10-n1){\lambda v}\ (n11-n10){v}\ (n12-n9){{\ghostlmd^{1}}}\ (n13-n8){{\ghostvar^{1}}}\ (n14-n13){{\ghostlmd^{1}}}\ (n15-n12){{\ghostvar^{1}}}\ (n16-n11){\lambda }\ (n17-n0,35){y_2}\ }$. The eta-expansion of $M$ with respect to $t$ is
$M^t = (\lambda u . u~(y_1~(\lambda \alpha. u (\lambda \beta.\alpha \beta))) (\lambda v . v~y_2)$.
\end{example}

\paragraph{Path-View correspondence}
We now generalize a known result from the theory of traversals for higher-order grammars~\cite{Ong2006} and simply-typed lambda terms~\cite[Proposition 4.29]{BlumPhd} to the untyped setting. The game-semantic concept of `Proponent view' from definition~\ref{def:views} corresponds to the concept of `tree path' in the following sense:

\begin{proposition}[Path-View correspondence for ULC]
\label{prop:pathview_ulc}
Let $M$ be an untyped term and $t$ be a \emph{finite} traversal of $M$ then
$\pview{t}$ is a path (with associated justification pointers) in the computation tree of the eta-expansion of $M$ with respect to $t$. Thus, if $t$ does not contain ghost nodes then $\pview{t}$ is a structural path in the computation tree of $M$.
\end{proposition}
\begin{proof}
Proven by induction on the traversal rules of $t$.
\end{proof}

\begin{property}
The traversal rules are well-defined.
\end{property}
\begin{proof}
(i) Rule \rulenamet{Lam}: By Proposition~\ref{prop:pathview_ulc}, the P-view is a path in the tree of the eta-expansion of $M$ with respect to $t$. But the last node of $t$ is a structural node therefore, since the eta-expanded tree contains the tree of $M$, this path is necessarily also the path to $t^\omega$ in the tree of $M$ (and thus $\pview{t}$ does not contain any ghost node, even though $t$ itself may contain ghost nodes).
 Consequently, if the last node is a variable node, its enabler necessarily occurs exactly once in the P-view.

(ii) Rule \rulenamet{Var}: In the concrete sub-case, $m$ is necessarily itself a structural node since $|m|\geq i>0$, and its $i$th child exists in $\Nodes$ since $m$ as arity greater than $i$.
\end{proof}

\begin{proposition}
\label{prop:eta_expanded_trav}
Let $M$ be an untyped term and $t$ be a \emph{finite} traversal of $M$.
There exists a one-to-one function $\eta_t : \ExtendedNodes(M) \longrightarrow \ExtendedNodes(M^t)$
such that its implicit \emph{element-wise pointer-preserving} extension $\eta_t : \mathcal{J}(M) \longrightarrow \mathcal{J}(M^t)$ verifies the properties:
\begin{enumerate}[label=(\roman*),nosep]
    \item For any $u \in \travulc(M)$, $\eta_t(u)$ is a traversal of $M^t$.

    \item If $v$ traverses
    $M^t$ then there exists a traversal $u$ of $M$ (with possibly ghost nodes) such that $v = \eta_t(u)$.

    \item The restriction of $\eta_t$ to $\travset(M)$ defines a strong bijection with $\travset(M^t)$ (\ie, a bijection preserving the traversal structure, but not necessarily labels).
\end{enumerate}
\end{proposition}
\proofatend
The tree $\ctree(M)$ being by definition a subtree of $\ctree(M^t)$, induces a one-to-one mapping from $\Nodes(M)$ to $\Nodes(M^t)$.
We extend it to ghost nodes by mapping ghost occurrences in $t$ to the corresponding node resulting from the eta-expansions of subterms of $M$ in $M^t$, and all other ghost nodes not occurring in $t$ to the corresponding ghost nodes in $\ctree(M^t)$. Formally, we define $\eta_t$ by induction on the traversal $t$, observing that in rules \rulenamet{Var} and \rulenamet{IVar}, the on-the-fly eta-expansion of Definition~\ref{def:onthefly_etaexpansion} increments the arity of the justifier node $m$, which guarantees that the corresponding structural node exists in $\ctree(M^t)$.

(i) By induction on $u$. Reusing the rule used to traverse $u$ to traverse $\eta_t(u)$. For the case \rulenamet{Var} and \rulenamet{IVar}, when the \emph{eta-expanded} subcase of the rule is used to traverse a ghost variable that also appeared in $t$, we use instead the \emph{concrete} sub-case used to traverse $\eta_t(u)$.

(ii) By induction on $v$, applying to $u$ each rule that is applied to $v$. For the case \rulenamet{Var} and \rulenamet{IVar}, if the index $i$ is greater than the arity of the node $m$ in $M$ then the lambda node from $v$ gets replaced by a ghost lambda node in $u$.

(iii) By (i) the function is well defined, it is injective by definition, and by (ii) it is surjective.
\endproofatend

\begin{example}
Continuing with the previous example: the following traversal obtained from $t$ is a valid traversal of $M^t$ where the occurrences corresponding to ghosts in $t$ are underlined:
$$\Pstr[0.7cm]{(n0){\lambda }~(n1){@}\ (n2-n1){\lambda u}~(n3-n2){u}~(n4-n1){\lambda v}~(n5-n4){v}\ (n6-n3){\bf \lambda}\ (n7-n0){y_1}\ (n8-n7){\bf \lambda}~(n9-n2){u}~(n10-n1){\lambda v}~(n11-n10){v}\ (n12-n9){\underline{\lambda\beta}}~(n13-n8){\underline{\alpha}}~(n14-n13){\underline{\lambda}}~(n15-n12){\underline{\beta}}\ (n16-n11){\lambda}~(n17-n0,35){y_2} }.$$
\end{example}

\paragraph{Traversal core}

We now generalize the notion of \emph{traversal core} from~\cite{BlumPhd} to imaginary traversals. In the typed setting, the \emph{core} is defined as the subsequence of a traversal consisting of external nodes only. In the untyped setting, we complement this filtering with a relabelling operation on lambda nodes.

We consider sequences of variable names in $\VarSet^*$ as stacks equipped with the standard stack operations:
for any sequence of variables $\overline{x} = x_1 \ldots x_n$, $n,j\geq0$, the `pop' operation $pop_j$ removes the first $j$ elements of the sequence: $pop_j (x_1 \ldots x_n) = x_{j+1} \ldots x_n$ if $j<n$, and $pop_j (x_1 \ldots x_n) = \epsilon$ otherwise. And for any sequence $\overline{y} = y_1 \ldots y_m$, $m\geq0$ we write $\overline{x}\overline{y}$ for
$x_1 \ldots x_n y_1 \ldots y_n$, the stack obtained after pushing the variables $\overline{x}$ (read backwards) onto $\overline{y}$.

\begin{definition}[Core projection]
\label{def:coreprojection}
Let $\overline{y} \in \VarSet^*$ denote a stack of variable names.
We define the partial function $\coresymbol_{\overline{y}}\colon \justseqset(M) \longrightarrow \justseqset(M)$ on the set of justified sequences by induction on the sequence.
Let $t \cdot n\in\justseqset(M)$ denote a justified sequence for some prefix subsequence $t$ and last occurrence $n$.
\begin{equation*}
\coresymbol_{\overline{y}}(t\cdot n) =
\begin{cases}
\coresymbol_{\epsilon}(t) \cdot n
    & \mbox{if } n\in\ImNodesVar\inter\ExtNodes \\
\coresymbol_{pop_{|n|}(\overline{y})}(t)
    & \mbox{if } n\in(\ImNodesVar\inter\IntNodes)\union\NodesApp \\
\coresymbol_{\epsilon}(t) \cdot \lambda \overline{x}\overline{y}
    & \mbox{if } n\in\NodesLmd\inter\ExtNodes, n \mbox{ labelled } \lambda\overline{x}\\
\coresymbol_{\overline{x} \overline{y}}(t)
    & \mbox{if } n\in\NodesLmd\inter\IntNodes, n \mbox{ labelled } \lambda\overline{x} \\
\coresymbol_{\epsilon}(t) \cdot\lambda\overline{y}
    & \mbox{if } n\in\ghostlmd\inter\ExtNodes \\
\coresymbol_{\overline{y}}(t)
    & \mbox{if } n\in\ghostlmd\inter\IntNodes \ .
\end{cases}
\end{equation*}
The above definition is pointer-preserving, so that the justifiers in $\coresymbol(t\cdot n)$ are defined to be the occurrences corresponding to the original justifiers in $t \cdot n$.

Note that if $t$ is a traversal, the second last case is not needed since ghost lambdas are necessarily internal by Prop~\ref{prop:trav_immediate}.

We call $\coresymbol_\epsilon(t)$ the \defname{core} of $t$ which we abbreviate as $\core{t}$. The P-view of the core $\pview{\core{t}}$ is called the \defname{core P-view} of $t$.
\end{definition}

The suffix parameter $\overline{y}$ in the definition of $\coresymbol$, is used as an accumulator for the list of variables forming the abstractions $\lambda \overline{y}$ that eventually gets prepended to the last external lambda nodes. We call it the \defname{stack of pending lambdas}. Hence, \emph{core projection} can be defined more succinctly as the sequence obtained by removing internal nodes and prepending to each external lambda node's label, the stack of \emph{pending lambdas} between that point to the next external node.

Since the core of a justified sequence only contains external nodes it follows that $\coresymbol$ is idempotent.

\begin{example}Take the term $\lambda w . (\lambda x y .x) z$ and consider traversal $t = \lambda w\cdot @ \cdot \lambda x y\cdot x\cdot\lambda \cdot z$.
    Then $\coresymbol(t) = \lambda wy \cdot z$.
\end{example}

\begin{example} Take $(\lambda x . x x)(\lambda y . y)$ with traversal
$t$ represented below. Projecting nodes hereditarily justified by the root gives
$\Pstr[0.7cm]{(n0){\lambda }\ (n1-n0){{\ghostvar^1}}}$
while the core projection gives
$\coresymbol(t) = \Pstr[0.7cm]{(n0){\lambda y}\ (n1-n0){y}}$.
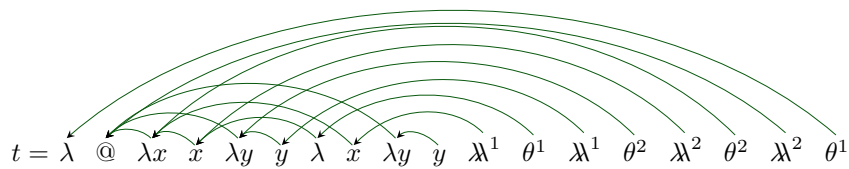
$$t = \Pstr[0.7cm]{(n0){\lambda}~(n1){@}~(n2-n1){\lambda x}~(n3-n2){x}~(n4-n1){\lambda y}~(n5-n4){y}~(n6-n3){\lambda }~(n7-n2){x}~(n8-n1){\lambda y}~(n9-n8){y}~(n10-n7){{\ghostlmd^{1}}}~(n11-n6){{\ghostvar^{1}}}~(n12-n5){{\ghostlmd^{1}}}~(n13-n4){{\ghostvar^{2}}}~(n14-n3){{\ghostlmd^{2}}}~(n15-n2){{\ghostvar^{2}}}~(n16-n1,41){{\ghostlmd^{2}}}~(n17-n0,41){{\ghostvar^{1}}}}$$
\end{example}

\subsection{Imaginary traversals subsume STLC traversals}

Traversals for simply-typed languages were previously studied in~\cite{BlumPhd}.
The traversal rules defined in Table~\ref{tab:trav_rules} closely match those of the simply-typed lambda calculus and PCF from~\cite{BlumPhd} with some important differences:
\begin{description}[nosep]
  \item[No interpreted constants] Unlike PCF, there are no interpreted constants in the present setting therefore the rule \rulenamet{Value} and \rulenamet{InputValue} from the original presentation are not needed.

  \item[No $\eta$-long expansion] In the original STLC traversals, the term is eta-long expanded prior to calculating the set of traversals. This guarantees that the operand of an application always exists in the tree.
  Imaginary traversals, on the other hand, are defined on the unmodified tree representation of the term.

  \item[On-the-fly $\eta$-expansion] In the untyped setting, eta-long expansion is an infinite process, so instead of eta-long expanding the term prior to traversing it, imaginary traversals perform eta-expansion  `on the fly'. Eta-expansion occurs in rule \rulenamet{Var} when the arity of a variable in operand position is too low to statically determine the operand of an application (case $i>|n|$). When the variable arity is high enough ($i\geq|n|$), the definition of the rule coincides with STLC and the static tree representation of the term dictates the next node to visit.

  \item[Free variables] Eta-expansion can also occur on free variables, so like for  rule \rulenamet{Var}, the input-variable rule allows for infinitely many eta-expansions ($k>0$).

  \item[Traversal core] In the typed setting, the core is  obtained by just filtering nodes with respect to the tree root. In the untyped settings, lambda nodes also get relabelled after filtering.

  \item[Traversing ghost nodes] There is an additional rule \rulenamet{Lam^\ghostlmd} for the case where a traversal ends with a ghost lambda node. In the concrete sub-case, rule \rulenamet{Lam} visits the unique child node of the last lambda node in the traversal. In the eta-expanded sub-case where the last node is a ghost lambda node, there is no such child node, so we visit an imaginary one: the variable node that would be created if we were to eta-expand the sub-term under the lambda.

   The ghost placeholder $\ghostvar$ thus represents an occurrence of the $j$th variable that would be bound by lambda node $\alpha$ if the sub-term at node $\alpha$ were eta-expanded $i-|n|$ times: hence $j = |\alpha| + i - |n|$.
   (Observe that in this case we necessarily have $i>|n|$ since the $i$th child of $n$ is a ghost variable node.)
\end{description}

Let's fix a simply-typed term-in-context $\Gamma \entail M :T$, with typed-context $\Gamma$ (a set of typed variables), and simple type $T$. Its \defname{eta-long form} is defined inductively on the type $T$ and is obtained by recursively eta-expanding every subterm as many times as possible with respect to the type of the subterm~\cite{Ong2006,BlumPhd}.
By abuse of language we say that an untyped term $M$ is in eta-long form if it inhabits a simple type $T$ such that it equates its own eta-long expansion with respect to $T$, modulo variable renaming.

STLC traversals from~\cite{BlumPhd}
are defined for eta-long form only. The rules
consist of the same rules as the imaginary traversal from Table~\ref{tab:trav_rules} with the exclusion of all sub-cases required to implement `on-the-fly eta-expansion' or to traverse ghost nodes:
\begin{definition}[STLC traversal~\cite{BlumPhd}]
Let $\Gamma \entail M : T$ be a simply-typed term in context. The set of traversals $\travstlc(\Gamma \entail M : T)$ from \cite{BlumPhd} is  the set of traversals of the \emph{eta-long form} of $M$ generated from the rules of Table~\ref{tab:trav_rules} with the exclusion of \rulenamet{Lam^\ghostlmd} and the `eta-expanded` sub-cases of \rulenamet{Var} and \rulenamet{IVar}.
\end{definition}

ULC and STLC traversals coincide on eta-long forms in the following sense:
\begin{proposition}
\label{prop:ulc_and_stlc_trav_coincide}
Let $M$ be an untyped term inhabiting a simple type $T$ for some context $\Gamma$ ($\Gamma \entail M : T$) and $\etalf{M}$ denote its \emph{eta-long normal form} with respect to $T$. Then  we have
$$\travstlc(\Gamma \entail M : T) = \travulc'(\etalf{M})$$
where $\travulc'$ denote the set of traversals obtained with rules of Table~\ref{tab:trav_rules} with the exclusion of the eta-expanded case in \rulenamet{IVar}.
\end{proposition}
\begin{proof}
Because the term is eta-long, the condition $i>|m|$ in rules \rulenamet{Var} never holds while traversing the term. This is shown by induction on the rules, observing that the arity of an $@$ node is necessarily equal to the arity of its $0$th child lambda node, and the arity of a variable node with binding index $i$ is necessarily equal to the arity of the $i$th child of its binder's parent.
And since the eta-expanded rule \rulenamet{IVar}  is excluded, traversals in $\travulc'(\etalf{M})$ do not contain any ghost node. Other rules coincide with STLC, therefore the equality holds.
\end{proof}

\subsection{Property of ghost nodes}

It is helpful to think of ghost nodes as the counterpart of complex numbers which are sometimes used in mathematics
to prove trigonometry identities: they are introduced intermediately to perform some calculation (e.g.\ using De Moivre's Theorem) but do not appear in the final result. Just like the imaginary number $i$ is created out of the impossibility of calculating the square root of $-1$, ghost nodes are defined from the impossibility of ``traversing'' a beta-redex of a lambda term due to an insufficient number of operands.

Ghost nodes appear in a traversal when the arity of a node is too low to be able to structurally traverse it:
\begin{property}
\label{prop:ghost_justifier_arity}
If $\Pstr[0.5cm]{(x){x} \cdots (y-x,30:i){y}} \in \travulc$ then
$ y \in \ghostvar \union \ghostlmd \iff i > |x|.$
In particular for all $n\in\NodesVar$ and $\alpha\in\NodesLmd$:
\begin{enumerate*}[nosep,label=(\roman*)]
\item $\Pstr[0.5cm]{(x){n} \cdots (y-x,30:i){\ghostlmd}} \in \travulc \implies i > |n|$
\item $\Pstr[0.5cm]{(x){\alpha} \cdots (y-x,30:i){\ghostvar}} \in \travulc \implies i > |n|$.
\end{enumerate*}
\end{property}

\begin{definition}
We call \defname{ghost materialization} any application of a traversal rule where the last occurrence in the traversal prior to applying the rule is a ghost node ($\ghostvar$ or $\ghostlmd$), and the node traversed after applying the rule is a structural node of the tree (in $\Nodes$).
\end{definition}

\begin{remark}[Materialization rule]
Observe that among all the rules defined in Table~\ref{tab:trav_rules}, the rule \rulenamet{Var} is the only rule that can materialize a structural node in $\Nodes$ from a traversal ending with a ghost node. This means that after traversing ghost nodes, the only way to `return' to structural nodes is to visit a ghost variable node $\ghostvar$ with an application of rule \rulenamet{Var} of the following form:

$$\rulename{Var}\ \  \Pstr[0.5cm]{ t \cdot(beta){\beta} \cdot
(y){y} \cdot (l){\alpha}  \ldots (t-l,30:i){\ghostvar}
    \cdot (lx-y,40:i){\lambda \overline{x}}
     \in\travulc}$$
where
\begin{itemize*}[nosep,label=]
\item $y \in \NodesVar$,
\item $\lambda \overline{x} \in \NodesLmd$ is the $i$th child lambda node of $y \in \NodesVar$,
\item $\alpha$ is either a structural lambda node in $\Nodes_\lambda$ or a ghost lambda node in $\ghostlmd$,
\item $0\leq \alpha < i \leq |y|$.
\end{itemize*}
\end{remark}

\subsection{Correspondence with Game semantics}
\label{sec:traversal_correspondence_stlc}

In~\cite{BlumPhd} we formalized the correspondence between the theory of traversals and Game Semantics in the setting of simply-typed languages: there is a bijection between the set of traversals of a term and its revealed interaction game denotation. Furthermore, the core projection yields a bijection with the standard innocent game denotation.
In other words, the traversal cores are precisely plays from the game denotation of the term.
Formally:
\begin{theorem}[Traversal-Play Correspondence for STLC (Theorem 4.96 in~\cite{BlumPhd})]
\label{thm:gamesem_correspondence_stlc}
The following two bijections hold for every simply-typed term $\Gamma \entail M :T$ with eta-long normal form $\etalf{M}$:
\begin{eqnarray*}
    \travstlc(M) & \cong & \revsem{\Gamma \entail M :T} \\
    \travstlc^\coresymbol(\etalf{M}) & \cong & \sem{\Gamma \entail M :T} \enspace .
\end{eqnarray*}
where $\revsem{\Gamma \entail M : T}$ and
$\sem{\Gamma \entail M : T}$
denote respectively the \emph{revealed game denotation} and \emph{innocent game denotation} of
$\Gamma \entail M : T$;
and $travstlc^\coresymbol$ denotes the image of $travstlc$ by $\coresymbol$ (\ie, set of justified sequences that are \emph{core} of some STLC-traversal).
\end{theorem}

\paragraph{Game semantics Correspondence for ULC}
What would be the equivalent of Theorem~\ref{thm:gamesem_correspondence_stlc} in the untyped case? In his thesis, Andrew Ker studied Game models for the untyped lambda calculus~\cite{KerThesis}.  We conjecture that the traversal-game semantics correspondence for STLC also holds for ULC
Table~\ref{tab:trav_rules}
between the traversals from Table~\ref{tab:trav_rules} and the game model of ULC from Ker's thesis, \ie there is an isomorphism between the set of imaginary traversals and the revealed game denotation of the ULC term, and further, an isomorphism between the standard game semantics and the set of traversal cores:

\begin{conjecture}[Traversal-Play Correspondence for ULC]
\label{conj:ulc_corresp}
For every untyped lambda-term $M$ there are two bijections
$ \travulc(M) \cong \revsem{M}$ and
 $\travulc^\coresymbol(M) \cong \sem{M}$
where $\sem{M}$ denotes the innocent \emph{effectively almost everywhere copycat(EAC)} game denotation of $M$ defined in~\cite{KerThesis},
$\revsem{M}$ denotes the corresponding interaction game denotation where internal moves are not hidden during strategy composition,
and $\travulc^\coresymbol(M)$ denotes the set of traversal cores of $M$.
\end{conjecture}
A possible proof of this conjecture might consist in a similar argument to the STLC case~\cite{BlumPhd} but based on Ker's game model of the untyped lambda calculus instead of the innocent game model of STLC. There seems, for instance, to exist a possible connection between the ``on-the-fly'' eta-expansion of imaginary traversals and  the morphism $Fun : U \rightarrow (U \Rightarrow U)$ of Ker's game category where $U$ denotes the maximal game arena.

\section{Normalizing Traversals}

In this section we show how traversals can help evaluate lambda terms and find their normal form when it exists. The crux of the algorithm lies in the \emph{Path Characterization} (Proposition~\ref{prop:path_charact_stlc} and Theorem~\ref{thm:path_charact_ulc}) which states that the set of traversals \emph{core P-views} captures what is strictly required from traversals to reconstruct the normal form of a term.
For STLC, the characterization result was shown using the Game Semantics correspondence; we prove its untyped counterpart using instead a term-rewriting argument.

Path Characterization suggests a normalization method based on enumeration of all maximal traversals. Unfortunately, due to the generous non-determinism of the free variable rule, traversals can be arbitrarily long and therefore the set of traversals may be infinite, even for terms with a beta-normal from! Not all traversals need to be enumerated, fortunately. It is sufficient to enumerate a subset of traversals that covers all the possible core P-views. Some traversals are redundant in the sense that there exist shorter traversals with identical core P-view.
We capture this notion by considering equivalence classes on traversals: two traversals are in the same class just if they have the same core P-view. It then suffices to exhibit a traversal subset that is \emph{complete} for the equivalence classes, in the sense that it contains at least one traversal for each equivalence class. If such subset is effectively computable and finite then it yields a normalization procedure. We show how this can be done for both STLC and ULC.

For the simply-typed lambda calculus, we consider the subset of \emph{branching traversals} which verifies the desired completeness property. For simply typed terms in eta-long form this subset is finite. This yields a normalization procedure for STLC. Soundness follows from the characterization theorem which is an immediate consequence of the Game Semantics correspondence of~\cite{BlumPhd}.

For the untyped lambda calculus, \emph{branching traversals} are not adequate as they can still be infinite even for terms with a normal form. We introduce another kind of traversals, called \emph{normalizing traversals} to limit the non-determinism of the traversals and derive  from it a procedure to calculate $\beta$-normal forms of untyped lambda terms when they exist.

\subsection{Quotienting}

We introduce a quotient relation that identifies traversals with identical core P-views:
\begin{definition}[Quotienting]
The core P-view function
$\pview{\core{\_}} \colon \travulc \longrightarrow \justseqset $
is the composition of $\pview{\_}\colon \justseqset \longrightarrow \justseqset$ with the core projection $\coresymbol\colon\travulc  \longrightarrow \justseqset$ that is $t \longmapsto \pview{\core{t}}$.
We define $\sim$ as the equivalence relation over $\travulc$ induced by $\pview{\core{\_}}$ up to relabelling (same structure but not necessarily the same labels). Formally:
$$t \sim u \quad \mbox{ iff } \quad  \pview{\core{t}} \equiv \pview{\core{u}} .$$
We write $\travulc/{\sim}$ for the set of equivalence classes of $\travulc$. We identify a $\sim$-equivalence class with the structure of the P-view core of the traversals it contains. A subset $T\subseteq \travulc$ is \defname{$\sim$-complete} if it contains at least one element for each $\sim$-equivalence class; that is if $T/{\sim} = \travulc/{\sim}$.
\end{definition}

In the next sub-section, we explore
\emph{branching traversals}, a $\sim$-complete traversal subset that can be effectively computed for lambda terms inhabiting some simple type.

\subsection{Branching traversals}

A normalization procedure based on enumerating all traversals is not practical because the set of traversals can be infinite. This is expected for non-normalizing term such as $\Omega = (\lambda x. x x)(\lambda y. y y)$ which has infinitely long traversals of the form $\lambda x \cdot x \cdot \lambda y \cdot y \cdot x \cdot \lambda y \cdot y \cdot \ldots$ But this is also the case for terms having a normal form. Take for example $M = \lambda f . f (\lambda x. x)$, then for all $k\geq0$ the justified sequence $t_k = \lambda f \cdot f \cdot (\lambda x \cdot  x)^k$ (with appropriate pointers) is a traversal.
Remember that rule \rulenamet{IVar} offers two non-deterministic choices when traversing a free variable:
\begin{itemize}
\item[(J)] The variable to pick in the O-view (the justifier),
\item[(L)] The child lambda node to pick amongst the children of variable picked in (J).
\end{itemize}
The choice (J) alone gives rise to the pattern displayed in $t_k$ where the node $\lambda x$ from the O-view is repeatedly picked {\it ad infinitum} within a single traversal.

\begin{remark}[Game-semantic intuition]
The Game Semantics correspondence explains why traversals can be infinite: In the game denotation of a lambda term $M$, at every point in a play where it is Opponent's turn to play, all possible Opponent moves must be accounted for. More generally, the game denotation must allow Opponent's moves modeling the behaviour of \emph{all} contexts that can possibly interact with $M$. The $t_k$s traversals thus accounts for all possible denotations of the function parameter $f$ of $M$: for each $k$, there exists a term that applies its argument $k$ times: $F_k = \lambda g . g (g ( \ldots (g z)))$ with $k$ applications of $g$, and there got to be plays in the game denotation of the term $M$ to accounts for those possible values of argument $f$. Fortunately, traversals accounting for all those contexts are redundant for normalization: due to the absence of side-effects, calling the same argument multiple times always involves the same underlying computation in $M$. We formalize this intuition with the notion of \emph{branching traversals} (Definition~\ref{dfn:branching_traversals})  which prevent such repetitive behaviour while still covering all $\sim$-equivalence classes (Proposition~\ref{prop:branching_traversal_simcomplete}).
\end{remark}

One may view traversals as a mechanism to explore the beta-normal form of the term in a dept-first search manner. Under such a view, one can interpret the non-determinism in \rulenamet{IVar} as a `branching point' in the exploration. In the absence of side-effect, it is sufficient to explore each possible branch only once. \emph{Branching traversals} implement this idea by restricting the rule \rulenamet{IVar} so as to traverse only nodes leading to paths in the computation tree that are yet unexplored: when choosing the next lambda node to visit, it forces choice `(J)' to be the \emph{latest} variable node in the traversal, and restricts choice `(L)' to be some child lambda node of that variable node.

\begin{definition}
\label{dfn:branching_traversals}
The set of \defname{branching traversals} $\travsetbr$ is the subset of $\travulc$ defined by induction with the rules of Table~\ref{tab:trav_rules} where the justifier node in the
input-variable rule \rulenamet{IVar} is restricted to be necessarily the last node in the traversal ($m=n$).
We will use the subscript `${}_\branching$' to refer to this system of rules. The input variable rule can thus be stated by the following derivation rule:
\infrule[$\rulefont{IVar_\branching}$]
     {t \cdot n \in\travsetbr
      \andalso n \in\ExtNodes\inter\NodesVar
      \andalso n \enables_i\alpha
      \andalso i \geq 1
     }
     {\Pstr[0.5cm]{t \cdot (n){n} \cdot (l-n,25:i){\alpha}} \in \travsetbr}
\end{definition}

\begin{remark}[Game semantic intuition]
Lambda nodes correspond to opponent moves in game semantics. Restriction on the opponent moves limit the set of contexts in which a term can appear.
By limiting the possible lambda nodes to traverse after a free variable, the branching restriction essentially eliminates all contexts containing a sub-expression where the same argument is referred more than once, or two arguments of the same higher-order function are called in the same function body. This for instance exclude the Kierstead contexts (in which the argument $f$ is called twice) but also contexts like  $\lambda f \lambda g . f (\lambda f (\lambda x . g (\lambda y . x)))$ where the consecutive lambdas abstractions $\lambda f \lambda g$ have in their body an occurrence of both $f$ and $g$.
\end{remark}

\begin{property}
\label{prop:core_truncation_at_externallambda}
Let $t\in\travulc$ be a traversal which does not contain any ghost occurrence, and $m$ be an occurrence in $t$ of an external $\lambda$-node (\ie, $m \in \NodesLmd\inter\ExtNodes$). Then $\core{t_{<m}} = \core{t}_{<m}$.
\end{property}
\begin{proof}
By an easy induction on $t$ using the fact that when recursively calculating $\coresymbol(t)$, external lambda nodes reset the stack of pending lambdas.
\end{proof}

\begin{proposition}[Branching traversals are $\sim$-complete]
\label{prop:branching_traversal_simcomplete}
  $\travsetbr$ is  $\sim$-complete.
\end{proposition}
\proofatend
Let $t\in\travulc$ and let's first assume that $t$ does not contain any ghost occurrence. We show by strong induction on $t$ that there is a subsequence $u\in\travsetbr$ of $t$ such that $\pview{\core{t}} = \pview{\core{u}}$, and so in particular
$\pview{\core{t}} \equiv \pview{\core{u}}$.
We do a case analysis on the last node $t^\omega$ of $t$.
\begin{itemize}
\item $t^\omega\in \ExtNodes$: Suppose $t^\omega$ is a variable or an $@$ node then we can conclude immediately from the induction hypothesis on $t_{<t^\omega}$ and using the rules \rulenamet{Lmd} of $\travsetnorm$.

Suppose $t^\omega$ is a lambda node. If it has no justifier then it is the root in which case we conclude by taking  $u=t=t^\omega$. Otherwise let $m$ denote $t^\omega$'s justifier in $t$. By the I.H.~on $t_{\leq m}$ there is $u' \in \travsetnorm$ such that $\core{u'} = \core{t_{\leq m}}$. Take $u = u' \cdot t^\omega$ where $t^\omega$ points to its immediate predecessor $m$. Then $u$ is clearly a $\travsetnorm$-traversal by rule \rulenamet{IVar} and:
\begin{align*}
\pview{\core{u}} &= \pview{\core{u'}} \cdot t^\omega & \mbox{Def.~of $\coresymbol$}\\
 &= \pview{\core{t_{\leq m}} \cdot t^\omega} & \mbox{By I.H.~on $t_{\leq m}$}\\
 &= \pview{\core{t}_{\leq m} \cdot t^\omega} & \parbox[t]{8cm}{By Prop.~\ref{prop:core_truncation_at_externallambda} since $m$ is necessarily followed in $t$ by a $\lambda$-node in $\ExtNodes$.}
\end{align*}
Now because $m$ justifies $t^\omega$, by Def.~of P-view, \emph{up to renaming of lambda variables} the sequences $\pview{\core{t}_{\leq m} \cdot t^\omega}$ and $\pview{\core{t}}$ are equal: $\pview{\core{t}_{\leq m} \cdot t^\omega} \equiv \pview{\core{t}}$.
 But because $m$ is a variable node, its label is kept untouched by the transformation $\coresymbol$ and therefore the equality holds.

\item $t^\omega\not\in \ExtNodes$: Let $n$ be the last occurrence in $t$ that is in $\ExtNodes$, and let $m_1 \ldots m_q$, $q>0$ be the occurrences of internal nodes following $n$ in $t$ (so that $m_q = t^\omega$). By definition of the traversal rules, a variable in $\ExtNodes$ is necessarily followed by a lambda node in $\ExtNodes$, therefore $n$ is necessarily a lambda node. By definition of $\coresymbol$, $n$ is therefore also the last occurrence in $\core{t}$ therefore $\core{t} = \core{t}_{\leq n}$.
But \emph{up to relabelling}, $\core{t}_{\leq n} = \core{t_{\leq n}}$; more precisely, $\core{t}_{\leq n}$ is obtained from $\core{t_{\leq n}}$
by pre-pending to $n$'s label the stack of pending lambdas of $t_{>m}$ (the internal nodes following $n$ in $t$). Applying the induction hypothesis on $t_{\leq n}$ gives $\core{t_{\leq n}} = \core{u'}$ for some $u' \in \travsetnorm$.
To conclude it therefore suffices that to show that $u = u' \cdot m_1 \ldots m_q$ is also a traversal of $\travsetnorm$.

We prove by finite induction on $q$ that $u_{\leq m_q} \in\travsetnorm$ and $m_q$'s justifying node and label are same in $u$ and $t$. For $q=0$, we just apply rule \rulenamet{Lam} on $u'$.
For $q>0$: by case analysis on the rule used to visit $m_q$ in $t$.
For structural rules \rulenamet{Root}, \rulenamet{App} and \rulenamet{Lam} it follows immediately by induction.
Rule \rulenamet{Var}: If $m_q \in \NodesVar$ then by the Path-View correspondence, $\pview{u_{<m_q}}$ is a path in the tree from the root to $m_q$ therefore $m_q$'s binder necessarily occur in $u$, we can therefore conclude using the I.H.~on $u_{<m_q}$ and applying \rulenamet{Var} on $u_{<m_q}$ to get $u_{\leq m_q}$.
\end{itemize}
Suppose that $t$ contains ghost occurrences then we consider the term $M^t$.
By Prop.~\ref{prop:eta_expanded_trav}(i), $\eta_t(t)$ is a traversal of $M^t$. Hence by the above, there exists $u'\in\travsetnorm(M^t)$ such that
$\pview{\core{u'}}=\pview{\core{\eta_t(t)}}$.
By Prop.~\ref{prop:eta_expanded_trav}(ii),
there exists $u \in \travulc(M)$ such that
$u' = \eta_t(u)$. Recall that the normalizing traversals is defined as the subset of traversals where external lambda nodes always point to their immediate predecessor. Therefore since $\eta_t$ is pointer-preserving $u$ must necessarily belong to $\travsetnorm(M)$.

By definition of $\coresymbol$ we have $\core{t} = \core{\eta_t(t)}$, and since $u$ is a subsequence of $t$ we also have $\core{u} = \core{\eta_t(u)}$.
Hence $\pview{\core{t}}=\pview{\core{\eta_t(t)}}=
\pview{\core{u'}} = \pview{\core{\eta_t(u)}} = \pview{\core{u}}$.
\endproofatend

\begin{property}[Infinite branching traversals]
\label{prop:branching_spine_property}
Let $t\in \travsetbr$ be a branching traversal. If $t$ is infinite then it necessarily contains a ghost node occurrence.
\end{property}
\begin{proof}
Suppose $t$ does not contain any ghost node then
$t$ is obtained without using \rulenamet{Lam^\ghostlmd_\branching} and without the `eta-expansion' subcases of the variable rules. The remaining rules correspond precisely to the traversal rules of
\cite{Ong2006} (where \rulenamet{IVar} is renamed \rulenamet{Sig}).
We can therefore appeal to the Spinal Decomposition Lemma
\cite[Lemma 14]{Ong2006} which shows that if $t$ is infinite then there is an infinite sequence of prefixes of $t$ whose P-views (the \emph{spine} of $t$) is strictly increasing. By the Path correspondence this means that there is an infinite path in the computation tree of $M$ which gives a contradiction. Hence $t$ is either finite or contains a ghost node.
\end{proof}

\begin{property}
\label{prop:etalong_trav_finite}
Let $M$ be an untyped term in eta-long form (with respect to some simple type it inhabits). Then
\begin{enumerate*}
\item[(i)] All traversals in $\travsetbr(M)$ are finite;
\item[(ii)] The set $\travsetbr(M)$ is finite.
\end{enumerate*}
\end{property}
\begin{proof}
By Proposition~\ref{prop:ulc_and_stlc_trav_coincide}, traversals in $\travsetbr(M)$ do not contain any ghost node therefore Prop.~\ref{prop:branching_spine_property} implies (i).
(ii) Traversal rules that are not involving ghost nodes all have bounded non-determinism therefore a traversal only has a finite number of immediate extensions. Since traversals are finite by (i), this implies that the set of traversals is itself finite.
\end{proof}

\subsection{Normalization procedure for eta-long forms (STLC)}

The normalization procedure for the simply-typed lambda calculus is given by Algorithm~\ref{algo:stlc_normalization_by_traversals}.

\begin{algorithm}[!ht]
\caption{Eta-long normalization by traversals for STLC}
\label{algo:stlc_normalization_by_traversals}
\begin{algorithmic}
\REQUIRE{A term $M$ admitting an eta-long form \emph{wrt} some simple-type $T$ that it inhabits.}
\ENSURE{The tree of the eta-long beta-normal form of $M$.}
\begin{enumerate}[nosep]
  \item Calculate the eta-long form $\etalf{M}$ of $M$ \emph{wrt} $T$;
  \item Enumerate maximal \emph{branching} traversals of $\etalf{M}$ from Def.~\ref{dfn:branching_traversals}:
  \begin{itemize}[nosep]
  \item For each traversal $t$, get the \emph{traversal core} $\core{t}$ by removing all internal nodes;
  \item Calculate the P-view of the traversal core $\pview{\core{t}}$;
  \item Interpret $\pview{\core{t}}$ as a path in the tree representation of the eta-long $\beta$-nf of $M$;
  \end{itemize}
  \item Reconstruct from the paths thus obtained the tree of the eta-long $\beta$-nf of $M$.
\end{enumerate}
\end{algorithmic}
\end{algorithm}

\subsubsection*{Correctness of STLC normalization}


We now state the Paths Characterization of simply-typed terms as an immediate consequence of the Game-Semantic Correspondence Theorem for STLC~\cite{BlumPhd}:
\begin{proposition}[Normalized Paths Characterization for STLC]
\label{prop:path_charact_stlc}
For every term $M$ inhabiting some simple type $T$, let $\etabetalnf{M}$ denote the eta-long beta-normal form of $M$, then we have:
\begin{equation*}
\travstlc(M)/{\sim}\ = \pathset(\etabetalnf{M})
\end{equation*}
\end{proposition}
\begin{proof}
 By soundness of Game Semantics, beta-eta equivalent terms have the same game denotation, therefore by Theorem~\ref{thm:gamesem_correspondence_stlc}, the set of traversals of the eta-long beta-normal form of $M$ is also given by the set of cores of traversals of $\etalf{M}$. The Path-View correspondence (Proposition~\ref{prop:pathview_ulc}) and Proposition~\ref{prop:ulc_and_stlc_trav_coincide} then show that the P-views of traversal cores of $\etalf{M}$ give the tree paths of the eta-long beta normal form.
\end{proof}

\begin{theorem}[Correctness of STLC normalization]
Algorithm~\ref{algo:stlc_normalization_by_traversals} terminates and returns the eta-long beta-normal form of the input term.
\end{theorem}
\begin{proof}
\emph{Soundness} A beta-normal term is uniquely determined by the set of maximal paths in its tree representation (Property~\ref{prop:tree_path_charact}). By Proposition~\ref{prop:path_charact_stlc}, this set corresponds precisely to the set of core P-views,
and by Proposition~\ref{prop:branching_traversal_simcomplete} its is also given by the core P-views of \emph{branching} traversals.
\emph{Termination} By Property~\ref{prop:etalong_trav_finite} for terms in eta-long form the set of branching traversals is finite and each traversal is itself finite, therefore the enumeration in Algorithm~\ref{algo:stlc_normalization_by_traversals} terminates.
\end{proof}

\paragraph{Implementation} The normalization procedure from algorithm~\ref{algo:stlc_normalization_by_traversals} was first implemented in the HOG program\cite{BlumGalop2008, BlumPhd,BlumHogTool2008}. The program takes as input any simply-typed lambda term and lets the user interactively generate all the traversals by ``playing the traversal game'' over the tree representation of the term. It also offers an  environment to calculate and perform operations over the traversals, including filtering, views and core projection.

\subsection{Arity threshold and normalizing traversals}

We showed in the previous section that for eta-long forms, there is a finite number of branching traversals, which in turn implies termination of the normalization procedure of Algorithm~\ref{algo:stlc_normalization_by_traversals}. For untyped terms, however, branching traversals can be infinite because of \emph{on-the-fly eta-expansion} in rule $\rulename{IVar}$: the justification label $i>0$ in the traversal rule is unbounded. In this section, we show that, for the purpose of term evaluation, $i$ can be bounded by a computable quantity called the \defname{arity threshold} of a traversal.

When traversing a lambda term, ghost nodes are introduced on-demand each time an eta-expansion is deemed necessary to lookup the argument of an application.
We should not need to eta-expand ad-infinitum as ghost nodes are only useful if they eventually lead to traversing some structural node of the tree. The intuition behind the \defname{arity threshold} is that after a sufficiently large number of eta-expansions, we are guaranteed to keep on traversing ghost nodes that never materialize back to structural nodes.

This section formalizes this intuition by defining the \emph{arity threshold} as the maximum number of times necessary to eta-expand a given subterm (using rule $\rulename{IVar}$)  in order to discover all set of paths of the beta-normal form of the term. In the next section we will show that such limit is sufficient to characterize the normal form of a lambda term when it exists.

\begin{definition}[Strand]
\label{ref:strand}
A \defname{strand} of a traversal $t$, is any sub-sequence of consecutive nodes from $t$,
with even length, starting with an external lambda node and ending with an external variable node, and such that all the occurrences in-between are internal.
\end{definition}

From the parity property of traversals, a strand consists of alternations of lambda nodes and variable/application nodes. It is convenient to highlight strands within a traversal by underlining the first and last occurrence of the strand. We will often write strands as follows:
$$ t = \cdots \underline{\alpha_k}\ n_k\ \alpha_{k-1}\ n_{k-1}\ \cdots\ \alpha_2\ n_2\ \alpha_1\ \underline{n_1} \cdots $$
for some $k\geq 1$ where for $j$ ranging from $k$ down to $1$, $\alpha_j$ is a lambda node in $\ImNodesLmd$, $n_j$ is a variable/application node in $\ImNodesVar$,
and the $2k-2$ nodes occurring  between $\alpha_k$ and $n_1$ are the internal nodes of the strand.

For any occurrence $n$ in $t$ of an external node, we call \defname{strand ending at $n$} the sequence of nodes ending at $n$ that constitutes a strand. It can be obtained by taking the longest subsequence of nodes preceding $n$ consisting only of internal nodes.

\begin{definition}[Traversal arity threshold]
\label{dfn:arity-threshold}
Let $t$ be traversal ending with an external variable node.
Let $\underline{\alpha_k}\ n_k\ \alpha_{k-1}\ n_{k-1}\ \cdots\ \alpha_2\ n_2\ \alpha_1\ \underline{n_1}$ be the strand of $t$ ending at $t^\omega$, for some $k>0$  where $\alpha_j \in \ImNodesLmd$, and $n_j \in \ImNodesVar$ for all $1\leq j\leq k$.
We define the \defname{arity threshold} of $t$ as:
\begin{align*}
\arth(t) &= \max_{q=1..k-1} \left( |n_q| + \sum_{j=1..q-1} (|n_j| - |\alpha_j|) \right)\ .
\end{align*}
\end{definition}

Or equivalently, using the abbreviation $ \delta_j = |n_j| - |\alpha_j|$:
\begin{align*}
    \arth(t) &= \max_{q=1..k-1} \left( |n_q| + \sum_{j=1..q-1} \delta_j \right)     = \max_{q=1..k-1} \left( |\alpha_q| + \sum_{j=1..q} \delta_j\right)
\end{align*}

\begin{remark}[Calculating the arity threshold]
The arity threshold can be rewritten as
$\arth(t) = \max_{q=1..k-1} b_q$
where $b_q = |\alpha_q| + \sum_{j=1..q} \delta_j = |n_1| - |\alpha_1| + |n_2| - |\alpha_2| + \cdots + |n_{q-1}| - |\alpha_{q-1}| + |n_q|$, for $1\leq q\leq k-1$ or equivalently defined by the induction
 $b_1 = |n_1|$, $b_{q+1} = b_q + |n_{q+1}| - |\alpha_q|$, for $1 \leq q \leq k-2$.
In other words, calculating the arity threshold boils down to  adding and subtracting the arity of the nodes starting from the last variable occurrence in the traversal strand and reading backwards until reaching the previous external variable. The maximal value calculated in that summation gives the arity threshold.
\end{remark}

Observe that traversal rule $\rulename{IVar}$ leaves infinitely many choices for the link label: any value greater than $1$. The following property shows that if the link label exceeds the arity threshold then the traversal ends up traversing only ghost nodes that never materialize back to a structural node.
\begin{proposition}[Weaving]
\label{prop:weaving}
Let $t \in \travsetbr$ be a branching traversal ending with an external variable node, and suppose its last strand is
$\underline{\alpha_k}\ n_k\ \alpha_{k-1}\ n_{k-1}\ \cdots\ \alpha_2\ n_2\ \alpha_1\ \underline{n_1}$ of length $2k$, for some $k\geq1$.
Let $t_{\max} \in \travsetbr$ be a maximal traversal extension of $t$. By rule $\rulename{IVar_\branching}$, the last occurrence in $t^\omega$ is necessarily immediately followed in $t_{\max}$ by an external lambda node justified by $t^\omega$ with label $i\geq 1$.


\begin{enumerate}[label=(\roman*)]
\item If $i>\arth(t)$ then $t^\omega$ is followed in $t_{\max}$ by a strand of length $2k$ consisting only of ghost nodes with labels defined as follows:
\begin{eqnarray*}
t_{\max} &=& \Pstr[0.3cm]{ \cdots
 (alpha){\underline{\alpha_k}} \ (nk){n_k} \
 (alphakm1){\alpha_{k-1}} \
 \cdots\
 (a2){\alpha_2}\ (n2){n_2}\ (a1){\alpha_1}\ (n1){\underline{t^\omega}}\
 (l1-n1,30:i_1){\underline{\ghostlmd}}\
 (t1-a1,40:i_2){\ghostvar} \ (l2-n2,41:i_2){\ghostlmd}\
  (t2-a2,45:i_3){\ghostvar} \cdots
(tkm1-alphakm1,40:i_{k}){\ghostvar} \ (lkm1-nk,43:i_{k}){\ghostlmd}\
 (tk-alpha,45:i_{k+1}){\underline{\ghostvar}} \cdots }\\
i_{q+1} &=& i + \sum_{j=1..q} (|\alpha_j| - |n_j|) = i - \sum_{j=1..q} \delta_j, \quad \mbox{ for $0\leq q\leq k$}
\end{eqnarray*}
where underlined occurrences indicate external nodes.

\item If $i>\arth(t)$ then \emph{all the nodes} following $t^\omega$ in $t_{\max}$ are ghost occurrences.

\item If $i\leq\arth(t)$ then for some $0 \leq r <k$ and link labels $i_q$, $0\leq q \leq r$ defined as in (i), we have:
\begin{eqnarray*}
t_{\max} &=& \Pstr[0.3cm]{
    \underline{\alpha_k}\ n_k \
    \cdots
    (alphar){\alpha_r} \ (nr){n_r} \
    (alpharm1){\alpha_{r-1}} \
 \cdots\
 (a2){\alpha_2}\ (n2){n_2}\ (a1){\alpha_1}\ (n1){\underline{t^\omega}}\
 (l1-n1,30:i_1){\underline{\ghostlmd}}\
 (t1-a1,40:i_2){\ghostvar} \ (l2-n2,41:i_2){\ghostlmd}\
  (t2-a2,45:i_3){\ghostvar} \cdots
(tkm1-alpharm1,40:i_{r}){\ghostvar} \ (lkm1-nr,43:i_{r}){\lambda\overline{x}}\
\cdots }
\end{eqnarray*}
where $\lambda\overline{x}$ is a structural internal lambda node, and consequently $i_r \leq |n_r|$.
\end{enumerate}
\end{proposition}
\proofatend
(i) By the alternation property of traversals, the last strand of $t$ consists of a succession of lambda nodes $\alpha_q$ and variable or $@$-nodes $n_q$, with indices going from $k$ down to $1$.

We show by finite induction on $1 \leq q \leq k$ that the first $2k$ nodes after $t^\omega$ are successive pairs of ghost lambda and ghost variable nodes justified in order by $n_1, \alpha_1, n_2, \alpha_2, n_3, \ldots, n_k, \alpha_k$ with respective labels $i_1, i_2, i_2, i_3, i_3, \ldots, i_k, i_k, i_{k+1}$ defined by:
\begin{equation*}
\left\{
    \begin{aligned}
    i_1 &= i \\
    i_{q+1} &= i_q + |\alpha_q| - |n_q| \hbox{, $1\leq q \leq k$}.
    \end{aligned}
\right.
\end{equation*}

\begin{itemize}
\item Base case $q=1$: Because $t^\omega$ is a ghost variable node hereditarily justified by the root, the only next rule that can be applied is \rulenamet{IVar_\branching}, and by assumption the following node is a ghost lambda node justified by $t^\omega$ with label $i_1 = i$. Then by rule \rulenamet{Lam^\ghostlmd_\branching} the next node is a ghost variable justified by $\alpha$ with label $i_2 = i + |\alpha_1| - |t^\omega|$.

$$ t_{\max} = \Pstr[0.3cm]{ \cdots (alpha){\alpha_1}\ (m){t^\omega}\ (l-m,30:i){\ghostlmd}\ (t-alpha,45:i_2){\ghostvar} \cdots} $$

\item Let $1<q< k$, applying the induction hypothesis on $q$ gives:
$$
t_{\max} = \Pstr[0.3cm]{
   \cdots
  (aq){\alpha_q}\
  (nq){n_q} \
  (aqm1){\alpha_{q-1}}\
  (nqm1){n_{q-1}}
 \cdots
 t^\omega\ \cdots
  (lk-nqm1,35:{i_{q-1}})\ghostlmd\
  (tq-aqm1,40:{i_{q}})\ghostvar\
   \cdots
}
$$

We then have $i_{q} > |n_q|$ since:
\begin{align*}
i_{q} &= i - \sum_{j=1..q-1} \delta_j
&\qquad\hbox{(by induction hypothesis)}\\
    &> \arth(t) - \sum_{j=1..q-1} \delta_j
&\qquad\hbox{(assumption $i> \arth(t)$)} \\
        &= \max_{r=1..k-1} \left( |n_r| + \sum_{j=1..r-1} \delta_j \right)
        - \sum_{j=1..q-1} \delta_j
& \qquad\hbox{(definition of $\arth$)} \\
    &\geq |n_q| + \sum_{j=1..q-1} \delta_j - \sum_{j=1..q-1} \delta_j
&\qquad\hbox{(take $r=q$)} \\
    &= |n_q|
\end{align*}

By definition of the strand, $\alpha_{q-1}$ is hereditarily justified by an $@$ node.
Since $i_q > |n_q|$, by rule \rulenamet{Var_\branching} the next node is necessarily a ghost lambda node justified by $n_q$ with label $i_q$. With rule \rulenamet{Lam^\ghostlmd_\branching} the following node is a ghost variable node $\ghostvar$ justified by $\alpha_q$ and labelled by $i_{q+1} = i_q + |\alpha_{q}| -|n_{q}|$.
\end{itemize}

(ii) By (i) the next strand following $t$ consists solely of ghost nodes and ends with a ghost variable $\ghostvar$ hereditarily justified by the root. Since ghost nodes have arity $0$, the arity threshold at that point is $0$. Hence, any label value $j$ chosen to extend the traversal at that point will be strictly greater than the arity threshold: thus by (i) the next strand consists solely of ghost nodes.
Applying the argument repeatedly shows that all the occurrences following $t$ are necessarily ghost nodes.

(iii) Take $r$ to be the smallest $q\geq 1$ such that $i_q\leq |n_q|$, it exists because:
\begin{align*}
i \leq \arth(t)
        &\iff i \leq \max_{q=1..k-1} \left( |n_q| + \sum_{j=1..q-1} \delta_j \right)\quad\hbox{(Definition of $\arth(t)$)} \\
        &\implies i \leq |n_r| + \sum_{j=1..r-1} \delta_j \quad \mbox{for some $1 \leq r \leq k-1$} \\
        &\iff i - \sum_{j=1..r-1} \delta_j \leq |n_r| \quad \mbox{for some $1 \leq r \leq k-1$} \\
        &\iff i_r \leq |n_r| \quad \mbox{for some $1 \leq r \leq k-1$} \quad \mbox{(Definition of $i_r$).}
\end{align*}
We then conclude by applying the same argument as (i) for all $q<r$.
\endproofatend

This weaving property suggests that instantiations of the external variable rule that exceeds the arity threshold are not relevant for normalization since they essentially expand external ghost variable nodes into deeper external ghost variable nodes. This observation gives rise to a stricter version of traversals:
\begin{definition}[Normalizing traversals]
\label{def:normalizing_traversals}
The set $\travsetnorm$ of \defname{normalizing traversals}, is the subset of branching traversals defined by the rules of Def.~\ref{dfn:branching_traversals} where the index $i$ in rule \rulenamet{IVar} is bounded by the arity threshold of the traversal.
Table~\ref{tab:normalizing_trav_rules} recapitulates the rules thus obtained  using judgement derivations, where the judgement notations `$\istraversal t$' means `$t\in\travsetnorm$'.
\end{definition}

\begin{remark}
    Observe that $\travsetnorm$ is \emph{not} $\sim$-complete. Indeed, take $M = \lambda x . x$ and the traversal
    $t  = \Pstr[0.5cm]{(l){\lambda x} \cdot (x-l){x} \cdot (gl-x){\ghostlmd} \cdot (gv-l){\ghostvar} }$
     in $\travsetbr$. We have $\pview{\core{t}} = t$ which is not equivalent to  $\pview{\core{u}}$ for any $u$ in $\travsetnorm = \{ \epsilon, \lambda x, \lambda x \cdot x \}$.
\end{remark}

\begin{table}
    \begin{ruletablebox}{Normalizing imaginary traversals of $\lambda$-calculus}
    \noindent {\bf PROGRAM - Structural rules}

    \infrule[$\rulefont{Root}_\normalizing$]
        {r \mbox{ is the root of } \ctree(M)}
        {\istraversal r}

    \infrule[$\rulefont{App}_\normalizing$]
        {\istraversal t \cdot @ \andalso @ \enables_0 \alpha }
        {\istraversal \Pstr[0.4cm]{t \cdot (at) @  \cdot (a-at,40:0) \alpha} \andalso (\alpha\in\NodesLmd)}

    \infrule[$\rulefont{Lam}^@_\normalizing$]
         {\istraversal t \cdot \alpha
         \andalso \alpha\in\NodesLmd
         \andalso \child(\alpha) \in \NodesApp }
         {\istraversal t \cdot \alpha \cdot \child(\alpha)\andalso \mbox{$\child(\alpha)$ has no pointer}}

    \infrule[$\rulefont{Lam}^{\sf var}_\normalizing$]
         {\istraversal u \cdot \beta \cdot v \cdot \alpha
         \quad \alpha\in\NodesLmd
         \quad \child(\alpha) \in \NodesVar
         \quad \beta \enables_i \child(\alpha), i \geq 1
         \quad \beta \mbox{ visible at } \alpha }
         {\istraversal \Pstr[0.5cm]{u \cdot (beta)\beta \cdot v \cdot \alpha \cdot (n-beta,30:i){\child(\alpha)
         }}}

    \infrule[$\rulefont{Lam^\ghostlmd_\normalizing}$]
        {
            \istraversal \Pstr[0.5cm]{ t \cdot(alpha){\alpha} \cdot (m){m} \ldots (gl-m,40:i){\ghostlmd}}
        }
        {
            \Pstr[0.5cm]{ \istraversal t \cdot(alpha){\alpha} \cdot
            (m){m}
             \ldots
            (gl-m,40:i){\ghostlmd}\cdot (al-alpha,40:{|\alpha| + i - |m|}){\ghostvar}
               }
            \andalso(\alpha\in\ImNodesLmd, m \in\ImNodesVar)
        }

    \emph{\bf PROGRAM - Copy-cat rules}

    \infrule[$\rulefont{Var}_\normalizing$]
        {
            \istraversal \Pstr[0.5cm]{t \cdot m \cdot (alpha){\alpha}
            \ldots (n-alpha,50:i){n}}
            \andalso n \in \ImNodesVar \setminus \ExtNodes
            \andalso m \enables_i \beta
            \andalso i>0
        }
        {
            \istraversal \Pstr[0.5cm]{ t \cdot
            (m){m} \cdot (a){\alpha}  \ldots (n-a,30:i){n}
                \cdot (letai-m,35:i){\beta}
                }
            \andalso (m \in \ImNodesVar\union\NodesApp, \alpha,\beta\in\ImNodesLmd)
        }

    \emph{\bf DATA - Input-variable rules}
    \infrule[$\rulefont{IVar_\normalizing}$]
         {\istraversal t \cdot n
          \andalso n \in \ImNodesVar\inter\ExtNodes
          \andalso n \enables_i \alpha
          \andalso 1\leq i \leq\arth(t)
         }
         {\istraversal\Pstr[0.5cm]{t \cdot (n){n} \cdot (l-n,30:i){\alpha}}
         \andalso(\alpha\in\NodesLmd)
         }

    where $t, u, v$ range over (sub-sequences of) justified sequences of nodes.

    \caption{Normalizing traversals $\travsetnorm$ of the untyped lambda calculus.
    The judgement rules define $\travsetnorm$ as a strict subset of imaginary traversals $\travulc$ from Table \ref{tab:trav_rules} where only rules $\rulefont{IVar_\normalizing}$ and $\rulefont{IVar}$ differ.}
    \label{tab:normalizing_trav_rules}
    \end{ruletablebox}
\end{table}

\subsection{Normalization procedure for ULC}

Algorithm~\ref{algo:ulc_normalization_by_traversals} describes the normalization procedure to compute the $\beta$-normal of an untyped lambda term if it exists. It generalizes Algorithm~\ref{algo:stlc_normalization_by_traversals} for STLC based on \emph{branching} traversals to ULC using \emph{normalizing} traversals.
We will prove correctness of the procedure in Section~\ref{sec:correctness_ulc_normalization}.

\begin{algorithm}[!ht]
\begin{algorithmic}
\caption{Normalization by traversals for the Untyped Lambda Calculus}
\label{algo:ulc_normalization_by_traversals}
\REQUIRE{Untyped term $M$ having a beta-normal form}
\ENSURE{Computation tree of the beta-normal form of $M$ modulo $\alpha$-conversion}
\begin{enumerate}[nosep]
  \item Calculate $\mathcal{P} \leftarrow \{ \pview{\core{t}} \ | \ t \in \travsetnorm \}$ as follows:
  \begin{itemize}[leftmargin=0.5em,nosep]
    \item Enumerate maximal \emph{normalizing} traversals $\travsetnorm(M)$ using the rules of Table~\ref{tab:normalizing_trav_rules};
    \item For each traversal, apply transformation $\coresymbol$ to get the traversal core;
    \item Take the P-view of the traversal core.
  \end{itemize}
  \item Apply Algorithm~\ref{algo:termtree_readout_from_justitied_paths} to reconstruct the tree of the $\beta$-nf of $M$ from $\mathcal{P} \equiv \pathset(\betanf{M})$.
\end{enumerate}
\end{algorithmic}
\end{algorithm}

\section{Examples of Normalization by Traversals}
This section walks you through the ULC normalization algorithm~\ref{algo:ulc_normalization_by_traversals} and traversal enumeration of Table~\ref{tab:normalizing_trav_rules} on several examples of lambda terms.

For each example we will enumerate the normalizing traversal of Table~\ref{tab:normalizing_trav_rules}. We write $t_\epsilon$  for the initial traversal up to the first non-deterministic choice. For every $s \in \nat^*$, we write $t_{s \cdot k}$ for the maximal traversal obtained after extending $t_s$ using a single application of the eta-expanded subcase of rule \rulenamet{IVar} with link-label $k \in \nat$. In other words, the subscript $s$ in $t_s$ records the list of non-deterministic choices made so far by rule \rulenamet{IVar}.

\begin{example}[Baby example]
  Take $M \alphaequiv (\lambda x. x x) (\lambda y. y)$. Repeatedly applying the structural traversal rules until reaching a non-deterministic choice yields traversal

  $t_\epsilon = \Pstr[0.7cm]{(n0){\lambda}\,(n1){@}\,(n2-n1,20){\lambda x}\,(n3-n2,20){x}\,(n4-n1,20){\lambda y}\,(n5-n4,20){y}\,(n6-n3,20){\lambda}\,(n7-n2,20){x}\,(n8-n1,20){\lambda y}\,(n9-n8,20){y}\,(n10-n7,20){\ghostlmd^1}\,(n11-n6,20){\ghostvar^1}\,(n12-n5,20){\ghostlmd^{1}}\,(n13-n4,20){\ghostvar^2}\,(n14-n3,20){\ghostlmd^2}\,(n15-n2,20){\ghostvar^2}\,(n16-n1,20){\ghostlmd^2}\,(n17-n0,20){\ghostvar^1}}$. Since the last node is an external variable (because it is justified by the root) rule $\rulename{IVar}$ could extend $t_\epsilon$ if considered as a pure \emph{imaginary} traversals. Since, the traversal's arity threshold is $\arth(t) = 0$, however, there is no more \emph{normalizing} traversal to explore: $t_\epsilon$ is maximal. The P-view core is $\pview{\core{t_\epsilon}} = \Pstr[0.7cm]{(n0){\lambda }\ (n17-n0){\ghostvar^1}}$ thus the beta-normal form of $M$ is, up to $\alpha$-conversion $\lambda y . y$.
\end{example}

\begin{example}[Church increment] Consider Church numerals $k \alphaequiv\lambda s z . s^k z$ for $k\geq0$.

\begin{tabular}{p{10cm}r}
Consider a term $M$ implementing the increment function so that $M k$ reduces to $k+1$ for all $k$. Such term can be defined as $M \alphaequiv {\sf add}\ 1$ where
${\sf add} \alphaequiv \lambda x y s z. x\, s (y\, s\, z)$.
The right figure represents the computation tree of $M$.
&
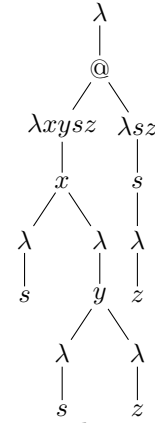
\begin{tikzpicture}[baseline=(root.base),level distance=5ex,inner ysep=0.5mm,sibling distance=10mm]
    \node (root)
    {$\lambda$}
    child {node{$@$}
        child{node{$\lambda x y s z$}
            child { node{$x$}
                child{node{$\lambda$}{
                    child {node {$s$}}}
                }
                child{node{$\lambda$}
                    child{node{$y$}
                        child{node{$\lambda$}
                            child{ node {$s$}}
                        }
                        child{node{$\lambda$}
                            child{node{$z$}}}
                    }
                }
            }
        }
        child{node{$\lambda s z$}
            child{node{$s$}
                child{node{$\lambda$} child{node{$z$}}}
            }
        }
    }
    ;
\end{tikzpicture}
\end{tabular}

\begin{itemize}[nosep]
\item The maximal normalizing traversal obtained without making any non-deterministic choice is:

$t_\epsilon = \Pstr[0.7cm]{(n0){\lambda }\ (n1){@}\ (n2-n1){\lambda x y s z}\ (n3-n2){x}\ (n4-n1){\lambda s z}\ (n5-n4,30){s}\ (n6-n3){\lambda }\ (n7-n2,32){s}\ (n8-n1,32){\ghostlmd^3}\ (n9-n0,32){\ghostvar^2} }$.
The core of $t_\epsilon$ is
$\core{t_\epsilon} = \Pstr[0.7cm]{(n0){\lambda} \cdot (n9-n0){{\ghostvar^2}} }$
and therefore $\pview{\core{t_\epsilon}} =  \lambda\cdot \ghostvar^2$ is a path in  $\betanf{M}$.
This means that $\betanf{M}$ must be of the form $\lambda y s \ldots \cdot y N_1 \ldots \ldots N_q$ for some fresh variable $y$ and $s$ and $q\geq0$.

\item To determine each argument $N_k$ from the normal form, we enumerate every possible non-deterministic choice corresponding to  eta-expansion of each possible argument for $k\geq 1$ using rule $\rulename{IVar}$. We then  maximally extend the traversal until the next non-deterministic choice. For $k=1$ this gives
$t_1 = \Pstr[0.7cm]{(n0){\lambda }\ (n1){@}\ (n2-n1){\lambda x y s z}\ (n3-n2){x}\ (n4-n1){\lambda s z}\
(n5-n4){s}\ (n6-n3){\lambda }\
(n7-n2){s}\ (n8-n1){\ghostlmd^3}\
(n9-n0){\ghostvar^2}
(n10-n9){\ghostlmd^1}
(n11-n8){\ghostvar^1}
(n12-n7){\ghostlmd^1}
(n13-n6){\ghostvar^1}
(n14-n5){\lambda^1}
(n15-n4)z
(n16-n3){\lambda^2}
(n17-n2)y
(n18-n1){\ghostlmd^2}
(n19-n0){\ghostvar^1}
}$.
Its core P-view is
$\pview{\core{t_1}} = \lambda \cdot \ghostvar^2 \cdot\ghostlmd^1
\cdot \ghostvar^1$
therefore the normal form is of the shape $\lambda y s \ldots \cdot s (y R_1 \ldots R_{q_2}) N_2 \ldots N_q$ for some terms $R_1$, \ldots $R_{q_2}$, and $q,q_2\geq 0$.

\item How many more $k$ do we need to look at? The answer: we need to keep iterating until the point where applying the traversal rules is guaranteed to only produce ghost variables and ghost lambda-nodes. Because there is a finite number of nodes in the computation tree, the variable node arities are bounded. Therefore, for high enough index $k$, the eta-expansion case from rule \rulenamet{Var} will never be met:
after applying rule \rulenamet{IVar} on $t_\epsilon$, all subsequent extensions of the traversal will be constructed using repeated applications of rule $\rulename{Lam^\ghostvar}$ or the eta-expanded case of rule $\rulename{Var}$.

The upper-bound $q$ for $k$ is precisely given by the \emph{arity threshold} of the traversal $t_\epsilon$ as defined in~\ref{dfn:arity-threshold}:
    \begin{align*}
     \arth(t_\epsilon)
     &= \max \{ |s^2| , 
                |s^1| + (|s^2| - |\lambda|) , 
                |x| +  (|s^1| - |\lambda s z|) + (|s^2| - |\lambda|)
               \} \\
    & = \max \{   0 , 1 + (0 - 0) ,  2 + (1 - 2) + (0 - 0)
            \} = 1
     \end{align*}
     Thus, we do not have to consider higher value of $k$ at $t_\epsilon$. Hence we have $\betanf{M} = \lambda y s \ldots \cdot s (y R_1 \ldots R_{q_2})$.

\item The arity threshold of $t_1$ is $\arth(t_1) = |z| + |y| - |\lambda^2| = 0+2-0 = 2$ hence  $\betanf{M}$ is of the form $\lambda y s \ldots \cdot s (y R_1 R_2)$.

\item Let's eta-expand using rule \rulenamet{IVar} for child index $k_2$ ranging from $1$ to $q_2 = 2$. For the case $k_2 = 1$, we obtain the traversal:

$t_{11} = \Pstr[0.7cm]{
(n0){\lambda }\
(n1){@}\ (n2-n1){\lambda x y s z}\ (n3-n2){x}\ (n4-n1){\lambda s z}\ (n5-n4){s}\ (n6-n3){\lambda }\ (n7-n2){s}\ (n8-n1){\ghostlmd^3}\ (n9-n0){\ghostvar^2}
(n10-n9){\ghostlmd^1}
(n11-n8){\ghostvar^1}
(n12-n7){\ghostlmd^1}
(n13-n6){\ghostvar^1}
(n14-n5){\lambda^1}
(n15-n4)z
(n16-n3){\lambda^2}
(n17-n2)y
(n18-n1){\ghostlmd^2}
(n19-n0){\ghostvar^1}
(n20-n19){\ghostlmd^1}
(n21-n18){\ghostvar^1}
(n22-n17,60:1)\lambda 
(n23-n2,45:3) s
(n24-n1,45:3) {\ghostlmd^3}
(n25-n0,48:2) {\ghostvar^2}
}$

Thus $\pview{\core{t_{11}}} =
\Pstr[0.7cm]{
(n0){\lambda }\
 (n9-n0){\ghostvar^2}
 (n10-n9){\ghostlmd^1}
(n19-n0){\ghostvar^1}
(n20-n19){\ghostlmd^1}
(n25-n0,48:2){\ghostvar^2}
}$. Hence $\betanf{M}$ is of the form $\lambda y s \ldots \cdot s\ (y\ s\ R_2)$.

\item Extending $t_1$ with \rulenamet{IVar^\lambda} for $k_2 = 2$ gives:

$t_{12} = \Pstr[0.7cm]{
(n0){\lambda }\
(n1){@}\ (n2-n1){\lambda x y s z}\
(n3-n2){x}\ (n4-n1){\lambda s z}\
(n5-n4){s}\
(n6-n3){\lambda }\
(n7-n2){s}\
(n8-n1){\ghostlmd^3}\
(n9-n0) {\ghostvar^2}
(n10-n9) {\ghostlmd^1}
(n11-n8){\ghostvar^1}
(n12-n7){\ghostlmd^1}
(n13-n6){\ghostvar^1}
(n14-n5)\lambda^1
(n15-n4)z
(n16-n3){\lambda^2}
(n17-n2)y
(n18-n1){\ghostlmd^2}
(n19-n0){\ghostvar^1}
(n20-n19){\ghostlmd^2} 
(n21-n18){\ghostvar^2}
(n22-n17,60:2)\lambda 
(n23-n2,45:4) z
(n24-n1,45:4) {\ghostlmd^4}
(n25-n0,48:3) {\ghostvar^3}
}$

Thus $\pview{\core{t_{12}}} =
\Pstr[0.7cm]{
(n0){\lambda }\
 (n9-n0){{\ghostvar^2}}
 (n10-n9){\ghostlmd^1}
(n19-n0){\ghostvar^1}
(n20-n19){\ghostlmd^2}
(n25-n0,48:3) {\ghostvar^3}
}$.
\end{itemize}
Hence $\betanf{M} = \lambda y s z \cdot s\ (y\ s\ z)$.
\end{example}

\begin{example}[``Missing operand'' example by Neil Jones]
\label{ex:missingoperand}
This small example illustrates how on-the-fly eta-expansion helps resolve the ``missing argument'' problem faced when using the traversal rules of STLC. Take $M = (\lambda u . u\ (y_1\ u)) (\lambda v . v\ y_2)$.
The computation tree and its only two maximal normalizing traversals are shown below.

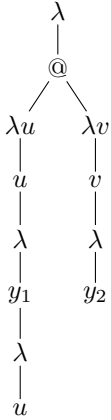
\begin{tabular}{lp{12cm}}
\begin{tikzpicture}[baseline=(root.base),level distance=5ex,inner ysep=0.5mm,sibling distance=10mm]
\node (root)
{$\lambda$}
child {node{$@$}
        child{node{$\lambda u $}
           child {node {$u$}
              child {node {$\lambda$}
                  child {node {$y_1$}
                      child {node {$\lambda$}
                          child {node {$u$}
                          }
                      }
                  }
              }
           }
        }
        child{node{$\lambda v$}
            child{node{$v$}
                child{node{$\lambda$}
                    child{ node {$y_2$}}
                }
            }
        }
    }
;
\end{tikzpicture}
&
With STLC traversals, one can traverse $t_1$ all the way to variable $v$ at which point one gets stuck due to the lack of operand applied to the last occurrence of $u$ (the occurrence  at the bottom of the left branch in the tree).
$$t_1 = \Pstr[0.7cm]{(n0){\lambda }\,(n1){@}\,(n2-n1){\lambda u}\,(n3-n2){u}\,(n4-n1){\lambda v}\,(n5-n4){v}\,(n6-n3){\lambda }\,(n7-n0){y_1}\,(n8-n7){\lambda }\,(n9-n2){u}\,(n10-n1){\lambda v}\,(n11-n10){v}\,(n12-n9){{\ghostlmd^{1}}}\,(n13-n8){{\ghostvar^{1}}}\,(n14-n13){{\ghostlmd^{1}}}\,(n15-n12){{\ghostvar^{1}}}\,(n16-n11){\lambda }\,(n17-n0,35){y_2}\,}
$$
Using imaginary traversals, that operands gets created on-the-fly through eta-expansion and is represented by ghost lambda node $\ghostlmd^1$.
The other maximal traversal is$$
t_2 = \Pstr[0.7cm]{(n0){\lambda }\,(n1){@}\,(n2-n1){\lambda u}\,(n3-n2){u}\,(n4-n1){\lambda v}\,(n5-n4){v}\,(n6-n3){\lambda }\,(n7-n0,30){y_1}\,(n8-n7){{\ghostlmd^{2}}}\,(n9-n6){{\ghostvar^{1}}}\,(n10-n5){\lambda }\,(n11-n0,30){y_2}\,}
$$
The two core P-views give the two maximal paths in the beta-normal form:
$\pview{\core{t_1}} = \lambda \ y_1\ \lambda \ \ghostvar^{1}
\ \ghostlmd^{1}\ y2$ and  $\pview{\core{t_2}} = \lambda\ y_1\ \ghostlmd^{2}\ y_2$.
The beta-normal form of $M$ is thus $y_1 (\lambda  z.z\ y_2) y_2$.
\end{tabular}
\end{example}

\begin{example}
    Take the fixed-point combinator $\Omega = (\lambda x. x x) (\lambda y. y y)$.

    \begin{tabular}{lp{12cm}}
    \begin{tikzpicture}[baseline=(root.base),level distance=5ex,inner ysep=0.5mm,sibling distance=10mm]
        \node (root)
        {$\lambda$}
        child {node{$@$}
                child{node{$\lambda x $}
                   child {node {$x_1$}
                      child {node {$\lambda$}
                          child {node {$x_2$}
                          }
                      }
                   }
                }
                child{node{$\lambda y$}
                    child{node{$y_1$}
                        child{node{$\lambda$}
                            child{ node {$y_2$}}
                        }
                    }
                }
            }
        ;
    \end{tikzpicture}
    &
    It's a well-known fact that $\Omega$ does not have a normal form. The term has a unique (infinite) traversal represented below, which starts with
    $\Pstr[0.7cm]{{\lambda}\ (n1){@}\ (n2-n1){\lambda x}\ (n3-n2){x_1}\ (n4-n1){\lambda y}\ (n5-n4){y_1}\ (n6-n3){\lambda }\ (n7-n2){x_2}\ (n8-n1){\lambda y}\ (n9-n8){y_1}\ldots }$ and then follows a repeated pattern with an increasing number of ghost nodes between occurrences of $y_1$ and $y_2$:
    \end{tabular}

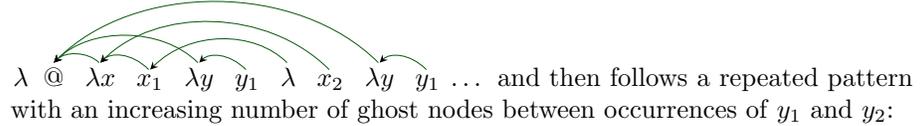
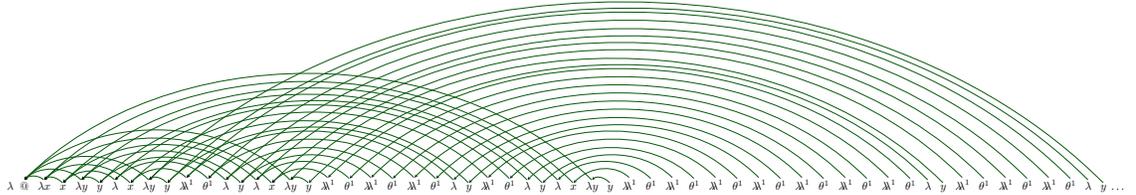
    \resizebox{1\hsize}{!}{
    $$
    \Pstr[0.7cm]{{\lambda }\ (n1){@}\ (n2-n1){\lambda x}\ (n3-n2){x}\ (n4-n1){\lambda y}\ (n5-n4){y}\ (n6-n3){\lambda }\ (n7-n2){x}\ (n8-n1){\lambda y}\ (n9-n8){y}\ (n10-n7){{\ghostlmd^{1}}}\ (n11-n6){{\ghostvar^{1}}}\ (n12-n5){\lambda }\ (n13-n4){y}\ (n14-n3){\lambda }\ (n15-n2){x}\ (n16-n1){\lambda y}\ (n17-n16){y}\ (n18-n15){{\ghostlmd^{1}}}\ (n19-n14){{\ghostvar^{1}}}\ (n20-n13){{\ghostlmd^{1}}}\ (n21-n12){{\ghostvar^{1}}}\ (n22-n11){{\ghostlmd^{1}}}\ (n23-n10){{\ghostvar^{1}}}\ (n24-n9){\lambda }\ (n25-n8){y}\ (n26-n7){{\ghostlmd^{1}}}\ (n27-n6){{\ghostvar^{1}}}\ (n28-n5){\lambda }\ (n29-n4){y}\ (n30-n3){\lambda }\ (n31-n2){x}\ (n32-n1){\lambda y}\ (n33-n32){y}\ (n34-n31){{\ghostlmd^{1}}}\ (n35-n30){{\ghostvar^{1}}}\ (n36-n29){{\ghostlmd^{1}}}\ (n37-n28){{\ghostvar^{1}}}\ (n38-n27){{\ghostlmd^{1}}}\ (n39-n26){{\ghostvar^{1}}}\ (n40-n25){{\ghostlmd^{1}}}\ (n41-n24){{\ghostvar^{1}}}\ (n42-n23){{\ghostlmd^{1}}}\ (n43-n22){{\ghostvar^{1}}}\ (n44-n21){{\ghostlmd^{1}}}\ (n45-n20){{\ghostvar^{1}}}\ (n46-n19){{\ghostlmd^{1}}}\ (n47-n18){{\ghostvar^{1}}}\ (n48-n17){\lambda }\ (n49-n16){y}\ (n50-n15){{\ghostlmd^{1}}}\ (n51-n14){{\ghostvar^{1}}}\ (n52-n13){{\ghostlmd^{1}}}\ (n53-n12){{\ghostvar^{1}}}\ (n54-n11){{\ghostlmd^{1}}}\ (n55-n10){{\ghostvar^{1}}}\ (n56-n9){\lambda }\ (n57-n8){y}\ldots }
    $$
    }

    Consider each section of the traversal starting with $y_1$ and ending with $y_2$ with no other variable occurrence in-between, and define its length as the number of ghost nodes in-between. Then the lengths of successive such traversal sections have the following pattern: $2,6,14,30,62\ldots$
    One can show that the length of the $i$th section is precisely given by $2\times(2^i-1)$, for $i\geq1$.
    So as the term is being evaluated, it takes exponentially longer to determine that $y_2$ is the argument of $y_1$.
    The infinite traversal $t$ can be expressed as follows (omitting justification pointers):
    \begin{equation*}
    t_\Omega = \lambda \cdot @ \cdot \lambda x \cdot  x_1 \cdot \lambda y \cdot y_1
    \cdot  \sum_{i\geq 1}
        \left[
            \lambda \cdot x_2 \cdot \lambda y \cdot y_1 \cdot
             {(\ghostlmd\ \ghostvar)}^{2^i -1}\cdot \lambda \cdot y_2
                 \sum_{i\leq j<i}
                    \left(
                        {(\ghostlmd\ \ghostvar)}^{2^j -1}
                        \cdot
                        \lambda \cdot y_2
                    \right)
        \right]
    \end{equation*}
    where $\sum$ represents sequence concatenation and $u^i$ represents the concatenation of $i$ copies of $u$ for $i\geq 0$.
    Represented more succinctly by keeping only variable nodes and omitting redundant lambda and @-nodes:
    $t_\Omega \filter \NodesVar =  x_1 \cdot y_1
        \cdot  \sum_{i\geq 1}
            \left[
                x_2 \cdot y_1 \cdot
                 \ghostvar^{2^i -1}\cdot y_2
                     \sum_{1\leq j <i} (\ghostvar^{2^j -1} \cdot y_2)
            \right]
    $.
    Without ghost nodes, $\Omega$'s traversal semantics can thus be summarized as
    $
        t_\Omega \filter (\NodesVar\inter\ExtNodes) =  x_1 \cdot y_1
        \cdot  \sum_{i\geq 1} x_2 \cdot y_1 \cdot {y_2}^i \ .
    $
\end{example}

\begin{example}[Neil Jones' varity example]
    \label{example:varity}
Take $varity\ two$ where
$two \alphaequiv \lambda s_2 z_2 . s_2 (s_2\ z_2)$ and
$varity \alphaequiv \lambda t.t (\lambda n a x . n (\lambda s z . a s (x s z))) (\lambda a. a) (\lambda z_0 . z_0)$.
The initial traversal $t_\epsilon$ up to the first non-deterministic choice is:

\resizebox{1\hsize}{!}{$
t_\epsilon = \Pstr[0.7cm]{(n0){\lambda }\ (n1){@}\ (n2-n1){\lambda t}\ (n3-n2){t}\ (n4-n1){\lambda s_2 z_2}\ (n5-n4){s_2}\ (n6-n3){\lambda n a x}\ (n7-n6){n}\ (n8-n5){\lambda }\ (n9-n4){s_2}\ (n10-n3){\lambda n a x}\ (n11-n10){n}\ (n12-n9){\lambda }\ (n13-n4){z_2}\ (n14-n3){\lambda a}\ (n15-n14){a}\ (n16-n13){{\ghostlmd^{1}}}\ (n17-n12){{\ghostvar^{1}}}\ (n18-n11){\lambda s z}\ (n19-n10){a}\ (n20-n9){{\ghostlmd^{2}}}\ (n21-n8){{\ghostvar^{1}}}\ (n22-n7){\lambda s z}\ (n23-n6){a}\ (n24-n5){{\ghostlmd^{2}}}\ (n25-n4){{\ghostvar^{3}}}\
(n26-n3){\lambda z_0}\ (n27-n26){z_0}\ (n28-n25){{\ghostlmd^{1}}}\ (n29-n24){{\ghostvar^{1}}}\
(n30-n23){\lambda }\ (n31-n22){s}\ (n32-n21){{\ghostlmd^{1}}}\ (n33-n20){{\ghostvar^{1}}}\ (n34-n19){\lambda }\ (n35-n18){s}\ (n36-n17){{\ghostlmd^{1}}}\ (n37-n16){{\ghostvar^{1}}}\ (n38-n15){{\ghostlmd^{1}}}\ (n39-n14){{\ghostvar^{2}}}\ (n40-n13){{\ghostlmd^{2}}}\ (n41-n12){{\ghostvar^{2}}}\ (n42-n11){{\ghostlmd^{2}}}\ (n43-n10){{\ghostvar^{4}}}\ (n44-n9){{\ghostlmd^{4}}}\ (n45-n8){{\ghostvar^{3}}}\ (n46-n7){{\ghostlmd^{3}}}\ (n47-n6){{\ghostvar^{5}}}\ (n48-n5){{\ghostlmd^{5}}}\ (n49-n4){{\ghostvar^{6}}}\ (n50-n3){{\ghostlmd^{6}}}\ (n51-n2){{\ghostvar^{4}}}\ (n52-n1){{\ghostlmd^{4}}}\ (n53-n0){{\ghostvar^{3}}}}$}

Enumerating normalizing traversals yields the set $\{t_\epsilon, t_{11}, t_{12}, t_{121}, t_{122} \}$. Keeping only maximal traversals gives $\travsetnorm_{\max}(M) = \{ t_{11}, t_{121}, t_{122} \}$ from Table~\ref{tab:varity2_trav} in appendix. The set of maximal paths in the beta-normal form is thus given by the core P-views
$\pview{\core{t_{11}}} =
        \lambda \ \ghostvar^{3}\ \ghostlmd^{1}\
        \ghostvar^{1}\ \ghostlmd^{1}\ \ghostvar^{3}
$,
$\pview{\core{t_{121}}} =
        \lambda \ \ghostvar^{3}\ \ghostlmd^{1}\ \ghostvar^{1}\ \ghostlmd^{2}\ \ghostvar^{2}\ \ghostlmd^{1}\ \ghostvar^{3}$
and
$\pview{\core{t_{122}}} =
    \lambda \ \ghostvar^{3}\ \ghostlmd^{1}\ \ghostvar^{1}\ \ghostlmd^{2}\ \ghostvar^{2}\ \ghostlmd^{2}\ \ghostvar^{4}$
Therefore the beta-normal from of $varity\ 2$ is
$\lambda x_1 x_2 x_3 x_4 . x_3 (x_1 x_3 (x_2 x_3 x_4)) \alphaequiv \lambda x y s z . s (x s (y s z))
$.
\end{example}


\section{Leftmost linear reduction}
\label{sec:leftmostlinearred}

\subsection{Lambda calculus background and notations}
We recall standard results of the lambda calculus.
A \defname{redex} is a sub-term of the form $(\lambda x. M) N$.
Reducing redex $(\lambda x. M) N$, or also \emph{firing} the redex, means substituting all free occurrences of $x$ in $M$ by the term $N$ using capture avoiding substitution (the bound variable is renamed afresh when recursively substituting under a lambda).
A term is said to be in \defname{normal form} if it does not contain any redex.
A term is in \defname{head normal form} if it can be written $\lambda x_1 \ldots x_n . y A_1 \ldots A_m$ for $n,m\geq0$. If a term is not in head normal form then its \defname{head beta-redex} is the left-most redex, otherwise the term does not have any head beta-redex. For any reduction relation $\rightarrow$ between terms we will write
$\rightarrow^*$ to denote its reflexive transitive closure.

The \defname{head reduction}, denoted $\rightarrow_{h}$, fires the head redex of a term if it exists. It can be shown that $\rightarrow^*_{h}$ yields the head-normal form. The \defname{normal-order reduction strategy}, also called leftmost-outermost reduction strategy, performs head reduction until reaching the head-normal form and then recursively applies head reduction on each operand of the head variable. A standard result is that this reduction strategy yields the normal form if it exists.

\subsection{Head-linear reduction}

In the lambda calculus, a redex is necessarily formed by the outermost lambda in operator position of an application: if the operator consists of consecutive lambda abstractions (\eg, as in $(\lambda x \lambda y . M) A_1 A_2$) then the outermost lambda (\eg, $\lambda x$) is the one that will form the redex. The notion of redex can be generalized to allow evaluation of arguments in any order. In particular, one can allow any of the consecutive $\lambda$-abstractions to form a redex (\eg, the abstraction $\lambda y$ and corresponding argument $A_2$ would be a valid redex). This is formalized by the notion of \emph{generalized redex}--a generalization of the notion of \emph{prime redex} from \cite{danos-head}:
\begin{definition}[Generalized redex]
\label{dfn:generalized_redex}
The set of generalized redexes of a term $M$, written $gr(M)$, is a set of pairs $(\lambda x, A)$ where $\lambda x$ is some abstraction in $M$ and $A$ is a subterm called the argument of $\lambda x$. The head $\lambda$ list of $M$, written $\lambda_l(M)$ is a list of lambda abstractions of $M$. They are defined by induction:
\begin{align*}
gr(v) &= \emptyset & \lambda_l(v) &= \epsilon\\
gr(\lambda x. U) &= gr(U) & \lambda_l(\lambda x. U) &= \lambda x \cdot \lambda_l(U) \\
gr(U V) &= \{ (\lambda x, V) \} \union gr(U) \union gr(V) &
\lambda_l(U V) &= l & \mbox{if $\lambda_l(U) = \lambda x \cdot l$} \\
gr(U V) &= gr(U) \union gr(V) & \lambda_l(U V) &= \epsilon & \mbox{if $\lambda_l(U) = \epsilon$.}
\end{align*}
where $v$ ranges over variable occurrences, $x$ ranges over variable names, $U, V$ range over subterms of $M$, and $\epsilon$ denotes the empty list.
\end{definition}

\begin{example} For any term $M, N, A_1, A_2$ we have
$gr((\lambda x \lambda y . M) A_1 A_2) = \{ (\lambda x, A_1), (\lambda y, A_2)\}$ and
 $gr((\lambda z . (\lambda x \lambda y . M) N) A_1 A_2) = \{ (\lambda z, A_1), (\lambda x, N), (\lambda y, A_2)\}$.
\end{example}

To define head-linear reduction, one needs to consider specific \emph{occurrences} of variables and sub-terms in a given term. In particular, let's emphasize that a generalized redex  $(\lambda x, A)$ refers to specific \emph{occurrence} of $\lambda x$ and subterm $A$.

We say that a variable occurrence is \defname{involved} in the generalized redex $(\lambda x, A)$ if the variable occurrence is bound by $\lambda x$. A variable occurrence can therefore be involved in at most one generalized redex. We define the \defname{linear substitution} of $x$ for $A$ as the term obtained by performing capture-avoiding substitution of that single occurrence of $x$ by $A$. (Compare this to the standard substitution that applies to \emph{every} occurrence of $x$ in $M$.) When performing such substitution we say that we \defname{linearly fire} the generalized redex $(\lambda x, A)$ for that occurrence of $x$.

The \defname{head variable occurrence} (abbreviated \emph{hoc}) of a term is the left-most variable occurring in the term (\ie, the first variable found by depth-first traversal of the term tree.) If the head variable occurrence of a term is involved in a generalized redex then we call that redex the \defname{head-linear redex}.
A term that does not have a head-linear redex is said to be in \defname{quasi head-normal form}.
The \defname{head-linear reduction} $\hlred$ is defined as the reduction that linearly fires the head linear redex, if it exists. It can be shown that the reflexive transitive closure $\rightarrow^*_{hl}$ yields the quasi head-normal form \cite{danos-head,danosherbelinregnier1996}.

\paragraph{Soundness}
The set of \defname{spine subterms} is defined by induction: a term is a spine subterm of itself; the spine subterms of $U V$ and $\lambda x. U$ are those of $U$.
A \defname{prime redex} (ala Danos-Regnier) is a generalized redex $(\lambda x, A)$ such that the operator of the redex ($\lambda x . U$ for some $U$) is a spine subterm. Danos-Regnier showed the following result (Theorem~2 in \cite{danos-head}):
\begin{theorem}[Soundness and completeness of head-linear reduction~\cite{danos-head}] \
\label{thm:danosreigner_headlinred}
\begin{itemize}[nosep]
\item If $T \rightarrow^*_{hl} T$  then $T$ and $T'$ are $\beta$-equivalent.
\item If $T$ is in quasi-head normal form and has $n$ prime redexes then the head reduction of $T$ leads to a head normal form in exactly $n$ steps.
\item If $T$ is any term, the head linear reduction of $T$ terminates iff the head reduction of $T$ terminates.
\end{itemize}
\end{theorem}

\subsection{Leftmost linear reduction}

As the name indicates, the head-linear reduction is a linear version of the \emph{head} reduction. It thus only yields the quasi \emph{head}-normal form, not the normal form, and therefore is not complete for normalization.

In the standard lambda calculus, normal-order reduction strategy is obtained by repeatedly applying head reduction to get to the head-normal form, and then continuously applying the head-reduction on each argument of the head variable.
This ultimately yields the normal form of the term if it exists.
We now define the linear counterpart of the normal-order reduction: Informally, the \defname{leftmost linear reduction strategy} is the strategy that performs head-linear reduction if possible, and otherwise, if the term is in quasi-head normal form, continuously (and recursively) applies the head-linear reduction on each argument (from left to right) of the head variable occurrence.
Formally:

\begin{definition}[Leftmost linear reduction]
    \label{def:leftmostlinearreduction}
Given a term $M$, we write $\VarOcc(M)$ for the set of variable occurrences in $M$
and $\VarOcc^\bot(M)$ for $\{\bot \} + \VarOcc(M)$ where the bottom element $\bot$ represents the `undefined' occurrence. We introduce the partial order $\sqsubseteq$ on $\VarOcc^\bot(M)$ defined by: for all $x,y \in \VarOcc^\bot(M)$, $x \sqsubseteq y$ if and only if $x = y$ or $x = \bot$. We define the partial function $lloc_M$ from subterms of $M$ to variable occurrences by induction on the subterms of $M$:
\begin{align*}
lloc_M(v) &=
    \begin{cases}
    v &\mbox{if $v$ is involved in a generalized redex in $M$,} \\
    \bot & \mbox {otherwise.}
    \end{cases}  \\
lloc_M(\lambda x . U) &= lloc_M(U) \\
lloc_M(U V) &= \begin{cases}
                lloc_M(U) &\mbox{if $lloc_M(U)\neq\bot$,} \\
                lloc_M(V) &\mbox{if $lloc_M(U)=\bot$ and $\lambda_l(U) = \epsilon$,} \\
                \bot & \mbox{if $lloc_M(U)=\bot$ and $\lambda_l(U) \neq \epsilon$.}
              \end{cases}
\end{align*}
where $v$ ranges over variable occurrences in $M$,
and for any subterm $N$, $lloc_M(N) = \bot$ denotes that $lloc$ is undefined at $N$.

The \defname{leftmost linear variable occurrence} of $M$
is defined as $lloc_M(M)$ if it exists and the generalized redex it is involved in is called the \defname{lloc redex}, otherwise $M$ is said to be in \defname{quasi normal form}, or \emph{qnf} for short.
The \defname{leftmost linear reduction} strategy, written $\llred$, is defined as the strategy that linearly fires the generalized redex involving the leftmost linear variable occurrence.
\end{definition}

Leftmost linear reduction proceeds by first locating the left-most variable occurrence and, if it is involved in a redex, linearly fires it. In comparison, the standard left-most outermost reduction strategy first locates the leftmost redex and then fires it by substituting all the variables involved in it.

Since the \emph{lloc} is uniquely defined, a given term has a unique left-most linear reduction sequence $M \llred M_1 \llred M_2 \llred \ldots$ (also abbreviated as \defname{linear reduction sequence}).

\begin{example}
Take $M = (\lambda x. x x N) z$. We have $M \hlred (\lambda x. z x N) z$ which is in quasi head normal form because the head variable $z$ is not involved in any generalized redex.
But since the left-most occurrence $x$ is involved in the redex $(\lambda x, z)$ the leftmost linear reduction gives $(\lambda x. z x N) z \llred (\lambda x. z z N) z$.
\end{example}

The \emph{lloc} generalizes the notion of \emph{hoc} to arguments of the head variable. Observe, however, that contrary to the \emph{hoc}, the definition requires the \emph{lloc} to be involved in a generalized redex.
\begin{example}In the term $M = (\lambda x y . z y x) (\lambda u . u) (\lambda v . v)$, the \emph{hoc} variable $z$ is not involved in any generalized redex, whereas the \emph{lloc} variable $y$ is involved in generalized redex $(\lambda y, (\lambda v.v))$.
\end{example}
On the other hand, if the \emph{hoc} is involved in a generalized redex then by definition it must coincide with the \emph{lloc}.

\begin{property}
    If $\lambda_l(M) = \langle \lambda x_1 \ldots x_n \rangle$ and $M$ has a $\beta$-normal form then its beta-normal form starts with at least $n$ lambda abstractions, \ie, up to variable renaming it is of the shape $\lambda x_1 \ldots x_n . M'$ for some beta-normal term $M'$.
\end{property}
\begin{proof}
    Trivial induction.
\end{proof}

\begin{property}
\label{prop:qnf_longapply}
    Let $M$ be an untyped term and $y A_1 \ldots A_n$ be a subterm of $M$ for $n\geq1$ where $y$ is a variable not involved in any generalized redex then:
    \begin{align*}
    lloc_M(y A_1 \ldots A_n) &=
        \begin{cases}
         lloc_M (A_j) &\mbox{ where } j = \min \{ 1\leq i\leq n \ | \ lloc_M (A_i) \neq \bot\} \\
         \bot &\mbox{ if } lloc_M(A_i) = \bot \mbox{ for all } 1\leq i\leq n.
        \end{cases}
    \end{align*}
\end{property}
\begin{proof}
    Immediate from the definition of $lloc$ since $\lambda_l(y) = \epsilon$.
\end{proof}

We say that a standard redex $(\lambda x. N) A$ is \defname{trivial} if $x$ does not occur freely in $N$ in which case the redex can be \emph{trivially reduced} to just $N$.

\begin{property}
\label{prop:lloc_properties}
Some immediate properties:
\begin{enumerate}[noitemsep,label=(\roman*)]
\item
 If $N$ is a subterm of $M$ then $lloc_N(N) \sqsubseteq lloc_M(N)$;
\item
 If $M$ is beta-normal then $lloc_M(M) = \bot$;
 \item
 If $U V$ is in quasi normal form then so is $U$;
\item
 If $\lambda x . U$ is in quasi normal form then so is $U$;
 \item If $M$ has a \emph{lloc} redex then $M$ has a redex;
 \item If all the redexes in $M$ are trivial then $M$ is in \emph{qnf}.
\end{enumerate}
\end{property}
\begin{proof}
(i) Immediate from the fact that generalized redexes of $N$ are also generalized redexes of $M$. (ii) If $lloc_M(M) \neq \bot$ then $M$ must contain a generalized redex and therefore must also contain a standard redex hence $M$ is not beta-normal.
(iii) If $U V$ is in quasi normal form then $lloc_{UV}(UV) = \bot$ so by definition of $lloc$ we must have $lloc_{UV}(U) = \bot$. By (i) since $U$ is a subterm of $UV$ this implies $lloc_{UV}(U) = \bot$.
(iv) We have $\bot = lloc_{\lambda x . U}(\lambda x . U) = lloc_{\lambda x . U}(U) \sqsupseteq lloc_{U}(U)$.
(v) and (vi) are immediate from the definition of \emph{lloc} redex.
\end{proof}

The following example shows that the converse of (vi) does not hold:
\begin{example}
Take $M= (\lambda z.u)((\lambda x . x) y)$. Then $lloc_M(M) = \bot$ since $\lambda_l(\lambda z.y)\ne\epsilon$ and $lloc_M(\lambda z.u) = \bot$.
So $M$ is in quasi-normal form, even though it has a non-trivial redex in the argument $(\lambda x . x) y$.
\end{example}

\begin{definition}[Arguments of the hoc]
Let $z$ be the \emph{hoc} of $M$. The arguments of the \emph{hoc} occurrence, written ${\sf Args}(M)$, is a list of subterms of $M$ defined by induction:
\begin{align*}
{\sf Args}(z) &= \epsilon \\
{\sf Args}(\lambda x . U ) &= {\sf Args}(U) \\
{\sf Args}(U V) &= \begin{cases}
             {\sf Args}(U) \cdot V & \mbox{if } \lambda_l(U) = \epsilon\\
             {\sf Args}(U)         & \mbox{if } \lambda_l(U) \neq \epsilon
            \end{cases}
\end{align*}
Note that the variable case where $x\neq z$ is not needed since the induction proceeds by spinal descent to the head occurrence.
\end{definition}

\begin{property}
\label{property:hoc_argument_count_in_headnf}
    If $M$ is in quasi-head-normal form, then the head variable of its head-normal form has precisely $|{\sf Args}(M)|$ arguments.
\end{property}
\begin{proof}
By Theorem~\ref{thm:danosreigner_headlinred}, the \emph{hoc} of a quasi-head normal form becomes the head variable of the head-normal form after reducing the prime redexes, and it is easy to show by structural induction that standard beta-reduction preserves the number of pending hoc arguments.
\end{proof}

We now introduce a $\beta$-reduction strategy (in the standard sense).
We define the \defname{leftmost spinal-innermost} standard $\beta$-redex of $M$ as the $\beta$-redex $(\lambda x . N) A$ such that there is no other redex in $M$ occurring left of it and such that $N$ contains no other redex; the argument $A$ on the other-hand, may contain other redexes.
The idea consists in reducing $M$'s redexes in `inside-out left-to-right' order.
This definition should be distinguished from the \emph{leftmost innermost} reduction strategy from the lambda calculus literature.
In particular $(\lambda x . y) \Omega$ reduces to $y$ with the leftmost spinal-innermost strategy, while the leftmost innermost strategy does not terminate.


\begin{proposition}[Trivial reduction of quasi normal forms]
\label{prop:qnf_nf}
Let $M$ be a term in \emph{qnf}.
\begin{itemize}[nosep]
\item[(i)]
 If $M$ is not in beta-normal form then its \emph{leftmost spinal-innermost standard $\beta$-redex} $(\lambda x . N) A$ is trivial.
(\ie, $x$ does not occur freely in $N$).

\item[(ii)] $M$ has a beta-normal form and the leftmost spinal-innermost strategy normalizes $M$.

\end{itemize}
\end{proposition}
\proofatend
\begin{compactitem}

\item[(i)]
By induction on $M$. Suppose $M = (\lambda x . N) A$. By assumption $N$ is in beta-normal form therefore if $x$ occurs freely in $N$ it is involved in the generalized redex $(\lambda x, A)$ of $M$ and by bottom-up evaluation of $lloc$ this gives $lloc_M(M) = x$ which contradict the assumption that $M$ is in \emph{qnf}.

Suppose $M = \lambda z . U$ then $U$ is in \emph{qnf} by Property~\ref{prop:lloc_properties}(iv) and the result follows by the induction hypothesis
since the redexes of $M$ are those of $U$ and the free variables of $M$ are a subset of those of $U$.

Suppose $M = U V$. We consider the three sub-cases depending on where the redex is located within $U V$:
\begin{compactitem}
  \item[(1)] The redex $(\lambda x . N) A$  is in $U$. Since $M$ is in \emph{qnf}, by Property~\ref{prop:lloc_properties}(iii) so is $U$, we can thus conclude by the induction hypothesis.

  \item[(2)] The redex $(\lambda x . N) A$  is in $V$.
If $\lambda_l(U) = \epsilon$ then by definition of $lloc$ we have $lloc_M(UV) = lloc_M(V) = \bot$ so we conclude by the induction hypothesis.
If $\lambda_l(U) \neq \epsilon$, then either $U$ contains a redex or $U$ is an abstraction. But $U$ cannot be a redex since by assumption $N$ does not contain any redex, and $U$ being abstraction would make $M = U V$ itself a redex which would contradict the fact that redex $(\lambda x . N) A$ in $V$ is the leftmost redex in $M$.

  \item[(3)] $U = \lambda x . N$ and $V = A$. This is the base case of the induction already treated above.
\end{compactitem}

\item[(ii)]
By (i) each reduction of the leftmost spinal-innermost reduction strategy is trivial and it is easy to see that $lloc$ is preserved by trivial reduction, therefore each term in the reduction sequence is in \emph{qnf}.
Since a trivial reduction cannot possibly create new $\beta$-redexes, the reduction sequence necessarily terminates and yields the beta-normal form.
\end{compactitem}
\endproofatend

Observe that if $M$ is in \emph{qnf}, its redexes are not all necessarily trivial, but the non-trivial redexes get eliminated by the trivial reductions:
\begin{example}
The term $M = (\lambda x . u) ((\lambda z .z) y)$ is in \emph{qnf} and trivially reduces to $u$ but contains the non-trivial redex $(\lambda z .z) y$.
\end{example}

\begin{example}
    Take $M = (\lambda x . z ((\lambda w y . y)x)) U$ for any term $U$.
    Then $M$ is in quasi-head normal form with one prime spine redex $(x,U)$.
    Performing one head reduction gives $M \rightarrow_h z ((\lambda w y. y) U)$
    which brings the \emph{hoc} in head position.
$M$ is also in quasi normal form. Indeed $lloc_M(M) = lloc_M (\lambda w y . y) = \bot$. Reducing the term with the leftmost spinal-innermost reduction strategy gives  $M \rightarrow (\lambda x . z (\lambda y . y)) U
    \rightarrow  z (\lambda y . y) $ which yields a beta-normal form.
\end{example}

\begin{theorem}[Soundness of leftmost linear reduction]
\label{thm:soundness_leftmostlinearred}
If $M \rightarrow_{l} N$ then $M$ and $N$ are beta-equivalent.
\end{theorem}
\begin{proof}
The proof is like that of Theorem~\ref{thm:danosreigner_headlinred}. We adapt the notion of \emph{consecutive redex} from \cite{danos-head} to  leftmost-linear reduction: Let $r = (\lambda x, V)$ and  $s =(\lambda y, W)$ be two generalized redexes. We say that $r$ contains $s$ if the node $\lambda y$ lies in the scope of $\lambda x$. We say that $r$ and $s$ are \emph{consecutive} if $r$ contains $s$; no other generalized redex other than $r$ and $s$ is contained in $r$ and contains $s$; and $s$ is the leftmost generalized redex contained in $r$ (\ie, any other redex $t$ contained in $r$ resides on the left of node $\lambda y$ in the usual tree representation of the term). Let $r_0, \ldots, r_p$ denote the consecutive generalized redexes for $p\geq 0$ where $r_0$ is not contained in any other generalized redexes and $r_p$ is the \emph{lloc redex}. (This sequence necessarily exists by definition of \emph{lloc}.) An induction on $p$ shows that $M$ and $N$ reduce to the same term by exactly $p-1$ standard reduction steps (\ie, the strategy reducing the leftmost outermost redex), observing that at every step, no redex can be introduced above or on the left of the \emph{lloc}.
\end{proof}


\begin{property}
\label{property:lambdalist_linearred}
    Let $M_1, M_2$ be two lambda terms.
    \begin{enumerate*}[noitemsep,label=(\roman*)]
        \item If $M_1 \rightarrow_\beta M_2$ then $|\lambda_l(M_2)| \geq
        |\lambda_l(M_1)|$;
        \item If $M_1 \llred M_2$ then $|\lambda_l(M_2)| \geq
        |\lambda_l(M_1)|$;
        \item If $M_1$ trivially reduces to $M_2$ then $\lambda_l(M_2) =
        \lambda_l(M_1)$.
    \end{enumerate*}
\end{property}
\begin{proof}
    (i) and (ii) Performing substitution in an operand of an application does not affect its lambda-list; while substituting a variable in operator position in the application, effectively appends the lambda-list of the substituted term to the lambda-list of the application.
    (iii) follows trivially by definition of $\lambda_l$.
\end{proof}
\begin{example}
Take $M_1 = (\lambda x . x) (\lambda y z . y)$. We have $\lambda_l(M_1) = \epsilon$ but $M \rightarrow_\beta \lambda y z . y = M_2$,
$M_1 \llred (\lambda x . (\lambda y z . y)) (\lambda y z . y) = M_3$ and
and $\lambda_l(M_2) = \lambda_l(M_3) = (\lambda y, \lambda z)$.
Take $N = (\lambda x . U) V$ where $x$ does not occur in $U$ then $N$ trivially reduces to $U$ and $\lambda_l(N) = \lambda_l(U)$.
\end{example}

\begin{property}
\label{prop:qnf_betanf_empty_lambdalist}
    \begin{enumerate}[noitemsep,label=(\roman*)]
    \item If $M$ is in beta-nf and $\lambda_l(M) = \epsilon$ then $M$ is not an abstraction.
    \item If $M$ is in \emph{qnf} and $M$ has a beta-nf then its beta-normal form is not an abstraction.
    \end{enumerate}
\end{property}
\begin{proof}
(i) If $M$ were an abstraction $\lambda x. N$ we would have $\lambda_l(M) = \lambda x$ which contradicts the assumption.
(ii) By Proposition\ref{prop:qnf_nf}(ii), we have that $M$ trivially reduces to
its beta-nf so by Property~\ref{property:lambdalist_linearred},
$M$ and its normal form must have the same lambda list. We can then conclude using (i).
\end{proof}

\begin{theorem}[Completeness of leftmost linear reduction]
\label{thm:completeness_leftmostlinearred}
If $M$ has a beta-normal form then its left-most linear reduction terminates (and thus yields a quasi normal form).
\end{theorem}
The technical proof is in appendix.
\proofatend
Suppose $M$ has a beta-normal form then by a standard result of lambda calculus its head-reduction terminates, therefore by Theorem~\ref{thm:danosreigner_headlinred} the $\hlred$-reduction of $M$ terminates and yields its quasi-\emph{head}-normal form. Since $\hlred$ coincides with $\llred$, for terms not in \emph{qhnf}, we can without loss of generality assume that $M$ is in \emph{qhnf}.
We prove by induction on the \emph{size} of $M$ that
\begin{quote}
{\bf (H1)} If $M$ has a beta-nf then its linear reduction sequence terminates.
\end{quote}
For the variable case, $\llred$ trivially terminates.
If $M = \lambda x . U$ then $U$ must also admit a beta-nf so we can conclude by induction. Suppose $M = U V$. If $M$ is in \emph{qnf} then we conclude immediately. Otherwise, $M$ must have a \emph{lloc} occurrence that is distinct from the \emph{hoc} occurrence.
\begin{enumerate}[noitemsep,label=(\roman*),leftmargin=1pc]
\item Suppose that the \emph{lloc} of $M$ is in $U$ ($lloc_M(U) \ne \bot$).
    We show by finite \emph{structural} induction on subterms $T$ of $U$ that:
    \begin{quote}
    {\bf (H2)} If $lloc_M(T) \neq \bot$ and $hoc(M)$ is in $T$ then there are at most a finite number of steps in the linear reduction sequence of $M$ that involve a linear substitutions in (reducts of) $T$.
    \end{quote}

    \begin{enumerate}[noitemsep,leftmargin=0pt,label=-]
        \item $T$ is a variable. Since $M$'s \hoc\ is in $T$, it must be $T$ itself. Because $M$ is in quasi-head-normal form, the \hoc~is not involved in any generalized redex and therefore $T$ is not involved in any substitution of the linear reduction sequence of $M$.

        \item $T =\lambda x . T'$ then the \emph{hoc} and \emph{lloc} must be in $T'$ and we can conclude by the induction hypothesis (H2).

        \item $T = T_1 T_2$. By definition, the \emph{hoc} is necessarily in $T_1$.

        \begin{enumerate}[noitemsep,label=(\alph*)]

            \item Suppose that $\emph{lloc}$ is in $T_2$. Then we have $lloc_M(T_1) =\bot$ therefore none of the steps in the linear reduction sequence of $M$ will involve substitutions in $T_1$ (there is no  variable in $T_1$ involved in a  general redex therefore subsequent linear reduction steps can only introduce general redexes on the right of $T_1$).

            \begin{compactitem}
                \item Suppose $T_1$ is a variable $x$ then it must be the \emph{hoc} and we have $T = x~T_2$.

                    We show the result by induction (H3) on the number of pending arguments of the \emph{hoc}.

                    \begin{compactitem}
                    \item Base: $z$ has a single argument which must be $T_2$.

                    By Theorem~\ref{thm:danosreigner_headlinred}, reducing all the prime redexes in $M$ yields its head-normal form
                    with $x$ as the head variable.
                    Let $\gamma$ denotes the sequence of substitutions performed on $M$ to reduce the prime redexes, so that $\gamma(M)$ is the head normal form of $M$.
                    Then $\gamma(T_2)$ will becomes the (unique) argument of $x$ in the head normal form. And because $M$ has a beta-nf, the term $\gamma(T_2)$ must also have a beta-nf.

                    Because $z$ has a single hoc argument, all the subsequent linear substitutions in $M$ must necessarily occur under $T_2$ and its reducts: the \emph{lloc} will never get out of $T_2$.
                    (Indeed, otherwise by the inductive definition of \emph{lloc} there would be a subterm $U' V'$ of $M$ where $U'$ contains $z T_2$ and $\lambda_l(U;) \neq \epsilon$, which by definition of {\sf Args} would imply that $V'$ is another hoc argument, contradicting the assumption that the hoc has a single argument.)

                    Therefore if we consider the term $M'$ obtained by replacing
                    the subterm $T$ by just $T_2$, the $\llred$-reduction sequence of $M$ coincides with that of $M'$.

                    Now $M'$ must necessarily have a beta-normal form. Otherwise it's standard beta-reduction sequence would not terminate, which would imply that $\gamma(T_2)$ does not terminate, which contradicts the fact that $\gamma(T_2)$ is the argument of the head variable in the head normal form!

                    Thus, since $M'$ is by definition smaller than $M$,by the induction hypothesis H1, its linear-reduction sequence terminates. We can then conclude since we have shown that
                    the linear reduction sequences of $M$ and $M'$ coincide.

                    \item Inductive case: suppose the result holds for terms with $n \geq 1$ \hoc\ arguments, and suppose ${\sf Args}(M) = n+1$.

                    Let $U_{n+1} V_{n+1}$ denote the subterm of $M$ such that $V_{n+1}$ is the last \hoc\ argument of $M$. We consider the term $M'$ obtained from $M$ by removing the last \hoc\ argument of $M$ (\ie~replacing $U_{n+1} V_{n+1}$ by $U_{n+1}$).
                    Observe that the last argument of the hoc cannot be possibly be argument of any generalized redex, therefore the linear reduction of $M$ must  coincide with that of $M'$ up until the \hoc\ of the reduct of $M$ lies under the last \hoc\ argument $U_{n+1}$.

                    By induction hypothesis (H3), the $\llred$-reductions sequences of $M'$ involves a finite number of substitutions under $T_2$. Consequently, the same holds in $M$.
                    \end{compactitem}


                \item  Suppose $T_1$ is not a variable. Let $z$ denote a fresh variable in $M$, and $M' = M[z/T_1]$ be the term obtained by replacing the subterm $T_1$ in $M$ by the fresh variable $z$.

                Since $T_1$ is not a variable, $M'$ must be strictly smaller than $M$. And since $M$ has a beta-nf, so must $M'$, thus by the induction hypothesis (H1), the $\llred$ reduction sequence of $M'$ terminates. Consequently, there is only a finite number of steps in the linear reduction of $M'$ that involve substitutions in reducts of $z~T_2$.

                Clearly, the linear reduction steps under $T_2$ in $M'$ coincide with those of $T_2$ in $M$. Hence, we can conclude that at most a finite number of substitutions take place under $T= T_1 T_2$ in the linear reduction sequence of $M$.
            \end{compactitem}

            \item Suppose that $\emph{lloc}$ is in $T_1$. Then by the induction hypothesis (H2) there is a finite number of steps in the linear reduction sequence of $M$ that involve (reducts of) $T_1$. Let us consider the first reduct $R$ of $M$ in the $\llred$ reduction sequence such that the \emph{lloc} is not under a reduct of $T_1$.  Observe that in the linear reduction sequence of $T$, the subterm $T_2$ remains unmodified as long as the \emph{lloc} remains in (the reduct) of $T_1$.

            At $R$, if the \emph{lloc} is not under $T_2$ then none of the further steps in the linear reduction sequence of $R$ will involve $T_2$ (\ie, subsequent substitutions will occur outside of $T$ -- on its right,  more precisely). We can therefore conclude.

            If instead at $R$ the \emph{lloc} lies under the subterm $T_2$: $lloc_R(T2) \ne\bot$, then $T_1$ cannot be a variable, otherwise it would be both the \emph{hoc} and the \emph{lloc} of $M$ which contradict the assumption. We can therefore conclude by applying the same reasoning as case (a) above to the reduct $R$.
        \end{enumerate}

    \end{enumerate}

    \item Suppose that the lloc is in $V$.
    Then by definition of the \emph{lloc} for the application case, we must have $lloc_M(U) =\bot$ and $\lambda_l(U)=\epsilon$.
    Hence $U$ is a term in \emph{qnf} and by Proposition~\ref{prop:qnf_nf} it must
    have a beta-normal form. By Property~\ref{prop:qnf_betanf_empty_lambdalist}
    $U$'s normal form is not an abstraction. Hence since $M$ has a normal form so must $V$ and we have $\betanf{M} = \betanf{U} \betanf{V}$.
    By the induction hypothesis (H1) the $\llred$-reduction of $V$ must terminate.
    We can then conclude by observing that the linear-substitutions performed in the
    linear-reduction sequence of $M$ map one-to-one to those in the linear-reduction sequence of $V$ (\ie, if $V \llred V_2 \llred \ldots V_i$ then $M = U V \llred U V_2 \llred \ldots U V_i$ for some $i\geq 0$)
\end{enumerate}
\endproofatend


\section{Correctness of ULC normalization}
\label{sec:correctness_ulc_normalization}
Correctness of the evaluation procedure relies, like in the STLC case, on a Path Characterization Theorem.
Unlike STLC, however, we show this characterization without appealing to the Game Semantics correspondence, since its untyped counterpart remains a conjecture (Conjecture~\ref{conj:ulc_corresp}). The proof relies instead on a soundness result showing that $\sim$-equivalence classes are preserved by \emph{leftmost-linear reduction} from Section~\ref{sec:leftmostlinearred}. It is based on a bisimulation argument requiring careful analysis of the traversals structure.

\begin{definition}
    \label{def:spinaldescent_pendingarglookup}
A subsequence of a traversal is a
\begin{enumerate*}[nosep,label=(\roman*)]
\item \defname{spinal descent} if it is a path in the computation tree consisting solely of lambda and $@$-nodes, and where each lambda node (except the first occurrence if it is a lambda node) is the $0$th child of the preceding $@$-node.
\item \defname{pending argument lookup} if it consists of an alternation of ghost lambda nodes and ghost variables nodes, starting with an external ghost lambda node and terminated by a structural internal lambda node in $\NodesLmd$.
\item \defname{branching descent} if it consists of an alternation of external structural lambda nodes and external structural variable nodes.
\end{enumerate*}
\end{definition}

\begin{definition}[lloc of a node]
For any tree node $n$ we write $lloc_M(n)$ to denote $lloc_M(N)$ where $N$ is the subterm of $M$ rooted at $n$ and $lloc$ denotes the leftmost linear occurrence from Def.~\ref{def:leftmostlinearreduction}.
\end{definition}

\begin{property}
\label{prop:strand_spinaldescent}
Let $t$ be a traversal of $M$ and $u$ be a subsequence of $t$ of the form
$u = \underline{\lambda \overline{x}} \cdot v \cdot \underline{x}$
for some subsequence $v$ of $t$, external structural variable node $x$ and external structural lambda node $\lambda \overline{x}$.
If $v$ is empty then the node following $u$ in $t$, if it exists, is necessarily a structural child lambda node of $x$ (\ie, it cannot be a ghost node).
\end{property}
\begin{proof}
 If $v$ is empty then the strand ending at $x$
 is just $\underline{\lambda \overline{x}} \cdot \underline{x}$ therefore
 $t$'s arity threshold is precisely $|x|$, in which case rule \rulenamet{IVar} can only visit a structural child of $x$.
\end{proof}

\begin{definition}[Strand types]
    \label{def:strandtypes} We distinguish two types of strands:
    \begin{itemize}[nosep]
        \item Structural argument strands (S) of the form $\underline{\lambda\overline{\eta}} \cdot u \cdot \underline{x}$
            where
            \begin{itemize}[nosep]
            \item $\lambda\overline{\eta}$ is an external structural lambda node verifying $lloc_M(\lambda\overline\eta) = \bot$,
            \item $u$ is a (possibly empty) \emph{spinal descents} of internal nodes,
            \item $x$ is an external structural variable node;
            \end{itemize}

        \item Ghost argument strands ($G$) of the form $\underline{\ghostlmd} \cdot  v \cdot \lambda\overline{y} \cdot u \cdot \underline{x}$
        where
        \begin{itemize}[nosep]
            \item $\ghostlmd$ is an external ghost lambda node,
            \item $v$ is a \emph{pending parameter lookup} (ghost nodes) of length shorter than $u$,
            \item $\lambda\overline{y}$ is an internal lambda node verifying $lloc_M(\lambda\overline{y}) = \bot$,
            \item x is an external structural variable node.
        \end{itemize}
    \end{itemize}
\end{definition}

\begin{proposition}[Quasi-normal forms strand decomposition]
\label{prop:qnf_strand_decomposition}
Any maximal traversal $t_{\max}$ of a \emph{qnf} $M$
consists of a succession of $S$ and $G$ strands determined by the following transition system with states $\{ S, G \}$:
\begin{center}
\begin{tabular}{r|ll l}
State (strand type) & Transition & Current strand & Next strand \\
\hline \hline
Initial state & $\rightarrow S$ & $\epsilon$ & $\underline{\lambda\overline{\eta_1}} \cdot u_1 \cdot \underline{x_1}$ \\
\hline
\multirow{2}{*}{S} & $S \rightarrow G$ & \multirow{2}{*}{$\underline{\lambda\overline{\eta_1}} \cdot u_1 \cdot \underline{x_1} $}
& $\underline{\lambda\overline{\eta_2}} \cdot u_2 \cdot \underline{x_2}$  \\
& $S \rightarrow S$ & & $\underline{\ghostlmd} \cdot  v \cdot \lambda\overline{y_2} \cdot u_2 \cdot \underline{x_2}$ \\
\hline
\multirow{2}{*}{G} & $G \rightarrow S$ & \multirow{2}{*}{$\underline{\ghostlmd} \cdot  v \cdot \lambda\overline{y_1} \cdot u_1 \cdot \underline{x_1}$}
& $\underline{\lambda\overline{\eta_2}} \cdot u_2 \cdot \underline{x_2}$
\\
  & $G \rightarrow G$ & & $\underline{\ghostlmd} \cdot v \cdot \lambda\overline{y_2} \cdot u_2 \cdot \underline{x_2}$
\end{tabular}
\end{center}
where $\lambda\overline{y_2}$ is justified by a $@$-node in $u_1$.
\end{proposition}
\proofatend
By case analysis on the current state:
\begin{enumerate}
\item[(Init)] The first strand is obtained by applying the rule \rulenamet{Root} followed by repeated applications of \rulenamet{App} and \rulenamet{Lam}, which yields $u$, a path from the root $\lambda\overline{\eta}$ to some variable node $x$ in the tree of $M$, so that $u$ is a spinal descent consisting only of lambda nodes and application nodes.
Suppose that $x$ is an internal variable then by definition it is
involved in some generalized redex and therefore we have $lloc_M(\underline{\lambda\eta}) = x$, but since $M$ is in \emph{qnf} we have $lloc_M(\lambda \overline{\eta}) = lloc_M(M) = \bot$ which gives a contradiction.
Hence $x$ is an external variable which ends the first strand of type $S$.

\item[(Structural)] The previous strand is of form $\underline{\lambda\overline{\eta_1}} \cdot u_1 \cdot \underline{x_1} $.

If the node following $x_1$ is a structural node $\lambda\overline{\eta_2}$ then by rule \rulenamet{IVar}, it's a child lambda node of $x_1$.
Because $u_1$ is a spinal descent, the subterm rooted at $\lambda \overline\eta_1$ is of the form $\lambda \overline\eta_1 . x_1 A_1 \ldots A_q$ for some $q\geq1$, therefore we have $lloc_M(\lambda \overline\eta_1 . x_1 A_1 \ldots A_q) = \bot$. Since $x_1$ is external, by Property~\ref{prop:qnf_longapply} we have $lloc_M(A_j)=\bot$ for all $1\leq j\leq q$, so in particular $lloc_M(\lambda\overline{\eta_2}) = \bot$. The same logic used in the \emph{Initialization} case lets us conclude that $u_2$ is a spinal descent to external node $x_2$. We have shown that the next strand is of type $S$.
\\

Now suppose instead that the node following $x_1$ is a ghost node then by Property~\ref{prop:strand_spinaldescent}, $u_1$ is necessarily non-empty: we have $u_1 =
@_1 \cdot\lambda\overline{\xi_1} \cdots @_q\cdot \lambda\overline{\xi_q}$ for some $q\geq 1$. The subterm at $@_r$,  represented on Figure~\ref{fig:strand_decomposition_proof_induction}, is of the form
$$(\lambda\overline{\xi_r}. A_0) A_1 \ldots A_{k_{r}-1}\ (\lambda\overline{y_2}. \ldots)A_{k_{r}+1}\ \ldots A_{|@_r|} .$$
By Proposition~\ref{prop:weaving}(iii), $x_1$ is followed by a \emph{pending parameter lookup} $v$ necessarily shorter than $u_1$ and terminated by a structural lambda node $\lambda\overline{y_2}$ justified by some application node $@_r$ in $u_1$ for some $1\leq r \leq q$:
\begin{align*}
 t &= \cdots \Pstr[10pt]{
(l){\underline{\lambda\overline{\eta}}} \cdot
{@_1} \cdot \lambda\overline{\xi_1} \cdot
\cdots
(atr){@_r} \lambda\overline{\xi_r} \cdot
\cdots
{@_q} \cdot\lambda\overline{\xi_q} \cdot
(x){\underline{x}}\cdot
(gl-x,25:k){\underline\ghostlmd}
\cdot v
\cdot
(ly-atr,25:k_r){\lambda\overline{y_2}}
 } \\
v &= \ghostvar_{k_q} \cdot \ghostlmd_{k_q} \cdot
\ghostvar_{k_{q-1}} \cdot \ghostlmd_{k_{q-1}} \cdot
\cdots
\ghostvar_{k_{r-1}} \cdot \ghostlmd_{k_{r-1}} \cdot
\ghostvar_{k_{r}}
\end{align*}
where, for all $i$ ranging from $r$ to $q$, ghost variable
$\ghostvar_{k_i}$ and ghost lambda $\ghostlmd_{k_i}$ point respectively to $\lambda\overline{\xi_i}$ and $@_i$ with label $k_i\geq1$ defined by:
\begin{align*}
  k_i &= k - |x| + |\lambda\overline{\xi}_q| + \sum_{i\leq j < q} (|\lambda\overline{\xi}_{j}| - |@_{j+1}|)
\\
   &= k - |x| +
  |\lambda\overline{\xi}_{i}| +
  \sum_{i< j\leq q} (|\lambda\overline{\xi}_j| - |@_j|)
\end{align*}

Since all the nodes in $v$ are ghost nodes, by Property~\ref{prop:ghost_justifier_arity} their justification label is greater than the justifier's arity; the opposite holds for $\lambda\overline{y_2}$ which is a structural node. Therefore:
\begin{align}
k > |x| \label{eqn:k_greater_than_x} \\
k_i > |@_i| & \quad\mbox{for all $r+1\leq i \leq q$ } \label{eqn:qnfdecomp_kj}\\
k_r \leq |@_r|
\end{align}

To prove that $lloc(\lambda\overline{y_2}) =\bot$ we first show that for all $i$ ranging from $r+1$ to $q$ we have
 $k_i>|\lambda_l(\lambda\overline{\xi_i})|$. We prove the result by a finite induction proceeding bottom-up from the lower tree node
 $\lambda\overline{\xi_q}$ up to $\lambda\overline{\xi_r}$.

\emph{Base case:} We have $i=q$. The subterm
rooted at $\lambda\overline{\xi_q}$ is in head normal form
$\lambda\overline{\xi_q}. x \cdots$ therefore its head lambda-list consists precisely of the lambda abstractions $\lambda\overline{\xi_q}$, and thus
$|\lambda_s(\lambda\overline{\xi_q})| = |\lambda\overline{\xi_q}|$.
We have $k_q = k -|x| + |\lambda\overline{\xi_q}|$ and by Equation~\ref{eqn:k_greater_than_x} we have $k>|x|$, hence $k_q >|\lambda\overline{\xi_q}|$.

\emph{Induction case:} Let $r+1\leq i \leq q$. We suppose that the result holds for $i$ we show that it must hold for $i-1$. By definition, the head lambda-list at $\lambda\overline{\xi_i}$ is precisely given by the variables in the abstraction concatenated with the head lambda-list at $@_{i+1}$, thus $|\lambda_l(\lambda\overline{\xi_i})| =
|\lambda\overline{\xi_i}| + |\lambda_l(@_{i+1})|$. The lambda-list at $@_i$, is by definition, obtained by popping $|@_i|$ lambdas from the lambda-list at $\lambda\overline{\xi_i}$; or is the empty list if the arity is greater than the number of pending lambdas. In terms of lengths this means:
\begin{align*}
    |\lambda_l(@_i)| &= \max(0, |\lambda_l(\lambda\overline{\xi_i})| - |@_i|) \\
     &= \max\left( 0, |\lambda_l(@_{i+1})| + |\lambda\overline{\xi_i}| - |@_i| \right) & \mbox{for $r\leq i \leq q-1$.}
\end{align*}

Furthermore by definition of $k_{i-1}$ we have:
\begin{equation}
k_{i-1} = k_i + |\lambda\overline{\xi_{i-1}}| - |@_i| \label{eqn:ki_minus_one}
\end{equation}
Hence,
\begin{align*}
    |\lambda_l(\lambda\overline{\xi_{i-1}})| - k_{i-1}
    &= |\lambda\overline{\xi_{i-1}}| +|\lambda_l(@_i)| - k_{i-1} &\mbox{(Def~of $\lambda_l$)}\\
    &= |\lambda\overline{\xi_{i-1}}| + \max\left( 0, |\lambda_l(\lambda\overline{\xi_i})| - |@_i| \right) - k_{i-1} &\mbox{(Def~of $\lambda_l$)}\\
    &= \max\left( |\lambda\overline{\xi_{i-1}}| - k_{i-1},    |\lambda\overline{\xi_{i-1}}| + |\lambda_l(\lambda\overline{\xi_i})| - |@_i| - k_{i-1} \right) \\
    &= \max\left(  |@_i|- k_i, |\lambda_l(\lambda\overline{\xi_i})| - k_i \right) &\mbox{(Eqn~\ref{eqn:ki_minus_one})}
\end{align*}
By the induction hypothesis we have $|\lambda_l(\lambda\overline{\xi_i})| < k_i$, and by Equation~\ref{eqn:qnfdecomp_kj} we have $ |@_i|< k_i$, therefore the maximum of the two quantities in the last equation is negative which shows
the desired result $|\lambda_l(\lambda\overline{\xi_{i-1}})| < k_{i-1}$.

\begin{figure}[htbp]
\centering
\begin{tikzpicture}[baseline=(root.base),level distance=6.5ex,inner ysep=0.5mm,sibling distance=15mm]
    \node (root){}
    child[dashed] {node{$\lambda\overline\eta$}
        child[solid] {node{$@_1$}
            child{node{$\ldots$}
                child{node{$@_{r}$}
                    child{node{$\lambda\overline{\xi_r}$}
                        child[dashed]{node{$@_{q-1}$}
                            child[solid]{node{$\lambda\overline{\xi_{q-1}}$}
                                child[solid]{node{$@_q$}
                                    child{node{$\lambda\overline\xi_q$}
                                        child{node {$x_1$}
                                            child[dashed,level distance=5ex]{node{}}}
                                        edge from parent node[above left]{$0$}
                                    }
                                    child{node{\ldots}}
                                    child[dashed]{node{}}
                                }
                                edge from parent node[above left]{$0$}
                            }
                            child[solid]{node{\ldots}}
                            child[dashed]{node{}}
                        }
                        edge from parent node[above left]{$0$}
                    }
                    child[dashed]{node{\ldots}}
                    child{node{$\lambda\overline{y_2}$}
                        child[dashed,level distance=5ex]{node{}}
                        edge from parent node[left]{$k_r$}
                    }
                    child[dashed]{node{\ldots}}
                    child{node{} edge from parent node[right]{$|@_r|$}}
                }
                edge from parent node[above left]{$0$}
            }
            child{node{\ldots}}
            child[dashed]{node{}}
        }
    };
    \end{tikzpicture}
\caption{Relevant sub-tree for a traversal strand of type $G$ (assuming $q>r+1$).}
\label{fig:strand_decomposition_proof_induction}
\end{figure}

We now show that $lloc_M(\lambda\overline{y_2})=\bot$.
By assumption we have that $lloc_M(\lambda\overline{\eta}) =\bot$ so since $@_r$ occurs in the spine of $M$, by definition of $lloc_M$ we must have $lloc_M(@_r) =\bot$, and thus also $lloc_M(\lambda\overline{\xi_r}) = \bot$. We have just shown that $|\lambda_l(\lambda\overline{\xi_r}| < k_r$, consequently, the sub-term $(\lambda\overline{\xi_r} . A_0) A_1 \ldots A_{k_r-1}$ has more operands than pending lambdas in the operator's lambda list. By definition of lambda list this implies that $\lambda_l((\lambda\overline{\xi_r} . A_0) A_1 \ldots A_{k_r-1})$ is empty. Hence, we necessarily have $lloc(\lambda\overline{y_2}) = \bot$, otherwise by definition of $lloc$ we would have
 $lloc((\lambda\overline{\xi_r} . A_0) A_1 \ldots A_{k_r-1} (\lambda\overline{y_2}. \ldots)) = lloc(\lambda\overline{y_2}) \ne \bot$, which subsequently implies $lloc(\lambda\overline{\eta_1})\ne\bot$, contradicting the assumption.
The sequence $u_2$ is shown to be a spinal descent by the same argument used in the \emph{structural argument} case above, using the fact that $lloc_M(\lambda\overline{y_2})=\bot$. We have thus shown that the next strand is of type $G$.

\item[(Ghost)] If the node following $x_1$ is a structural node then, like we have shown in the previous case, the next strand is necessarily a spinal descent of type $S$.

If the node following $x_1$ is a ghost.
Once again, the weaving property shows that the ghost node after $x_1$ is necessarily followed by a \emph{pending parameter lookup} terminated by structural lambda node $\lambda\overline{y_2}$ pointing to some node in $v \cdot \lambda\overline{y_1} \cdot u_1$.
But since $v$ consists solely of ghost node, the justifier is necessarily in $u_1$. We can thus conclude with the same argument as in the previous case that the next strand is necessarily of type
$G$.
\end{enumerate}
\endproofatend

Figure~\ref{fig:qnf_strand_decomposition_statemachine} represents the more fine-grained state machine underpinning the stand decomposition result. It has three states representing each possible type of external node: lambda, ghost lambda or input variable. The bottom two states correspond to the start of a new strand in the traversal. The solid-arrows represent consecutive applications of deterministic traversal rules while dashed-arrows represent the two possible non-deterministic choices in rule \rulenamet{IVar}: one for structural child lambda nodes (left) and one for ghost children (right).
\begin{figure}[htbp]
\centering
\begin{tikzpicture}[->,>=stealth',shorten >=1pt,auto,node distance=3cm,
    semithick]
\tikzstyle{every state}=[fill=red,draw=none,text=white]

\node[initial,state]
    (L) {$\lambda$};
\node[state,fill=white,draw=black,text=black]
    (V) [above right of=L, xshift=1cm] {$IVar$};
\node[state,node distance=6cm]
    (G) [right of=L] {$\ghostlmd$};

\path (L) edge [bend right]          node[right,text width=3cm]{Descent to linear head-occurrence} (V)
      (V) edge [bend right, dashed]  node[left,text width=3cm] {Picked structural argument ($k\leq|x|$)} (L)
          edge [bend left, dashed]   node[text width=2cm] {Picked ghost argument ($k>|x|$)} (G)
      (G) edge                       node {Look-up pending argument} (L)
;
\end{tikzpicture}
\caption{State machine underpinning the \emph{qnf} strand decomposition.}
\label{fig:qnf_strand_decomposition_statemachine}
\end{figure}
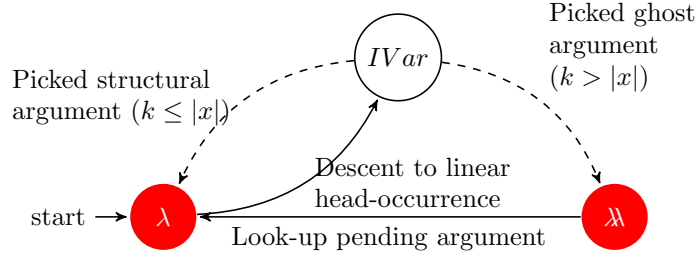

\begin{proposition}
\label{prop:qnf_traversals_are_finite}
Let $t \in \travsetnorm(M)$ for some \emph{quasi normal} term $M$. We have:
\begin{enumerate}[label=(\alph*), nosep]
\item $lloc_M(\alpha) = \bot$ for all external lambda node $\alpha \in \NodesLmd$ occurring in $t$.
\item $t$ does not contain any internal variable node (\ie, all the internal nodes are $@$ nodes);
\item $t$ is finite.
\end{enumerate}
\end{proposition}
\begin{proof}
(a) and (b) are direct consequences of Proposition~\ref{prop:qnf_strand_decomposition}.
(c) Let's consider the cases of the strand decomposition of Proposition~\ref{prop:qnf_strand_decomposition}
observe that for the two transitions $\rightarrow S$ and $S \rightarrow S$, because $u_1$ and $u_2$ are spinal descents, the last node in the next strand ($x_1$ or $x_2$) must be a descendant (in the computation tree of $M$) of the first node ($\lambda\overline{\eta_1}$) from the previous strand.
For transition $S\rightarrow G$, since both $u_1$ and $u_2$ are spinal descents, there are also paths in the tree, therefore since $\lambda\overline{y_2}$ points to $u_1$, $x_2$ must also be a descendant of $\lambda\overline{\eta_1}$. For transitions
$G \rightarrow S$ and $G \rightarrow G$, since $u_1$ and $u_2$ are spinal descents, $x_2$ is a descendant of $\lambda\overline{y_2}$, itself a descendant of $\lambda\overline{y_1}$.

Suppose $t$ is infinite then from the above strand decomposition we can construct an infinite path in the tree of $M$ which contradicts the fact that $M$ is a finitary term.
\end{proof}

\begin{lemma}[Generalized redex argument lookup]
\label{lemma:genredex_lookup}
Let $x$ be a variable occurrence involved in a generalized redex in $gr(M)$.
\begin{enumerate}[label=(\roman*), nosep]
    \item The path from the parent node of the redex argument down to $x$'s binder is of the form:
$ @_r \cdot \lambda\overline{x_r} \cdots
@_2 \cdot \lambda\overline{x_2} \cdot @_1 \cdot \lambda\overline{x}$
for $r\geq0$, where $@_r$ denotes the parent node of the redex argument, and $\lambda\overline{x}$ denotes $x$'s binder, such that $\lambda x$ is the $i$th lambda in the bulk lambdas $\lambda\overline{x}$ for $i\geq0$.

\item The argument of $x$ is precisely given by the $i_k$th child of node $@_k$ where $k$ is the smallest index verifying $i_k<|@_k|$ where: $i_1 = i$ and $i_{j+1} = i_j - |@_j| + |\lambda\overline{x_{j+1}}|$ for $j\geq 0$.
\end{enumerate}
\end{lemma}
\proofatend
(i) By definition of generalized redexes and lambda lists (Def.~\ref{dfn:generalized_redex}), $@_r$ necessarily occurs in the path from $\lambda\overline{x}$ to the root. Further the path from $@_r$ to $\lambda\overline{x}$ is necessarily a spinal descent (\ie, it contains only lambda and application nodes, and no variable node). Indeed, if it contained a variable node $z$ then the lambda-list at the subterm rooted at $z$ would be empty, in which case no argument could possibly form a generalized redex with $\lambda x$.

(ii) We show by induction on $k$ that $x$ is the variable abstracted by the $i_k$th lambda in $\lambda_l(\lambda\overline{x_k})$, which implies the result since by definition $\lambda x$ forms a generalized redex with the $i_k$th argument of $@_k$.
Base case $k=1$. By assumption, $\lambda x$ is the $i$th lambda in the bulk lambda $\lambda\overline{x}$. Let $k\geq 1$. By the induction hypothesis $\lambda x$ is the $i_{k-1}$ lambda in $\lambda_l(\lambda\overline{x_{k-1}})$.
substitution involved iBy definition of $k$ we have $i_{k-1}\geq|@_{k-1}|$, therefore by definition of $\lambda_l$, the long-application $@_{k-1}$ has the effect of popping exactly $|@_{k-1}|$ arguments from the lambda-list, and the following abstraction $\lambda\overline{x_k}$ then adds $|\lambda\overline{x_k}|$ more lambdas in front of the lambda list. Hence $\lambda x$ gets moved to position $i_{k-1} - |@_{k-1}| + |\lambda\overline{x_k}| = i_k$ in the lambda list at $\lambda\overline{x_k}$.
\endproofatend

\paragraph{Traversal bisimulation}

We define $\travset^\dagger(M)$ as the subset of traversals of $\travsetnorm(M)$ ending with either a structural node or an external ghost variable (\ie, not ending with an internal ghost variable nor a ghost lambda).
Observe that every traversal in $\travsetnorm$ has a traversal extension in $\travset^\dagger$.
Indeed, if $t \in \travsetnorm$ ends with an internal ghost variable or a ghost lambda then repeatedly applying
rules \rulenamet{Var} and \rulenamet{Lmd} yields a traversal ending with an external node.

\begin{definition}[Traversal bisimulation]
\label{def:bisimilar_terms}
Let $M$ and $N$ be two terms with variables in $\mathcal{V}$. Let $\phi\colon \justseqset(N) \rightarrow\justseqset(M)$ be a function from the justified sequences of $N$ to justified sequences of $M$. We define the state transition system $(X, R_\phi, \rightarrow_X)$ as:
\begin{itemize}[nosep]
    \item \emph{States}: $X$ is the disjoint union $\travset^\dagger(M) + \travset^\dagger(N)$.

    \item \emph{Binary relation $R_\phi\subseteq X \times X$} defined for any traversals $t, u \in X$ as:
    $ t~R_\phi~u$ just if $t\in\travset^\dagger(M)$, $u\in\travset^\dagger(N)$ and $\phi(u) = t$.

    \item \emph{Transitions} For all $t, t' \in \travset^\dagger(M)$ we write $t \rightarrow_M t'$ just if $t'$ is a minimal extension of $t$ that remains in $\travset^\dagger(M)$ (\ie, $t'$ extends $t$ using rules of Table~\ref{tab:normalizing_trav_rules} and any occurrence strictly between the last occurrence in $t$ and $t'$ is an either internal ghost variable or a ghost lambda node.) We define $\rightarrow_N$  identically and write $t \rightarrow_X t'$ if either $t \rightarrow_M t'$ or $t \rightarrow_N t'$ holds for $t, t' \in X$.
\end{itemize}

$M$ and $N$ are \defname{$\phi$-bisimilar} if $R_\phi$ defines a bisimulation over $X$ with respect to $\rightarrow_X$, that is:
\begin{itemize}[nosep]
    \item (B1) for all $t_1, t_2 \in \travset^\dagger(M)$, $u_1 \in \travset^\dagger(N)$ with $t_1~R_\phi~u_1$ and $t_1 \rightarrow_M t_2$ there is $u_2 \in \travset^\dagger(N)$ such that $u_1 \rightarrow_N u_2$ and $t_2~R_\phi~u_2$.
    \item (B2) for all $t_1 \in \travset^\dagger(M)$, $u_1, u_2 \in \travset^\dagger(N)$ with $t_1~R_\phi~u_1$ and $u_1 \rightarrow_N u_2$ there is $t_2 \in \travset^\dagger(M)$ such that $t_1 \rightarrow_M t_2$ and $t_2~R_\phi~u_2$.
\end{itemize}
\end{definition}

\begin{lemma}
\label{lem:bisimulation_isomorphism}
If $M$ and $N$ are $\phi$-bisimilar;
the restriction of $\phi$ to $\coresymbol(\travset^\dagger(N))$ is injective;
and there is a mapping $\rho: \VarSet^*(N) \rightarrow \VarSet^*(M)$ such that for any variable names $\alpha\in\VarSet^*(N)$,  $\phi$ and $\coresymbol_\alpha$ commute modulo $\rho$:
(\ie, for all $u$, $ \coresymbol_{\rho(\alpha)}(\phi(u)) = \phi(\coresymbol_{\alpha}(u))$), then $\coresymbol(\travset(N)) \cong \coresymbol(\travset(M))$.
\end{lemma}
\begin{proof}
Consider $\phi$ as a function from $\coresymbol(\travset^\dagger(N))$ to $\coresymbol(\travset^\dagger(M))$.
It is \emph{well-defined}: $\travset^\dagger(N)$ is by definition the transitive closure of the empty traversal by $\rightarrow_N$, so for any traversal $u$ of $\travset^\dagger(N)$, repeatedly applying (B2) for each step of its $\rightarrow_N$ derivation shows that $t = \phi(u) \in \travset^\dagger(M)$. By commutativity of $\coresymbol$ and $\phi$ this yields $\coresymbol(t)=\phi(\coresymbol(u)) \in \coresymbol(\travset^\dagger(M))$.
Similarly, it is \emph{surjective} by (B1) because $\travset^\dagger(M)$ is the transitive closure of the empty traversal by $\rightarrow_M$. Finally it is \emph{injective} by assumption.
Hence we have $\coresymbol(\travset^\dagger(N)) \cong \coresymbol(\travset^\dagger(M))$, and because $\coresymbol$ eliminates internal nodes from traversals, we also have $\coresymbol(\travset(N)) \cong \coresymbol(\travset(M))$.
\end{proof}

The sets of normalizing traversals before and after linear reduction are identical, up to some mapping between nodes of the term and its reduct.
\begin{proposition}[Traversals are sound for linear reduction]
\label{prop:ulctrav_impl_linear_reduction}
Let $M$ be an untyped term.
If $M \llred N$ then $\coresymbol(\travsetnorm(M)) \cong \coresymbol(\travsetnorm(N))$.
\end{proposition}
\proofatend
Suppose that $M \llred N$. Let $x$ denote the leftmost linear variable occurrence in $M$ and
let $(\lambda x, A)$ be the generalized redex involving it, so that $N$ is the term obtained after
firing the generalized redex $(\lambda x, A)$. A property of linear head reduction is that the reduct contains the original
term itself, modulo alpha conversion and relabelling of the \emph{lloc} node. This induces an implicit map $\Phi$ from nodes
of the reduct $N$ to those of $M$:
$$\Phi : \ExtendedNodes(N)  \rightarrow \ExtendedNodes(M) $$
which maps nodes from the substituted subterm in $N$ to the corresponding nodes under argument $A$ in $M$;
and every node in $N$ that is not involved in the substitution, to its counterpart in $M$.

The substitution involved in the leftmost linear reduction affects the computation tree in three possible ways depending on the structure and number of $x$'s operands:
\begin{description}[itemindent=1em]
    \item[Case 1] $x$ has at least one operand ($|x|>0$) and $A$ is an abstraction;
    \item[Case 2] $x$ has at least one operand ($|x|>0$) and $A$ is not an abstraction;
    \item[Case 3] $x$ is unapplied ($|x|=0$).
\end{description}

%

We proceed by showing that for each of the three cases, there is an injection $\phi$ such that $M$ and $N$ are $\phi$-bisimilar.

\begin{description}[itemindent=0em,leftmargin=0cm]
    \item[Case 1]
    The reduction is represented on Figure~\ref{fig:firing_genredex_effect_on_tree_with_operands_and_lambda_argument}
    where $x$ has some operand ($|x|>0$), and $A$ is an abstraction $\lambda\overline{y}. R$.

    Let $\lambda\overline{x}$ denote $x$'s binding lambda node in $M$, $i$ denote $x$'s binding index, and $@_a$ denote the parent node of the subterm $A$.

    After reduction:
    (i) Node $x$ gets replaced by an application node $@_x$ and its existing children are preserved.
    (ii) The subterm $\lambda\overline{y}'. R'$, obtained from $R$ by renaming variable afresh to avoid variable capture, becomes the operator child of $@_x$.  (Note that we write $\lambda\overline{y}'$ to distinguish the root node of $R'$ from the root $\lambda\overline{y}$ of $R$ in $N$, even though the bound variable names are $\overline{y}$ for both nodes.)

    We thus have $\Phi(@_x) = x$, $\Phi(\lambda\overline{y}') = \lambda\overline{y}$, and $\Phi$ maps each node from subterm $\lambda\overline{y}'. R'$  to the corresponding node under $\lambda\overline{y}.R$.

\begin{figure}[htbp]
    \centering
    \begin{tikzpicture}[baseline=(root.base),level distance=5ex,inner ysep=0.5mm,sibling distance=8mm]
    \node (root){}
            child[dashed]{node{$@_{a}$}
                child[solid]{node{$\lambda\overline{\xi}_a$}
                    child[dashed]{node{$@$}
                        child{node{$\lambda\ldots$}
                            child[solid]{node{$@$}
                                child[solid]{node{$\lambda\overline{z}$}
                                    child[solid,level distance=5ex]{node {$x$}
                                        child{node{$B_1$}}
                                        child{node{$\ldots$}}
                                        child{node{$B_q$}}
                                        }
                                    edge from parent node[above left]{$0$}
                                }
                                child[dashed]{node{}}
                            }
                        }
                        child[dashed]{node{}}
                    }
                    edge from parent node[above left]{$0$}
                }
                child[dashed]{node{\ldots}}
                child[solid]{node{$\lambda\overline{y}$}
                    child{node{$R$}}
                }
                child[dashed]{node{\ldots}}
            }
    ;
    \node (root2) at (8,0) {}
        child[dashed]{node{$@_{a}$}
            child[solid]{node{$\lambda\overline{\xi}_a$}
                child[dashed]{node{$@$}
                    child{node{$\lambda\ldots$}
                        child[solid]{node{$@$}
                            child[solid]{node{$\lambda\overline{z}$}
                                child[solid,level distance=5ex]{node {$@_x$}
                                    child{node{$\lambda\overline{y}'$}
                                        child{node{$R'$}}
                                    }
                                    child{node{$B_1$}}
                                    child{node{$\ldots$}}
                                    child{node{$B_q$}}
                                    }
                                edge from parent node[above left]{$0$}
                            }
                            child[dashed]{node{}}
                        }
                    }
                    child[dashed]{node{}}
                }
                edge from parent node[above left]{$0$}
            }
            child[dashed]{node{\ldots}}
            child[solid]{node{$\lambda\overline{y}$}
                child{node{$R$}}
            }
            child[dashed]{node{\ldots}}
        }
    ;
    \draw[<-] ([xshift=2.5cm,yshift=-2cm]root.east) -- +(2cm,0)
    node[midway, above] {$\Phi$};
    \end{tikzpicture}
    \caption{(Case 1) Reduction $M\llred N$ where node $x$ has at least one operand, and $A$ is a non-dummy
     abstraction $\lambda\overline{y} . R$. Linearly firing  generalized redex $(\lambda x, A)$ for $x$ in $M$ (left)
     yields $N$ (right) where $R'$ is a copy of $R$ with fresh names to avoid variable capture.
    Function $\Phi$ maps nodes from the right tree to the left tree.}
    \label{fig:firing_genredex_effect_on_tree_with_operands_and_lambda_argument}
    \end{figure}
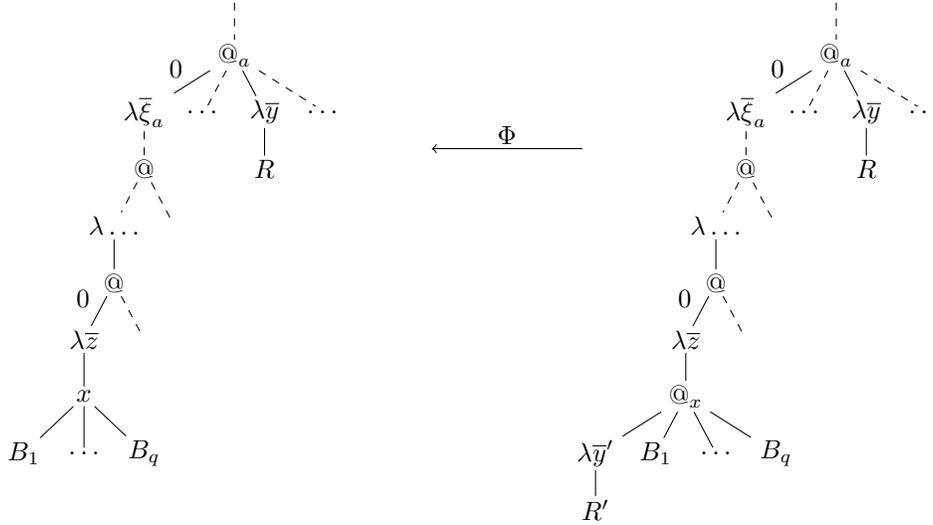

Observe that $\Phi$ has the following properties:
\begin{enumerate}[label=(\roman*)]
    \item $\Phi$ preserves the parent-child relation for lambda and variable nodes: for any lambda or variable node $n$, the child nodes of $n$'s image by $\Phi$ are the image of the $n$'s children. (This is not true for application node $@_x$, however);

    \item $\Phi$ is naming-consistent: any two variable nodes sharing the same name  in $N$ map to variable nodes with same name in $M$;

    \item $\Phi$ preserves variable binding: it maps the binder lambda node of a variable node $z$ in $N$ to the binder in $M$ of $\Phi(z)$;

    \item hence $\Phi$ preserves the enabling relation: if $n_1 \enables_i n_2$ in $N$ for $i\geq 0$, $n_1,n_2 \in \Nodes(N)$ then $\Phi(n_1) \enables_i \Phi(n_1)$ as long as $(n_1,n_2) \neq (@_x,\lambda\overline{y}')$.

    \item $\Phi$ preserves node types: images of external and internal nodes are respectively external and internal nodes; images of variable/@-nodes are variable/@-nodes; images of lambda nodes are lambda nodes; images of ghost and structural nodes are respectively ghost and structural;

    \item $\Phi$ preserves node arity. In particular $|x| = |@_x|$ since $@_x$'s operator does not contribute to its arity;

    \item $\Phi$ is not injective: since for any node in $R'$ there is a corresponding node in $R$ that is mapped to the same node in $M$.
\end{enumerate}

Let $\justseqset^P$ denote the subset $\justseqset$
consisting of justified sequences verifying the P-view-path correspondence:
the P-view at every occurrence of an internal node is precisely the path from the root to that node in the computation tree. By Property~\ref{prop:tree_path_charact} we have $\travset \subseteq \justseqset^P$.

We define $\phi$ as an extension of $\Phi$ to justified sequences of $\justseqset^P$ of $N$ as follows:
(i) maps all other occurrences to their image by $\Phi$;
(ii) it preserves all justification pointers except for $@_x$ and $\lambda\overline{y}'$;
(iii) inserts in-between each occurrence of $x$ and its immediately following lambda-node, the corresponding \emph{argument lookup} of $x$. Formally $\phi = \phi_1 \circ \phi_2$ where
\begin{align*}
\phi_1 \colon & \justseqset^P(N) \rightarrow \justseqset(M) \\
& \epsilon \mapsto \epsilon  \\
& t \cdot n \mapsto \phi(t) \ \cdot \Phi(n) &\mbox{$\Phi(n)$ has same link label and length as $n$, for $n\not\in \{ @_x, \lambda\overline{y} \}$}\\
 & t \cdot @_x \mapsto \phi(t) \ \cdot x &\mbox{$x$ points to the first occurrence of $x$'s binder in $\pview{t}$} \\
& t \cdot \lambda\overline{y}' \rightarrow \phi(t) \ \cdot \lambda\overline{y} &\mbox{$\lambda\overline{y}$ points to the last occurrence of $@_a$ in $\pview{t}$.}
\end{align*}
and $\phi_2$ is the function that inserts the argument lookup sequence of ghost nodes between each occurrence of $x$ and $\lambda\overline{y}$ in a justified sequence. Observe that the argument lookup of $x$ is entirely statically determined and thus can be reconstructed from any (visible) traversal $t_{\leq x}$ of $M$ by repeatedly instantiating rules \rulenamet{Lmd} and \rulenamet{Var} starting from $t_{\leq x}$ until reaching node $\lambda\overline{y}$.

The following example illustrates how $\phi$ is calculated on a generic
traversal $u$ of $N$ involving the lloc variable $x$:
\begin{eqnarray*}
    u &=& \Pstr[10pt]{ \ldots (aa){@_a} \cdot (la){\bulklambda{\xi_a}} \cdot \ldots (a1){@_1} \cdot (l1){\bulklambda{\xi_1}} \cdot @ \ldots \bulklambda{z} \cdot (ax){@_x} \cdot (ly-ax,20:0){\bulklambda{y}'} }
    \\
    \phi_1(u) &=& \Pstr[10pt]{ \ldots (aa){@_a} \cdot (la){\bulklambda{\xi_a}} \cdot \ldots (a1){@_1} \cdot (l1){\bulklambda{\xi_1}} \cdot @ \ldots \bulklambda{z} \cdot (x-l1,30:i) x \cdot (ly-a,25){\bulklambda{y}} }
    \\
    \phi(u) = \phi_2(\phi_1(u)) &=& \Pstr[10pt]{ \ldots (aa){@_a} \cdot (la){\bulklambda{\xi_a}} \cdot \ldots (a1){@_1} \cdot (l1){\bulklambda{\xi_1}} \cdot @ \ldots \bulklambda{z} \cdot (x-l1,30:i) x \cdot \ghostlmd \cdot \ghostvar \ldots \ghostlmd \cdot \ghostvar \cdot (ly-a,20:i_a){\bulklambda{y}} }
\end{eqnarray*}

Some immediate properties:
\begin{itemize}
\item $\phi_1$, and $\phi$ are well defined since by the P-view-path property, $x$'s binder and (resp.~$@_a$) must necessarily occurs in the P-view at $x$ (resp.~$\lambda\overline{y}'$).

\item The restriction of $\phi$ to $\coresymbol(\travset^\dagger(N))$
coincides precisely with the element-wise structure-preserving\footnote{Recall that two justified sequences over two distinct terms share the same structure just if their constituting nodes are of the same kind (variable/@ nodes or lambda node), have the same link labels, and same pointer lengths.} extension of $\Phi$ over justified sequences.
This is because $\phi_1$ only alters justification pointers of internal nodes ($@_x$ and $\lambda\overline y$) and $\phi_2$ only inserts internal nodes, while justified sequences in $\coresymbol(\travset^\dagger(N))$ do not contain any external node.

\item The restriction of $\phi$ to $\coresymbol(\travset^\dagger(N))$, is injective.
Suppose that $\coresymbol(u)\neq\coresymbol(v)$ for $u,v\in\travset^\dagger(N)$ then we must have $u\neq v$.
Let $s$ denote the longest common prefix of $u$ and $v$; and $n_1, n_2$ the nodes immediately following $s$ in $u$ and $v$
respectively. We have $s \rightarrow_N s \cdot n_1$ and $s \rightarrow_N s\cdot n_2$ with $n_1\ne n_2$ therefore $n_1, n_2$ are necessarily traversed
by the non-deterministic rule \rulenamet{IVar}. Hence they must both be external nodes and therefore they are not eliminated after applying $\coresymbol$: the last occurrence of $\coresymbol(s\cdot n_1)$ and $\coresymbol(s\cdot n_2)$ are different. Since $\Phi$ maps distinct external sibling nodes in $N$ to distinct sibling nodes in $M$, we necessarily have
$\phi(\coresymbol(s\cdot n_1))\neq \phi(\coresymbol(s\cdot n_2))$.

\item $\coresymbol$ and $\phi$ commute modulo the implicit renaming function $\rho: \VarSet^*(N) \rightarrow \VarSet^*(M)$ underlying $\Phi$. Formally, for any variable names $\alpha\in\VarSet^*$ occurring in $N$ we have:
$$
\phi(\coresymbol_\alpha(u)) = \coresymbol_{\rho(\alpha)}(\phi(u)) \ .
$$
This is shown by an easy induction and by assuming, without loss of generality, that variable names are not reused in either $M$ or $N$: each variable name is introduced by at most one lambda node in each term. (This condition can always be met via alpha-conversion.)
\end{itemize}

We now show that $M$ and $N$ are in $\phi$-bisimulation by induction on the traversals rules.

\emph{Base case}: $\rulename{Root}_\normalizing$ The empty traversal in $M$ is trivially the image by $\phi$ of the empty traversal in $N$. Same for the two singleton traversals consisting of the root nodes of $M$ and $N$.

\emph{Induction}: Let $u\in \travset^\dagger(N)$, $t\in\travset^\dagger(M)$ with $t = \phi (u)$ and $m, n$ denote respectively the last node in $t$ and $u$. We prove both directions by case analysis on the \emph{first} traversal rule used to extend $t$ (resp.~$u$).

Observe that the result holds trivially up until the traversal in $M$ reaches occurrence $x$ since one can use the exact same rules used to traverse $M$ in order to traverse the corresponding node in $N$, and reciprocally for traversals of $N$ until reaching node $@_x$. The crucial part of the induction takes place when traversing the \emph{lloc} variable $x$ (or $@_x$ in $N$).

\begin{itemize}[itemindent=0.5em, leftmargin=0.5em]
    \item $\rulename{App}_\normalizing$ We have $m = \Phi(n)$.

    (B1) Suppose $t$ extends to $t'$ with the application rule then $m, n \in\NodesApp$. Then the last occurrence $m'$ in $t'$ is the unique node verifying $m \enables_0 m'$.  Since $\Phi$ does not map $@_x$ to an application node we necessarily have  $n \neq @_x$. Thus $\phi$ preserves the enabling relation at $n$, and we can use the application rule in $N$ to extend $u$ into $u' = u \cdot n'$ where $n \enables_0 n'$ and $\phi(n') = m'$.

    (B2) Suppose $u'$ extends traversal $u$ with the application rule then the last occurrence $n$ in $u$ is an application node.
    If $n\neq@_x$ then $\phi$ preserves the enabling relation at $n$ and therefore the application rule can be used in $M$ to simulate $u'$.

    Otherwise, $n=@_x$ and since $t = \phi(u)$ we must have $m = \Phi(@_x) = x$.
    and $u' = u \cdot \lambda\overline{y}'$.
    Let $@_j$ and $\lambda\overline{\xi_j}$ for $1\leq j\leq a$, $a\geq 1$ denote respectively the application and lambda nodes in the spinal descent
    from $@_a$ to $x$'s binder. So that $\lambda\overline{\xi_1}$ is $x$'s binder, $@_1$ its parent,  and $@_a$ is $\lambda\overline{y}$'s parent.
    By Lemma~\ref{lemma:genredex_lookup}, $@_a$ is precisely the lowest application node in the spinal descent
    from $@_a$ to $\lambda\overline{\xi_1}$ verifying $|@_a| >  i_a$,
        where $i_a = i + \sum_{1\leq j< a} (|\lambda\overline{\xi_j}| - |@_j|)$.
    This formula precisely captures the ``pending argument lookup'' algorithm implicitly
    implemented by the rules \rulenamet{Lmd} and \rulenamet{Var}:
    with exactly $i_a$ instantiations of those two rules, traversal $t$ gets extended by $2(i_a-1)$ ghost nodes followed
    by the root of $A$:
            $$ t \rightarrow_M t \cdot v \cdot \lambda\overline{y}$$
    where  $v$ is a \emph{pending argument lookup}.
    Therefore by definition of $\phi$ if $v$ denote the argument lookup of $x$ we have:
    $\phi(u')
    = \phi(u \cdot \lambda\overline{y}')
    = \phi(u) \cdot \phi(\lambda\overline{y}')
    = t \cdot v \cdot \lambda\overline{y} \in \travset^\dagger(M)$.

    \item $\rulename{Lam}^@_\normalizing$
    (B1) follows from the fact that parent-child relation is preserved by $\Phi$
    (B2) Suppose $n \neq \lambda\overline{z}$ then it follows again from the fact that parent-child relation is preserved by $\Phi$.
    Otherwise, we have $n = \lambda\overline{z}$ and $u \rightarrow_N u \cdot \lambda\overline{y}'$ can be simulated in $M$ using rule $\rulename{Lam}^{\sf var}$: $t \rightarrow_M t \cdot x$.

    \item $\rulename{Lam}^{\sf var}_\normalizing$
    We have already shown that  $\phi$ preserves the local child-parent relationship for lambda nodes; node types, and arity of lambda nodes.
    It remains to show that any $\enables$-enabler of $m$ that is visible at  $u$ is mapped to an enabler of $n$ that is visible in $\phi(u)$.

    Since the last node of $t$ and $u$ are structural lambda nodes, their P-view do not contain any ghost node nodes, and in particular the argument lookups inserted by $\phi_2$ are not involved in the P-view calculation. Hence $\pview{\phi(u)} = \pview{\phi_1(u)}$.
    But by the P-view path characterization of traversals $\pview{\phi_1(u)}$ is precisely the path in $M$ from $\Phi(n)$ to the root of $M$,
    and $\pview{u}$ is the path in $N$ from $n$ to the root of $N$.

    If $n$ is not under $\lambda\overline{y}'.R'$ then the path to the root is not impacted by the reduction and therefore the path in $M$ and $N$ are identical modulo $\phi$: we thus have $\pview{\phi_1(u)} = \phi_1(\pview{u})$.
    Otherwise if $n$ is under $\lambda\overline{y}'.R'$ then the path from $m$ to the root in $M$ is precisely the element-wise image by $\Phi$ of the path from $n$ to the root in $N$ with the exclusion of the spinal descent from $@_a$ to $@_x$. But by construction of $N$, none of the variable in $R'$ are bound by the lambda nodes in the spinal descent.

    Hence, the rule $\rulename{Lam}^{\sf var}_\normalizing$ used to extend  $u$ can also be used to extend $t$, and reciprocally.

    \item $\rulename{Lam^\ghostlmd_\normalizing}$
    This case is excluded since $t$ and $u$ belong to $\travset^\dagger(M)$ and
    $\travset^\dagger(N)$ respectively and therefore their last occurrence cannot be a ghost lambda.

    \item $\rulename{Var_\normalizing}$
    (B2) Since $\Phi$ maps external variable nodes to external variable nodes,preserves the parent-child relation for variable nodes and $\phi$ preserves the justification pointer of external variable nodes, if this rule can extend traversal $u$ with node $n'$ then the same rule can be used to extend $t$ with $\phi(n')$.

    (B1) Suppose that $t$ ends with \emph{lloc} variable $x$ in which case $t' = t \cdot \lambda\overline{y}$, then
        we can just take $u = u' \cdot @_x$ and conclude exactly like we did for case $\rulename{App}_\normalizing$ (B2).
    Otherwise, if $m\neq x$, we conclude like in case (B1).

    \item $\rulename{IVar_\normalizing}$ $\Phi$ preserves node types, arity of variable nodes. It remains to show that it preserves \emph{arity threshold} of traversals: In an argument lookup, all occurrences are ghost with arity $0$ therefore they do not affect the arity threshold calculation. Hence $\arth(\phi(u))=\arth(\phi_1(u))$. Further $\Phi$ preserves arity and maps external and internal nodes to external and internal nodes respectively, thus $\arth(\phi_1(u))=\phi_1(\arth(u))$.
    Hence the rule $\rulename{Var_\normalizing}$ used to extend $u$ can be used to extend $t$, and reciprocally.

\end{itemize}

\item[Case 2] where $x$ has at least one operand ($|x|>0$) and $A$ is an application ($|\overline{y}|=0$). The root lambda node of $A$ is therefore a dummy lambda: $A = \lambda. R$. Let $r$ denote the root node of $R$, which is either a variable if $R$ is in head-normal form, or an $@$-node otherwise.
The tree of the reduct $N$ is obtained from $\ctree(M)$
by substituting $x$'s label with that of $r$ (modulo variable renaming to avoid capture if $r$ is a variable node) and prepending $r$'s children to existing $x$'s children. Other nodes, including $x$'s parent lambda node, remain untouched. Let $r'$ denote this updated node in the reduct $N$.
This transformation yields an implicit mapping function $\Phi : \ExtendedNodes(N)\rightarrow \ExtendedNodes(N)$ such that $\Phi(r') = x$.
Although $\Phi$ does not preserve the arity of $x$, it does not affect the arity threshold of traversals since $|r'| = |x| + |r|$.
We can then extend $\Phi$ to a function $\phi$ on justified sequence and show by induction that it yields a bisimulation similarly to the previous case.

\item[Case 3] where $x$ is unapplied ($|x|=0$).
After reduction, node $x$ gets replaced by subterm $A$ and the label of $x$'s parent lambda node gets concatenated with the label of $A$'s root node:
the parent lambda node $\lambda\overline{z}$ of $x$ in $M$ becomes $\lambda\overline{z}\overline{y}'$ in the reduct $N$.
Like in the first case, we can define a function $\phi$ such that
 traversals of $M$ of the form $\ldots \lambda\overline{z}\cdot x \cdot v \cdot \lambda\overline{y} \cdot \ldots$, where $v$ is an argument lookup, are images of traversals in $N$ of the form $\ldots \cdot \lambda\overline{z}\overline{y}' \ldots$. Once again, although the arity of $x$ is not preserved, this does not affect the arity threshold because $|\lambda{\overline{z} \overline{y}'}| =
 |\lambda{\overline{z}}| +  |\lambda{\overline{y}}|$.
\end{description}


Hence $M$ and $N$ are in $\phi$-bisimulation, thus because $\phi$ is injective and commutes with $\coresymbol$, by Lemma~\ref{lem:bisimulation_isomorphism} the two sets of traversals are isomorphic.
\endproofatend

\begin{proposition}[Traversals are sound for trivial reduction]
\label{prop:ulctrav_sound_for_trivialreduction}
    Let $M$ be an untyped term containing only trivial redexes, and $N$ the reduct obtained by firing the \emph{leftmost spinal-innermost} trivial redex, then $\coresymbol(\travsetnorm(M)) \cong \coresymbol(\travsetnorm(N))$.
\end{proposition}
\proofatend
Let's consider the tree structure at the trivial redex on Figure~\ref{fig:reducing_trivialredex}. It is of the form $(\lambda \overline{x} . T) A_1 \ldots A_r$ for $r\geq1$ where $(x_1, A_1)$ forms the leftmost trivial redex.
    Two cases must be considered:

    \begin{figure}[htbp]
        \centering
        \begin{tikzpicture}[baseline=(root.base),level distance=5ex,inner ysep=0.5mm,sibling distance=13mm]
        \node (root_rGt1_M){}
                child[dashed]{node{$\lambda\overline\eta$}
                    child[solid]{node{$@$}
                        child[solid]{
                            node{$\lambda x_1 x_2 \ldots x_n$}
                            child[solid,level distance=5ex]
                                {node {$T$}}
                            edge from parent node[above left]{$0$}
                        }
                        child{node{$A_1$}}
                        child[dashed]{node{$\ldots$}}
                        child{node{$A_r$}}
                    }
                };
        \node (root_rGt1_N) at (7,0) {}
                child[dashed]{node{$\lambda\overline\eta$}
                    child[solid]{node{$@$}
                        child[solid]{
                            node{$\lambda x_2 \ldots x_n$}
                            child[solid,level distance=5ex]
                                {node {$N$}}
                            edge from parent node[above left]{$0$}
                        }
                        child{node{$A_2$}}
                        child[dashed]{node{$\ldots$}}
                        child{node{$A_r$}}
                    }
                };
        \node (root_rEq1_M) at (0,-3) {}
                child[dashed]{node{$\lambda\overline\eta$}
                    child[solid]{node{$@$}
                        child[solid]{
                            node{$\lambda x_1 x_2 \ldots x_n$}
                            child[solid,level distance=5ex]
                                {node {$T$}}
                            edge from parent node[above left]{$0$}
                        }
                        child{node{$A_1$}}
                    }
                };
        \node (root_rEq1_N) at (7,-4) {}
                child[dashed]{node{$\lambda\overline\eta~x_2 \ldots x_n$}
                    child[solid]{node{$N$}
                    }
                }
        ;
        \draw[->] ([xshift=1.5cm,yshift=-1cm]root_rGt1_M.east) -- +(3.5cm,0)
        node[midway, above] {trivial reduction};
        \draw[->] ([xshift=1.5cm,yshift=-1.5cm]root_rEq1_M.east) -- +(3.5cm,0)
            node[midway, above] {trivial reduction} ;
        \end{tikzpicture}

        \caption{Reducing a (leftmost spinal-innermost) trivial redex. Case $r>1$ (top) and $r=1$ (bottom).}
\label{fig:reducing_trivialredex}
    \end{figure}
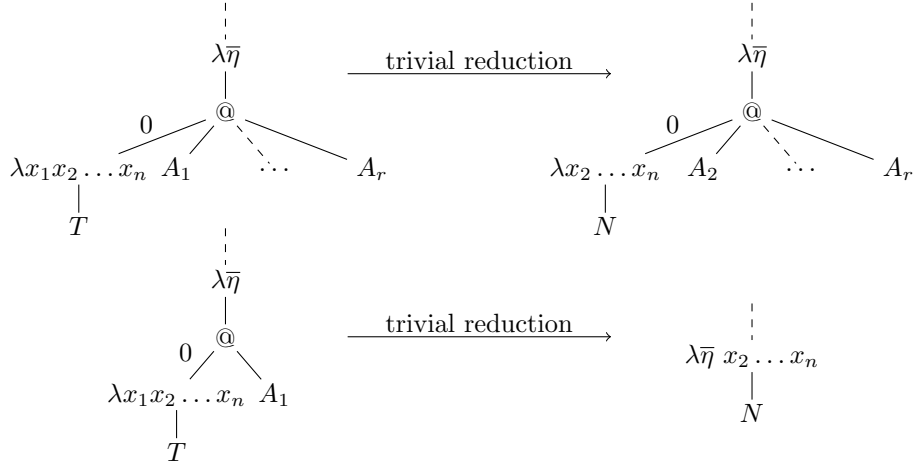

    \begin{description}
    \item[Case $r>1$] Because the structure of the tree is preserved after reduction, there is an implicit injection $\Phi$ from nodes of the reduct $N$ to nodes of $M$. Let $\phi$ denote its element-wise structure-reserving  extension to justified sequences.

    Firstly, $M$ and $N$ are in $\phi$-bisimulation. We omit the details, the key observations are:
    (i) Since $x_1$ does not occur in $T$, the term $A_1$ will never be traversed in $N$.
    (ii) Only nodes hereditarily justified by $\lambda x_1 \ldots x_n$ see their label (or link label) changed after applying $\Phi$. But since those are external nodes, they get filtered out by the core projection $\coresymbol$.
    (iii) The arity of node $\lambda x_2 \ldots x_n$ increases by $1$ after applying $\Phi$, but it does not impact the arity threshold of traversals because the arity of the application node also increases by $1$.

    Secondly, since $\Phi$ is injective, so is its element-wise extension to justified sequences.

    Finally, $\phi$ and $\coresymbol$ commute. This is shown by an easy induction, observing that nodes $@$ and $\lambda x_2\ldots x_n$ necessarily occur consecutively in a traversal of $N$, and because $|@| = |\Phi(@)| -1$ and $|\lambda x_2\ldots x_n| = |\Phi(\lambda x_2\ldots x_n)| - 1$, the core projection is preserved modulo $\Phi$.

    Hence by Lemma~\ref{lem:bisimulation_isomorphism} the two sets of traversals are isomorphic.

    \item[Case $r=1$] Depicted on Figure~\ref{fig:reducing_trivialredex} (bottom). The tree structure is altered: the application node and the argument $A_1$ are eliminated in the reduct. There is an implicit injection $\Phi$ from nodes of $N$ to nodes of $M$. We extent this function to \emph{occurrences} of nodes as follows. Let's introduce the abbreviation $\overline{y} = \overline\eta x_2 \ldots x_n$.
    \begin{align*}
        \Phi\colon O_4(N) &\rightarrow O_4(M) \\
        (n,p,d,k) &\longmapsto (\Phi(n),p,d,k) & \mbox{if $n\neq\lambda\overline{y}$} \\
        (\lambda\overline{y},p,d,k) &\longmapsto (\lambda\overline\eta,(x_2 \ldots x_n \cdot p),d,k)
    \end{align*}
    So in particular, $\Phi$ maps occurrences of $\lambda\overline{y}$ to occurrences of $\lambda\overline\eta$ with binding variables $x_2 \ldots x_n$ appended to $\overline\eta$.
    \begin{equation}
        \Phi(\lambda\overline{y})
        =
        \lambda\overline\eta^{[x_2\ldots x_n]}
        \label{eqn:phi_maps_lambday_to_lambdaeta}
    \end{equation}
    Observe also that $\Phi$ commutes with the rebinding operation: because in the case of $\lambda\overline{y}$, $\Phi$ appends variable names to the left whereas the rebinding operation appends variables names to the right.

    We extend $\Phi$ to justified sequences by taking the element-wise extension $\phi$ with the following additional modification: any occurrence of node $\lambda\overline{y}$ in a traversal of $N$ gets expanded by $\phi$ into the three occurrences $\lambda\overline\eta \cdot @ \cdot \lambda\overline{x}$ in $M$.

    We have the following properties:
    \begin{itemize}
    \item $\phi$ is injective. It follows from the fact that $\Phi$ is itself injective.

    \item $\phi$ and $\coresymbol$ commute. This is because $\phi$ carries over the list of pending lambdas and the core symbol operation adequately implements merging of nodes $\lambda\overline\eta$ and
    $\lambda\overline x_1 \ldots x_n$ into
    $\lambda\overline\eta x_2 \ldots x_n$. Formally:

    Let $u$ be a traversal of $\travset^\dagger(N)$
    ending with $\lambda\overline{y}$ that is
    $u = u_1 \cdot \lambda\overline{y} $ for some traversal $u_1$.
    The sequence $\phi(u)$ is of the form
    $\phi(u) = \phi(u_1) \cdot \lambda\overline{\eta} \cdot @ \cdot \lambda\overline{x}$. Let $\overline\alpha$ be the list of pending lambdas at $u_{\leq \lambda\overline{y}}$. There are two cases:

    Suppose $\lambda\overline\eta$ is internal. We have:
    \begin{align*}
        \coresymbol_{\overline\alpha}(\phi(u))
        &=  \coresymbol_{\overline\alpha}(\phi(u_1) \cdot \lambda\overline{\eta} \cdot @ \cdot \lambda\overline{x})
            & \mbox{(def.~of $\phi(u)$)}
        \\
        &=  \coresymbol_{\overline{x}\overline\alpha}(\phi(u_1) \cdot
        \lambda\overline\eta \cdot @)
            & \mbox{(def.~of $\coresymbol$)}
        \\
        &=  \coresymbol_{pop_1(\overline{x}\overline\alpha)}(\phi(u_1) \cdot
        \lambda\overline{\eta})
            & \mbox{(def.~of $\coresymbol$ and $|@|=1$)}
        \\
        &=  \coresymbol_{\overline\eta \cdot pop_1(\overline{x}\overline\alpha)}(\phi(u_1))
            & \mbox{(def.~of $\coresymbol$)}
        \\
        &=  \coresymbol_{\overline{y}\overline\alpha}(\phi(u_1))
            & \mbox{(def.~of $\overline{y}$)}
        \\
        &=  \phi(\coresymbol_{\overline{y}\overline\alpha}(u_1))
            & \mbox{(by induction hypothesis)}
        \\
        &= \phi(\coresymbol_{\overline\alpha}(u_1 \cdot \lambda\overline{y})) & \mbox{(def of $\coresymbol$)} \\
        &=  \phi(\coresymbol_{\overline\alpha}(u))
            & \mbox{(def.~of $u$).}
    \end{align*}

    Suppose $\lambda\overline\eta$ is external. We have:
    \begin{align*}
        \coresymbol_{\overline\alpha}(\phi(u))
        &=  \coresymbol_{\overline\alpha}(\phi(u_1) \cdot \lambda\overline{\eta} \cdot @ \cdot \lambda\overline{x})
            & \mbox{(def.~of $\phi(u)$)}
        \\
        &=  \coresymbol_{\overline{x}\overline\alpha}(\phi(u_1) \cdot
        \lambda\overline\eta \cdot @)
            & \mbox{(def.~of $\coresymbol$)}
        \\
        &=  \coresymbol_{pop_1(\overline{x}\overline\alpha)}(\phi(u_1) \cdot \lambda\overline\eta)
            & \mbox{(def.~of $\coresymbol$ and $|@|=1$)}
        \\
        &=  \coresymbol_{x_2\ldots x_n \overline\alpha}(\phi(u_1) \cdot \lambda\overline\eta)
        & \mbox{(def.~of $pop$)}
        \\
        &=  \coresymbol_{\epsilon}(\phi(u_1)) \cdot \lambda\overline\eta^{[x_2\ldots x_n \overline\alpha]}
            & \mbox{(def.~of $\coresymbol$)}
    \end{align*}
    And:
    \begin{align*}
        \phi(\coresymbol_{\overline\alpha}(u))
        &= \phi(\coresymbol_{\overline\alpha}(u_1 \cdot \lambda\overline{y}))
        & \mbox{(def.~of $u$)}
    \\
        &= \phi(\coresymbol_{\epsilon}(u_1) \cdot \lambda\overline{y}^{[\overline\alpha]})
        & \mbox{(def.~of $\coresymbol$)}
    \\
        &= \phi(\coresymbol_{\epsilon}(u_1)) \cdot \phi(\lambda\overline{y}^{[\overline\alpha]})
        & \mbox{(def.~of $\phi$)}
    \\
        &=  \coresymbol_{\epsilon}(\phi(u_1)) \cdot \phi(\lambda\overline{y}^{[\overline\alpha]})
        & \mbox{(commutativity, by I.H.)}
    \\
        &=  \coresymbol_{\epsilon}(\phi(u_1)) \cdot \phi(\lambda\overline{y})^{[\overline\alpha]}
        & \mbox{($\phi$ commutes with rebinding)}
    \\
    &=  \coresymbol_{\epsilon}(\phi(u_1)) \cdot (\lambda\overline\eta^{[x_2\ldots x_n]})^{[\overline\alpha]}
        & \mbox{(by Eqn.~\ref{eqn:phi_maps_lambday_to_lambdaeta})}
    \\
        &=  \coresymbol_{\epsilon}(\phi(u_1)) \cdot\lambda\overline\eta^{[x_2\ldots x_n \overline\alpha]}
        & \mbox{(rebinding composition)}
    \end{align*}
    Hence $\phi(\coresymbol_{\overline\alpha}(u)) = \coresymbol_{\overline\alpha}(\phi(u))$.

    \item $\rightarrow_M$ can simulate $\rightarrow_N$ and reciprocally. This is shown by induction on the traversal rules. One interesting case is when ghost nodes in $M$ get materialized into structural nodes in $N$:
    Consider the traversal of $M$ for some $i\geq1$:
    $t = \Pstr[10pt]{
        (l){\lambda\overline\eta} \cdot (a)@ \cdot (lx-a,25:0){\lambda\overline{x}} \cdot \ldots \cdot (gl-a,30:i)\ghostlmd \cdot (gv-l,32:{|\eta|-1+i})\ghostvar }
    $.
    If $i>n$ then corresponding traversal in $N$
    also ends with a ghost variable:
    $u = \Pstr[10pt]{
        (l){\lambda\overline\eta x_2 \ldots x_n} \cdot \ldots \cdot  (gv-l,32:{|\eta|-1+i})\ghostvar }
    $
    If $i\leq n$ then the ghost nodes $\ghostvar$ gets materialized into variable occurrence $x_i$ in $N$:
    $u = \Pstr[10pt]{
        (l){\lambda\overline\eta x_2\cdots x_n} \cdot \ldots \cdot  (gv-l,32:{|\eta|+i-1}){x_i} }
    $
    \end{itemize}
    We can then conclude by Lemma~\ref{lem:bisimulation_isomorphism} that the two sets of traversals are isomorphic.
    \end{description}
\endproofatend

\begin{example}
    The term
    $M = \lambda x. (\lambda y z.z) u$ from Example~\ref{examp:ghost_materialization} trivially reduces to $N = \lambda x z . z$. Take the traversal $t$ from Example~\ref{examp:ghost_materialization}, then we have $\coresymbol(t) = \Pstr[10pt]{ (l){\lambda x z} \cdot (gv-l,32:2)\ghostvar }$ which is equivalent to
$\Pstr[10pt]{(l){\lambda x^{[z]}} \cdot (gv-l,32:2)\ghostvar }$, itself equivalent to traversal $\Pstr[10pt]{ (l){\lambda x^{[z]}} \cdot (gv-l,32:2){z} }$ in $N$.
\end{example}

\begin{corollary}
\label{cor:qnf_and_nf_traveset_invariant}
If $M$ is in \emph{qnf} and has beta-normal form $N$ then $\coresymbol(\travsetnorm(M)) \cong \coresymbol(\travsetnorm(N))$.
\end{corollary}
\begin{proof}
By Proposition~\ref{prop:qnf_nf}(ii) \emph{qnf} trivially reduce to their beta-nf: there is a reduction sequence $M = M_1 \rightarrow_\beta M_2 \ldots \rightarrow_\beta M_q= N$, $q\geq1$  where in each reduction step, $M_{i+1}$ is obtained by reducing the leftmost spinal-innermost standard redex of $M_i$ (of the form $(\lambda x . U) T$ such that $U$ does not contain any redex). By Proposition~\ref{prop:qnf_nf}(i), each such redex is trivial. By Proposition~\ref{prop:ulctrav_sound_for_trivialreduction} the set $\core{\travsetnorm}$ is an invariant (up to isomorphism) throughout this trivial reduction sequence, which yields the desired result.
\end{proof}

\begin{proposition}
\label{prop:ulc_travnorm_finite}
Let $M$ be an untyped term with a beta-normal form. Then
\begin{enumerate*}[label=(\roman*)]
\item All traversals in $\travsetnorm(M)$ are finite;
\item The set $\travsetnorm(M)$ is finite.
\end{enumerate*}
\end{proposition}
\begin{proof}
(i) If $M$ has a normal form then by Theorem~\ref{thm:completeness_leftmostlinearred} its linear reduction sequence yields a quasi-normal form $Q$ and by Theorem~\ref{thm:soundness_leftmostlinearred} $Q$ and $M$ are beta-equivalent. Thus by Prop.~\ref{prop:ulctrav_impl_linear_reduction} the sets of traversals of $M$ and $Q$ are isomorphic. So if $M$ has an infinite traversal then so does $Q$, which contradicts Proposition~\ref{prop:qnf_traversals_are_finite}.
(ii) Traversals rules defining $\travsetnorm$ all have bounded non-determinism therefore by (i) $\travsetnorm$ is finite.
\end{proof}


\begin{proposition}
Algorithm~\ref{algo:ulc_normalization_by_traversals} terminates.
\end{proposition}
\begin{proof}
By Proposition~\ref{prop:ulc_travnorm_finite} there is a finite number of traversals in $\travsetnorm$ and each of them can be obtained with finitely many applications of traversal rules.
\end{proof}

\begin{theorem}[ULC Paths Characterization]
\label{thm:path_charact_ulc}
Let $M$ be an untyped term with normal form $T$.
$$\travsetnorm(M)/{\sim} \ \cong\ \pathset(T) \ .$$
\end{theorem}
\begin{proof}
Let's first assume that $M$ is in beta-normal form. Then its computation tree cannot contain any $@$-node and all nodes must be external.
Hence the arity threshold of a traversal ending with external variable $x$ is precisely $|x|$, the arity of $x$. A trivial induction on the traversal rules shows, that since the \rulenamet{Var} is never used, we necessarily have  $\travsetnorm(M) = \pathset(M)$. Further, since all nodes are externals, the projection $\coresymbol$ and $\pview{\_}$ functions both coincides with the identity function thus $\pview{\core{\travsetnorm(M)}} = \travsetnorm(M) = \pathset(M)$.

Otherwise, since $M$ has a beta-nf, by Theorem~\ref{thm:completeness_leftmostlinearred} its \emph{leftmost linear reduction sequence} terminates and yields a \emph{qnf}. By Prop.~\ref{prop:ulctrav_impl_linear_reduction} the set of traversals is preserved by \emph{linear reduction} (up to an isomorphism) thus the set of traversals of the \emph{qnf} is isomorphic to the set of traversals of $M$. Corollary~\ref{cor:qnf_and_nf_traveset_invariant} permits us to conclude.
\end{proof}

\paragraph{Soundness} of the normalization procedure follows from Theorem~\ref{thm:path_charact_ulc} and the fact that a term is uniquely characterized by its set of justified paths $\pathset$ (Property~\ref{prop:tree_path_charact}).


\section{Conclusion and further directions}

We presented a novel method to evaluate untyped lambda terms by generalizing the theory of traversals, originally introduced in a typed setting, to the untyped lambda calculus. We introduced the leftmost linear reduction, a generalization of the non-standard \emph{head-linear reduction} strategy\cite{danos-head}, and showed soundness of the evaluation procedure by showing that traversals implement leftmost linear reduction.

\subsection*{Connection with Berezun-Jones traversals of ULC}

Daniil Berezun and Neil Jones were first to introduce a notion of traversals for the untyped lambda calculus \cite{JonesBerezunLLL}. Their presentation has notable differences with ours. Firstly, they require two types of pointers: binding pointers and control pointers, instead of a single justification pointer for imaginary traversals.
The control pointers would be what Danos-Regnier call the ``price to pay for having no stacks or environments'' in their definition of the argument lookup transition of the Pointer Abstract Machine\cite{danos-head}.

In imaginary traversals, however, node occurrence have at most one justification pointer. Instead, we pay for the absence of stacks and environment by introducing ghost occurrences. In particular,
we performs the necessary ``gymnastics to find the argument of a subterm''\cite{danos-head} by traversing the `pending argument lookup` (Def.~\ref{def:spinaldescent_pendingarglookup}) consisting of ghost occurrences, until
a structural node gets materialized\footnote{The ghost materialization loop  implemented by
rule \rulenamet{Lam^\ghostlmd} and \rulenamet{Var} essentially calculates the quantity ``$|\alpha|+i-|m|$'' at each iteration, which is the exact same quantity ($r-a+l$) used in step 2. (b) of the PAM transition from \cite{danos-head}.}

Another difference is that Berezun-Jones define a single traversal for a given term, whereas our definition involves one traversal for every path in the resulting normal form. The `read-out' algorithm to reconstruct the normal form from the traversals (to be defined in a forthcoming paper by the same authors) involves recursively applying rewriting rules to the traversal. In contrast, our read-out algorithm reconstructs the tree representation of the normal from from tree-path extracted from each maximal traversal.

The connection between the two presentations lies in their respective correspondence with head-linear reduction: we have shown that imaginary traversals essentially implement the recursive application of the head-linear reduction. Berezun and Jones established a similar correspondence in their forthcoming paper.

\subsection*{Game semantics connection}

In the typed setting, Danos \etal showed the connection between game models of typed lambda and \emph{head-linear reduction} and proposed \emph{Pointer Abstract Machine (PAM)} as an abstract machine implementation of this reduction~\cite{danosherbelinregnier1996}.
More recently, the theory of traversals \cite{OngLics2006} was shown to correspond to Game Semantics in a very concrete way for typed languages like the simply-typed lambda calculus, PCF and Idealized Algol \cite{BlumPhd}: the set of traversals is in one to one correspondence with the game denotation of a term where all internal moves are revealed. We conjecture that such Game Semantic correspondence can be generalized to untyped terms using the game model of the untyped lambda calculus introduced in Andrew Ker's thesis~\cite{KerThesis} which denotes lambda-terms by \emph{effectively almost everywhere copycat} strategies (Conjecture~\ref{conj:ulc_corresp}).





\subsection*{Acknowledgement}
This work got inspired by the slides of the Galop 2016 presentation by Daniil Berezun and Neil D. Jones on Partial Evaluation and Normalization by Traversals \cite{berezunjones_partialevalbytraversals}. It benefited from email communications with the authors as well as memorable discussion with Neil D. Jones during his visit at Microsoft Research in May 2017.

\bibliographystyle{abbrv}
\bibliography{ulc-trav}

\begin{appendices}
\section{Tables}

\newgeometry{left=1cm,right=2cm, bottom=2cm, top=2cm}
\begin{landscape}
\thispagestyle{empty}
\begin{table}
\resizebox{1\hsize}{!}{$t_{11} = \Pstr[0.7cm]{(n0){\lambda }\ (n1){@}\ (n2-n1){\lambda t}\ (n3-n2){t}\ (n4-n1){\lambda s_2 z_2}\ (n5-n4){s_2}\ (n6-n3){\lambda n a x}\ (n7-n6){n}\ (n8-n5){\lambda }\ (n9-n4){s_2}\ (n10-n3){\lambda n a x}\ (n11-n10){n}\ (n12-n9){\lambda }\ (n13-n4){z_2}\ (n14-n3){\lambda a}\ (n15-n14){a}\ (n16-n13){{\ghostlmd^{1}}}\ (n17-n12){{\ghostvar^{1}}}\ (n18-n11){\lambda s z}\ (n19-n10){a}\ (n20-n9){{\ghostlmd^{2}}}\ (n21-n8){{\ghostvar^{1}}}\ (n22-n7){\lambda s z}\ (n23-n6){a}\ (n24-n5){{\ghostlmd^{2}}}\ (n25-n4){{\ghostvar^{3}}}\ (n26-n3){\lambda z_0}\ (n27-n26){z_0}\ (n28-n25){{\ghostlmd^{1}}}\ (n29-n24){{\ghostvar^{1}}}\ (n30-n23){\lambda }\ (n31-n22){s}\ (n32-n21){{\ghostlmd^{1}}}\ (n33-n20){{\ghostvar^{1}}}\ (n34-n19){\lambda }\ (n35-n18){s}\ (n36-n17){{\ghostlmd^{1}}}\ (n37-n16){{\ghostvar^{1}}}\ (n38-n15){{\ghostlmd^{1}}}\ (n39-n14){{\ghostvar^{2}}}\ (n40-n13){{\ghostlmd^{2}}}\ (n41-n12){{\ghostvar^{2}}}\ (n42-n11){{\ghostlmd^{2}}}\ (n43-n10){{\ghostvar^{4}}}\ (n44-n9){{\ghostlmd^{4}}}\ (n45-n8){{\ghostvar^{3}}}\ (n46-n7){{\ghostlmd^{3}}}\ (n47-n6){{\ghostvar^{5}}}\ (n48-n5){{\ghostlmd^{5}}}\ (n49-n4){{\ghostvar^{6}}}\ (n50-n3){{\ghostlmd^{6}}}\ (n51-n2){{\ghostvar^{4}}}\ (n52-n1){{\ghostlmd^{4}}}\ (n53-n0){{\ghostvar^{3}}}\ (n54-n53){{\ghostlmd^{1}}}\ (n55-n52){{\ghostvar^{1}}}\ (n56-n51){{\ghostlmd^{1}}}\ (n57-n50){{\ghostvar^{1}}}\ (n58-n49){{\ghostlmd^{1}}}\ (n59-n48){{\ghostvar^{1}}}\ (n60-n47){{\ghostlmd^{1}}}\ (n61-n46){{\ghostvar^{1}}}\ (n62-n45){{\ghostlmd^{1}}}\ (n63-n44){{\ghostvar^{1}}}\ (n64-n43){{\ghostlmd^{1}}}\ (n65-n42){{\ghostvar^{1}}}\ (n66-n41){{\ghostlmd^{1}}}\ (n67-n40){{\ghostvar^{1}}}\ (n68-n39){{\ghostlmd^{1}}}\ (n69-n38){{\ghostvar^{1}}}\ (n70-n37){{\ghostlmd^{1}}}\ (n71-n36){{\ghostvar^{1}}}\ (n72-n35){{\ghostlmd^{1}}}\ (n73-n34){{\ghostvar^{1}}}\ (n74-n33){{\ghostlmd^{1}}}\ (n75-n32){{\ghostvar^{1}}}\ (n76-n31){{\ghostlmd^{1}}}\ (n77-n30){{\ghostvar^{1}}}\ (n78-n29){{\ghostlmd^{1}}}\ (n79-n28){{\ghostvar^{1}}}\ (n80-n27){{\ghostlmd^{1}}}\ (n81-n26){{\ghostvar^{2}}}\ (n82-n25){{\ghostlmd^{2}}}\ (n83-n24){{\ghostvar^{2}}}\ (n84-n23){\lambda }\ (n85-n6){x}\ (n86-n5){{\ghostlmd^{3}}}\ (n87-n4){{\ghostvar^{4}}}\ (n88-n3){{\ghostlmd^{4}}}\ (n89-n2){{\ghostvar^{2}}}\ (n90-n1){{\ghostlmd^{2}}}\ (n91-n0){{\ghostvar^{1}}}\ (n92-n91){{\ghostlmd^{1}}}\ (n93-n90){{\ghostvar^{1}}}\ (n94-n89){{\ghostlmd^{1}}}\ (n95-n88){{\ghostvar^{1}}}\ (n96-n87){{\ghostlmd^{1}}}\ (n97-n86){{\ghostvar^{1}}}\ (n98-n85){\lambda }\ (n99-n22){s}\ (n100-n21){{\ghostlmd^{1}}}\ (n101-n20){{\ghostvar^{1}}}\ (n102-n19){\lambda }\ (n103-n18){s}\ (n104-n17){{\ghostlmd^{1}}}\ (n105-n16){{\ghostvar^{1}}}\ (n106-n15){{\ghostlmd^{1}}}\ (n107-n14){{\ghostvar^{2}}}\ (n108-n13){{\ghostlmd^{2}}}\ (n109-n12){{\ghostvar^{2}}}\ (n110-n11){{\ghostlmd^{2}}}\ (n111-n10){{\ghostvar^{4}}}\ (n112-n9){{\ghostlmd^{4}}}\ (n113-n8){{\ghostvar^{3}}}\ (n114-n7){{\ghostlmd^{3}}}\ (n115-n6){{\ghostvar^{5}}}\ (n116-n5){{\ghostlmd^{5}}}\ (n117-n4){{\ghostvar^{6}}}\ (n118-n3){{\ghostlmd^{6}}}\ (n119-n2){{\ghostvar^{4}}}\ (n120-n1){{\ghostlmd^{4}}}\ (n121-n0){{\ghostvar^{3}}}}$}
\resizebox{1\hsize}{!}{$t_{121} =
\Pstr[0.7cm]{(n0){\lambda }\ (n1){@}\ (n2-n1){\lambda t}\ (n3-n2){t}\ (n4-n1){\lambda s_2 z_2}\ (n5-n4){s_2}\ (n6-n3){\lambda n a x}\ (n7-n6){n}\ (n8-n5){\lambda }\ (n9-n4){s_2}\ (n10-n3){\lambda n a x}\ (n11-n10){n}\ (n12-n9){\lambda }\ (n13-n4){z_2}\ (n14-n3){\lambda a}\ (n15-n14){a}\ (n16-n13){{\ghostlmd^{1}}}\ (n17-n12){{\ghostvar^{1}}}\ (n18-n11){\lambda s z}\ (n19-n10){a}\ (n20-n9){{\ghostlmd^{2}}}\ (n21-n8){{\ghostvar^{1}}}\ (n22-n7){\lambda s z}\ (n23-n6){a}\ (n24-n5){{\ghostlmd^{2}}}\ (n25-n4){{\ghostvar^{3}}}\ (n26-n3){\lambda z_0}\ (n27-n26){z_0}\ (n28-n25){{\ghostlmd^{1}}}\ (n29-n24){{\ghostvar^{1}}}\ (n30-n23){\lambda }\ (n31-n22){s}\ (n32-n21){{\ghostlmd^{1}}}\ (n33-n20){{\ghostvar^{1}}}\ (n34-n19){\lambda }\ (n35-n18){s}\ (n36-n17){{\ghostlmd^{1}}}\ (n37-n16){{\ghostvar^{1}}}\ (n38-n15){{\ghostlmd^{1}}}\ (n39-n14){{\ghostvar^{2}}}\ (n40-n13){{\ghostlmd^{2}}}\ (n41-n12){{\ghostvar^{2}}}\ (n42-n11){{\ghostlmd^{2}}}\ (n43-n10){{\ghostvar^{4}}}\ (n44-n9){{\ghostlmd^{4}}}\ (n45-n8){{\ghostvar^{3}}}\ (n46-n7){{\ghostlmd^{3}}}\ (n47-n6){{\ghostvar^{5}}}\ (n48-n5){{\ghostlmd^{5}}}\ (n49-n4){{\ghostvar^{6}}}\ (n50-n3){{\ghostlmd^{6}}}\ (n51-n2){{\ghostvar^{4}}}\ (n52-n1){{\ghostlmd^{4}}}\ (n53-n0){{\ghostvar^{3}}}\ (n54-n53){{\ghostlmd^{1}}}\ (n55-n52){{\ghostvar^{1}}}\ (n56-n51){{\ghostlmd^{1}}}\ (n57-n50){{\ghostvar^{1}}}\ (n58-n49){{\ghostlmd^{1}}}\ (n59-n48){{\ghostvar^{1}}}\ (n60-n47){{\ghostlmd^{1}}}\ (n61-n46){{\ghostvar^{1}}}\ (n62-n45){{\ghostlmd^{1}}}\ (n63-n44){{\ghostvar^{1}}}\ (n64-n43){{\ghostlmd^{1}}}\ (n65-n42){{\ghostvar^{1}}}\ (n66-n41){{\ghostlmd^{1}}}\ (n67-n40){{\ghostvar^{1}}}\ (n68-n39){{\ghostlmd^{1}}}\ (n69-n38){{\ghostvar^{1}}}\ (n70-n37){{\ghostlmd^{1}}}\ (n71-n36){{\ghostvar^{1}}}\ (n72-n35){{\ghostlmd^{1}}}\ (n73-n34){{\ghostvar^{1}}}\ (n74-n33){{\ghostlmd^{1}}}\ (n75-n32){{\ghostvar^{1}}}\ (n76-n31){{\ghostlmd^{1}}}\ (n77-n30){{\ghostvar^{1}}}\ (n78-n29){{\ghostlmd^{1}}}\ (n79-n28){{\ghostvar^{1}}}\ (n80-n27){{\ghostlmd^{1}}}\ (n81-n26){{\ghostvar^{2}}}\ (n82-n25){{\ghostlmd^{2}}}\ (n83-n24){{\ghostvar^{2}}}\ (n84-n23){\lambda }\ (n85-n6){x}\ (n86-n5){{\ghostlmd^{3}}}\ (n87-n4){{\ghostvar^{4}}}\ (n88-n3){{\ghostlmd^{4}}}\ (n89-n2){{\ghostvar^{2}}}\ (n90-n1){{\ghostlmd^{2}}}\ (n91-n0){{\ghostvar^{1}}}\ (n92-n91){{\ghostlmd^{2}}}\ (n93-n90){{\ghostvar^{2}}}\ (n94-n89){{\ghostlmd^{2}}}\ (n95-n88){{\ghostvar^{2}}}\ (n96-n87){{\ghostlmd^{2}}}\ (n97-n86){{\ghostvar^{2}}}\ (n98-n85){\lambda }\ (n99-n22){z}\ (n100-n21){{\ghostlmd^{2}}}\ (n101-n20){{\ghostvar^{2}}}\ (n102-n19){\lambda }\ (n103-n10){x}\ (n104-n9){{\ghostlmd^{3}}}\ (n105-n8){{\ghostvar^{2}}}\ (n106-n7){{\ghostlmd^{2}}}\ (n107-n6){{\ghostvar^{4}}}\ (n108-n5){{\ghostlmd^{4}}}\ (n109-n4){{\ghostvar^{5}}}\ (n110-n3){{\ghostlmd^{5}}}\ (n111-n2){{\ghostvar^{3}}}\ (n112-n1){{\ghostlmd^{3}}}\ (n113-n0){{\ghostvar^{2}}}\ (n114-n113){{\ghostlmd^{1}}}\ (n115-n112){{\ghostvar^{1}}}\ (n116-n111){{\ghostlmd^{1}}}\ (n117-n110){{\ghostvar^{1}}}\ (n118-n109){{\ghostlmd^{1}}}\ (n119-n108){{\ghostvar^{1}}}\ (n120-n107){{\ghostlmd^{1}}}\ (n121-n106){{\ghostvar^{1}}}\ (n122-n105){{\ghostlmd^{1}}}\ (n123-n104){{\ghostvar^{1}}}\ (n124-n103){\lambda }\ (n125-n18){s}\ (n126-n17){{\ghostlmd^{1}}}\ (n127-n16){{\ghostvar^{1}}}\ (n128-n15){{\ghostlmd^{1}}}\ (n129-n14){{\ghostvar^{2}}}\ (n130-n13){{\ghostlmd^{2}}}\ (n131-n12){{\ghostvar^{2}}}\ (n132-n11){{\ghostlmd^{2}}}\ (n133-n10){{\ghostvar^{4}}}\ (n134-n9){{\ghostlmd^{4}}}\ (n135-n8){{\ghostvar^{3}}}\ (n136-n7){{\ghostlmd^{3}}}\ (n137-n6){{\ghostvar^{5}}}\ (n138-n5){{\ghostlmd^{5}}}\ (n139-n4){{\ghostvar^{6}}}\ (n140-n3){{\ghostlmd^{6}}}\ (n141-n2){{\ghostvar^{4}}}\ (n142-n1){{\ghostlmd^{4}}}\ (n143-n0){{\ghostvar^{3}}}}$}

\resizebox{1\hsize}{!}{$t_{122} =
\Pstr[0.7cm]{(n0){\lambda }\ (n1){@}\ (n2-n1){\lambda t}\ (n3-n2){t}\ (n4-n1){\lambda s_2 z_2}\ (n5-n4){s_2}\ (n6-n3){\lambda n a x}\ (n7-n6){n}\ (n8-n5){\lambda }\ (n9-n4){s_2}\ (n10-n3){\lambda n a x}\ (n11-n10){n}\ (n12-n9){\lambda }\ (n13-n4){z_2}\ (n14-n3){\lambda a}\ (n15-n14){a}\ (n16-n13){{\ghostlmd^{1}}}\ (n17-n12){{\ghostvar^{1}}}\ (n18-n11){\lambda s z}\ (n19-n10){a}\ (n20-n9){{\ghostlmd^{2}}}\ (n21-n8){{\ghostvar^{1}}}\ (n22-n7){\lambda s z}\ (n23-n6){a}\ (n24-n5){{\ghostlmd^{2}}}\ (n25-n4){{\ghostvar^{3}}}\ (n26-n3){\lambda z_0}\ (n27-n26){z_0}\ (n28-n25){{\ghostlmd^{1}}}\ (n29-n24){{\ghostvar^{1}}}\ (n30-n23){\lambda }\ (n31-n22){s}\ (n32-n21){{\ghostlmd^{1}}}\ (n33-n20){{\ghostvar^{1}}}\ (n34-n19){\lambda }\ (n35-n18){s}\ (n36-n17){{\ghostlmd^{1}}}\ (n37-n16){{\ghostvar^{1}}}\ (n38-n15){{\ghostlmd^{1}}}\ (n39-n14){{\ghostvar^{2}}}\ (n40-n13){{\ghostlmd^{2}}}\ (n41-n12){{\ghostvar^{2}}}\ (n42-n11){{\ghostlmd^{2}}}\ (n43-n10){{\ghostvar^{4}}}\ (n44-n9){{\ghostlmd^{4}}}\ (n45-n8){{\ghostvar^{3}}}\ (n46-n7){{\ghostlmd^{3}}}\ (n47-n6){{\ghostvar^{5}}}\ (n48-n5){{\ghostlmd^{5}}}\ (n49-n4){{\ghostvar^{6}}}\ (n50-n3){{\ghostlmd^{6}}}\ (n51-n2){{\ghostvar^{4}}}\ (n52-n1){{\ghostlmd^{4}}}\ (n53-n0){{\ghostvar^{3}}}\ (n54-n53){{\ghostlmd^{1}}}\ (n55-n52){{\ghostvar^{1}}}\ (n56-n51){{\ghostlmd^{1}}}\ (n57-n50){{\ghostvar^{1}}}\ (n58-n49){{\ghostlmd^{1}}}\ (n59-n48){{\ghostvar^{1}}}\ (n60-n47){{\ghostlmd^{1}}}\ (n61-n46){{\ghostvar^{1}}}\ (n62-n45){{\ghostlmd^{1}}}\ (n63-n44){{\ghostvar^{1}}}\ (n64-n43){{\ghostlmd^{1}}}\ (n65-n42){{\ghostvar^{1}}}\ (n66-n41){{\ghostlmd^{1}}}\ (n67-n40){{\ghostvar^{1}}}\ (n68-n39){{\ghostlmd^{1}}}\ (n69-n38){{\ghostvar^{1}}}\ (n70-n37){{\ghostlmd^{1}}}\ (n71-n36){{\ghostvar^{1}}}\ (n72-n35){{\ghostlmd^{1}}}\ (n73-n34){{\ghostvar^{1}}}\ (n74-n33){{\ghostlmd^{1}}}\ (n75-n32){{\ghostvar^{1}}}\ (n76-n31){{\ghostlmd^{1}}}\ (n77-n30){{\ghostvar^{1}}}\ (n78-n29){{\ghostlmd^{1}}}\ (n79-n28){{\ghostvar^{1}}}\ (n80-n27){{\ghostlmd^{1}}}\ (n81-n26){{\ghostvar^{2}}}\ (n82-n25){{\ghostlmd^{2}}}\ (n83-n24){{\ghostvar^{2}}}\ (n84-n23){\lambda }\ (n85-n6){x}\ (n86-n5){{\ghostlmd^{3}}}\ (n87-n4){{\ghostvar^{4}}}\ (n88-n3){{\ghostlmd^{4}}}\ (n89-n2){{\ghostvar^{2}}}\ (n90-n1){{\ghostlmd^{2}}}\ (n91-n0){{\ghostvar^{1}}}\ (n92-n91){{\ghostlmd^{2}}}\ (n93-n90){{\ghostvar^{2}}}\ (n94-n89){{\ghostlmd^{2}}}\ (n95-n88){{\ghostvar^{2}}}\ (n96-n87){{\ghostlmd^{2}}}\ (n97-n86){{\ghostvar^{2}}}\ (n98-n85){\lambda }\ (n99-n22){z}\ (n100-n21){{\ghostlmd^{2}}}\ (n101-n20){{\ghostvar^{2}}}\ (n102-n19){\lambda }\ (n103-n10){x}\ (n104-n9){{\ghostlmd^{3}}}\ (n105-n8){{\ghostvar^{2}}}\ (n106-n7){{\ghostlmd^{2}}}\ (n107-n6){{\ghostvar^{4}}}\ (n108-n5){{\ghostlmd^{4}}}\ (n109-n4){{\ghostvar^{5}}}\ (n110-n3){{\ghostlmd^{5}}}\ (n111-n2){{\ghostvar^{3}}}\ (n112-n1){{\ghostlmd^{3}}}\ (n113-n0){{\ghostvar^{2}}}\ (n114-n113){{\ghostlmd^{2}}}\ (n115-n112){{\ghostvar^{2}}}\ (n116-n111){{\ghostlmd^{2}}}\ (n117-n110){{\ghostvar^{2}}}\ (n118-n109){{\ghostlmd^{2}}}\ (n119-n108){{\ghostvar^{2}}}\ (n120-n107){{\ghostlmd^{2}}}\ (n121-n106){{\ghostvar^{2}}}\ (n122-n105){{\ghostlmd^{2}}}\ (n123-n104){{\ghostvar^{2}}}\ (n124-n103){\lambda }\ (n125-n18){z}\ (n126-n17){{\ghostlmd^{2}}}\ (n127-n16){{\ghostvar^{2}}}\ (n128-n15){{\ghostlmd^{2}}}\ (n129-n14){{\ghostvar^{3}}}\ (n130-n13){{\ghostlmd^{3}}}\ (n131-n12){{\ghostvar^{3}}}\ (n132-n11){{\ghostlmd^{3}}}\ (n133-n10){{\ghostvar^{5}}}\ (n134-n9){{\ghostlmd^{5}}}\ (n135-n8){{\ghostvar^{4}}}\ (n136-n7){{\ghostlmd^{4}}}\ (n137-n6){{\ghostvar^{6}}}\ (n138-n5){{\ghostlmd^{6}}}\ (n139-n4){{\ghostvar^{7}}}\ (n140-n3){{\ghostlmd^{7}}}\ (n141-n2){{\ghostvar^{5}}}\ (n142-n1){{\ghostlmd^{5}}}\ (n143-n0){{\ghostvar^{4}}}}$}
\caption{Maximal traversals of $varity\ 2$ from Example~\ref{example:varity}}
\label{tab:varity2_trav}
\end{table}
\end{landscape}
\restoregeometry

\ifdefined\placeproofsatend
\section{Proofs}
\printproofs
\fi
\end{appendices}

\end{document}